%% file: paper.tex
\documentclass[runningheads]{llncs}
\usepackage[normalem]{ulem} 
\AtBeginDocument{
  }
\usepackage{algorithm}            
\usepackage{algpseudocode}        
\usepackage[table]{xcolor}
\usepackage{pifont}
\usepackage{listings}             
\usepackage{tabularx}
\usepackage{booktabs}             
\usepackage{graphicx}
\usepackage{subcaption}
\usepackage{mdframed}
\usepackage{stmaryrd}
\usepackage{multirow}
\usepackage{bussproofs}
\usepackage[inline]{enumitem}
\usepackage[nounderscore]{syntax} 
\usepackage{amsmath}
\usepackage{amssymb}
\usepackage{hyperref}
\usepackage{tikz}
\usetikzlibrary{shapes,positioning,calc,shadows}
\usetikzlibrary{decorations.pathmorphing,decorations.pathreplacing,tikzmark}
\usetikzlibrary{trees,arrows.meta}
\newcolumntype{L}{>{\raggedright\arraybackslash}X}
\newcolumntype{R}{>{\raggedleft\arraybackslash}X}
\newcolumntype{C}{>{\centering\arraybackslash}X}
\usepackage{custom}
\usepackage{mathpartir}
\usepackage{pgfplots}
\pgfplotsset{compat=1.18}
\usepackage{pgfplotstable}

\title{Inferring Empirical Sound Resource Bounds via Symbolic Execution and Linear Programming (Extended Version)\thanks{This is an extended version of the paper accepted at SAS 2026.}}
\titlerunning{Inferring Empirical Sound Resource Bounds (Extended Version)}

\author{Samuel Frontull\inst{1}\orcidID{0009-0004-1230-4666} \and
  Manuel Meitinger\inst{1}\orcidID{0009-0000-6555-0060} \and
  Georg Moser\inst{1}\orcidID{0000-0001-9240-6128}}
\authorrunning{S.~Frontull et al.}

\institute{Department of Computer Science, University of Innsbruck, Austria
  \email{\{samuel.frontull,georg.moser\}@uibk.ac.at},
  \email{manuel.meitinger@student.uibk.ac.at}}

\begin{document}

\maketitle

\begin{abstract}
	Existing approaches to resource analysis of programs can be classified into two main paradigms:
	static analysis and dynamic analysis methods.
	The former allow for formal guarantees but are inherently incomplete;
	the latter are widely applicable
	but may miss rare but characteristic (worst-case) scenarios and thus lack soundness.
	Hybrid approaches attempt to combine the strengths of both paradigms,
	thereby enabling the analysis of programs that are either too complex for purely static techniques
	or where dynamic approaches suffer from combinatorial explosion.
	In this paper, we present a novel hybrid approach that systematically derives
	upper bounds for the worst-case resource consumption of functional programs.
	Our method combines dynamic symbolic execution to exhaustively explore all possible computation paths 
	within a constrained input space with mixed-integer linear programming to derive empirically sound upper bounds.
	We have implemented the methodology in a prototype tool, dubbed~\compas{}, 
	which we made available on Zenodo\footnote{\url{https://doi.org/10.5281/zenodo.21323697}}.
\end{abstract}


\begin{center}
	\includegraphics[scale=0.16]{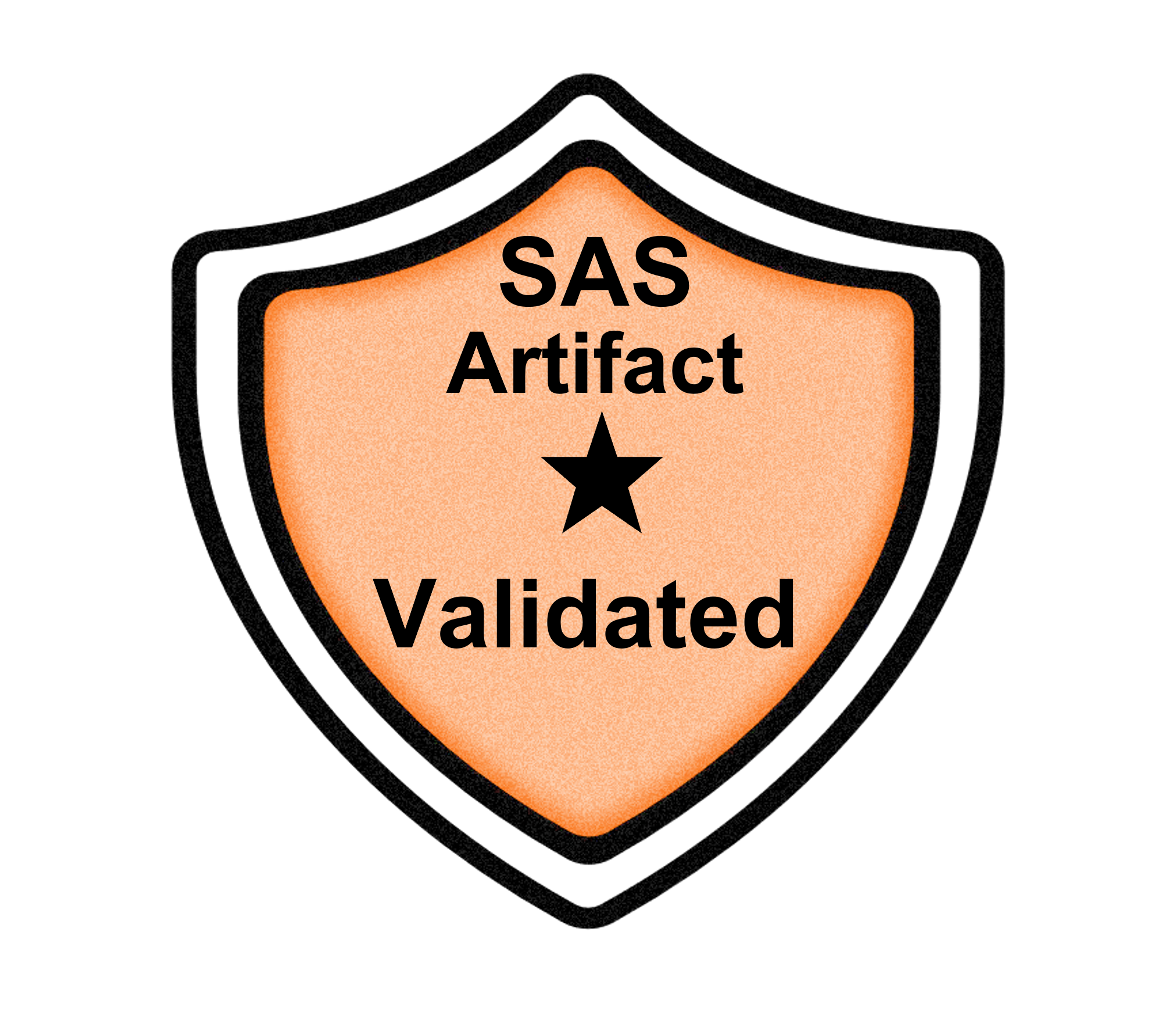} \includegraphics[scale=0.16]{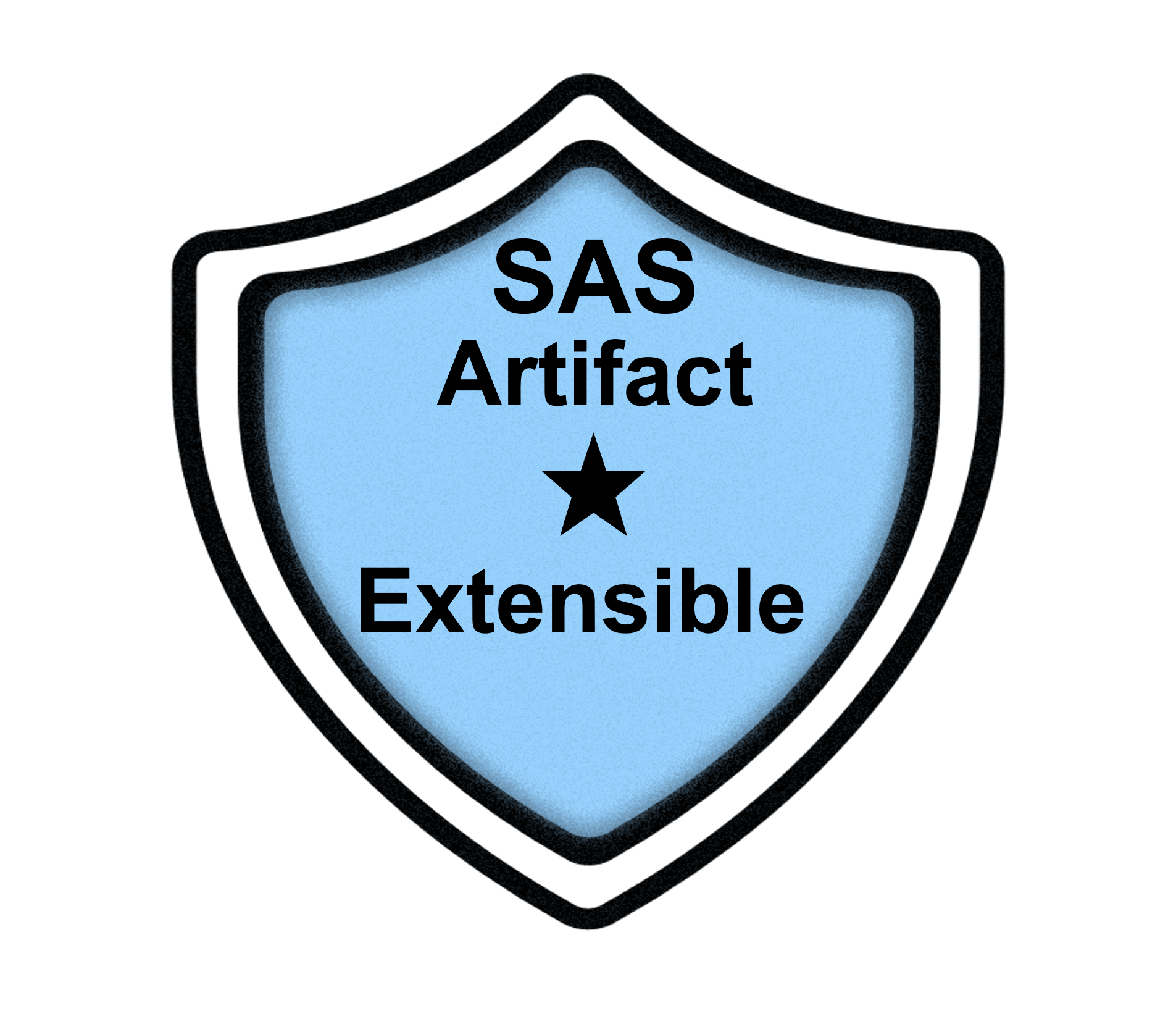} \includegraphics[scale=0.16]{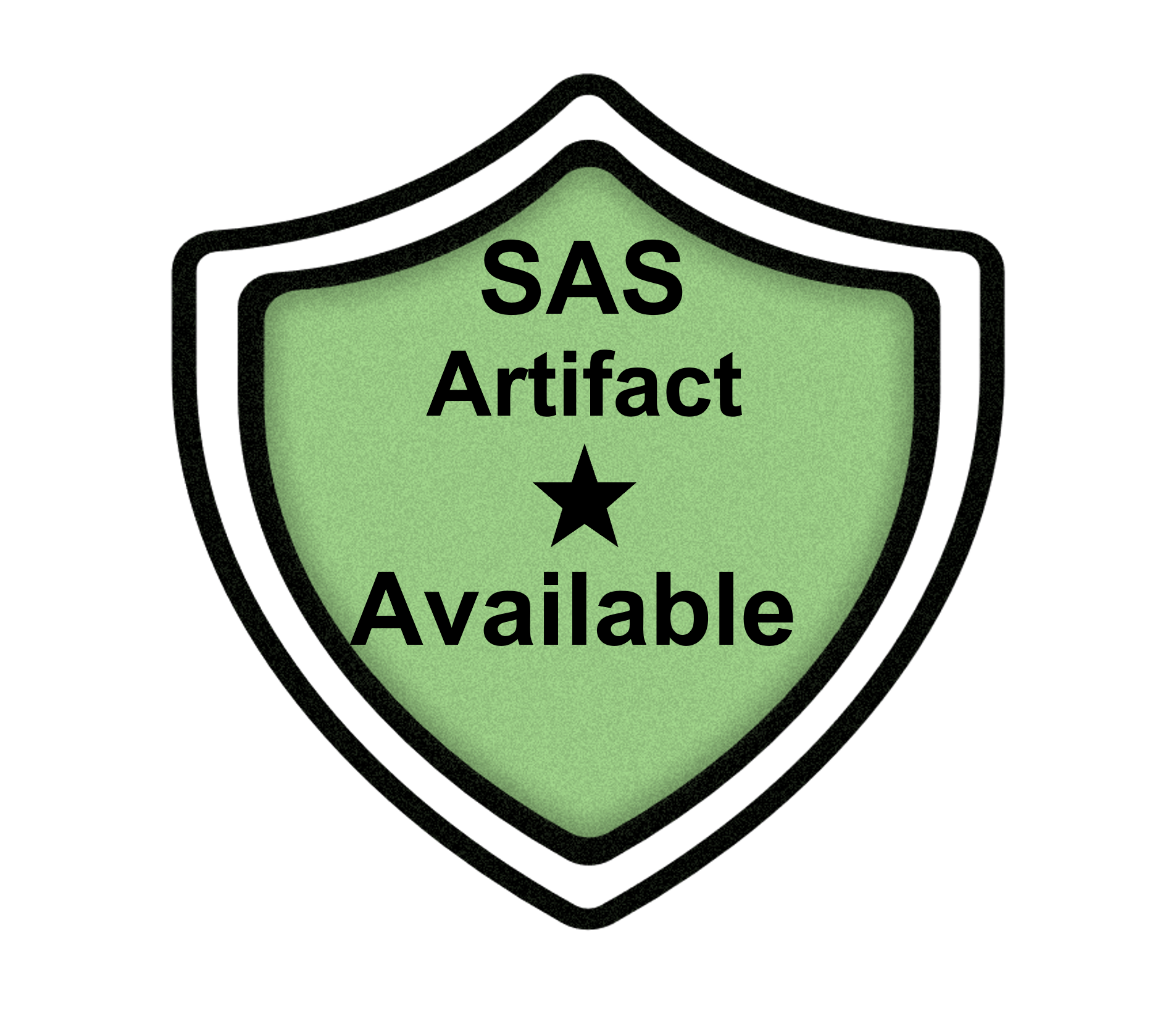}
\end{center}

\section{Introduction}\label{section:introduction}
\begin{table}[t]
	\centering
	\scriptsize
	\caption{Overview of soundness results of our method for the $10$ programs where \texttt{RaML} reported infeasibility (\infeasible) or timeout (\timeout{})
		with asymptotic bounds (where available) and timings of the hybrid approach as reported in~\cite{Pham:etal:2025}.}
	\label{tab:experiments1}
	\rowcolors{1}{gray!10}{white}
	\begin{tabularx}{\textwidth}{l c  @{\hspace{4pt}}|@{\hspace{4pt}} l r  @{\hspace{4pt}}|@{\hspace{4pt}} >{\raggedright\arraybackslash}X r c}
		\toprule\rowcolor{white}
		\textbf{\textsf{Program}} & \textbf{\textsf{RaML}}~\cite{Hoffmann:etal:2017} & \textbf{\textsf{Hybrid}}~\cite{Pham:etal:2025} & \textbf{\textsf{Time [s]}} & \textbf{\textsf{Our Method}}      & \textbf{\textsf{Time [s]}} & \textbf{\textsf{Sound}} \\
		\midrule
		\textsf{AVLTree}          & \infeasible                                      & $O((M + L) \log(M))$                           & 12.9                       & $3 + 0.64M^3 + 1.36M^2 + ML + 2L$ & 136.28                     & \sound{}                \\
		\textsf{BubbleSort2}      & \infeasible                                      & $O(M^2)$                                       & 10.1                       & $2 + 0.57M + 1.14M^2$             & 122.52                     & \sound{}                \\
		\textsf{DivBySub}         & \infeasible                                      & -                                              &                            & $1 + 2M$                          & 41.58                      & \sound{}                \\
		\textsf{FibMemo}          & \infeasible                                      & -                                              &                            & $2 + 11.95M + 1.35M^2$            & 2.26                       & \sound{}                \\
		\textsf{FoldSum}          & \infeasible                                      & -                                              &                            & $3 + 4M$                          & 4.63                       & \sound{}                \\
		\textsf{Id}               & \infeasible                                      & -                                              &                            & $2 + 2M$                          & 4.51                       & \sound{}                \\
		\textsf{MergeSortDC}      & \infeasible                                      & -                                              &                            & $2 + 13.57M \log(M)$              & 44.84                      & \sound{}                \\
		\textsf{RedBlackTree}     & \timeout{60s}                                    & $O((M + L) \log(M))$                           & 1030.9                     & $3 + M^2 + ML + 2M + 2L$          & 144.11                     & \sound{}                \\
		\textsf{RevDL}            & \infeasible                                      & -                                              &                            & $1 + M$                           & 4.46                       & \sound{}                \\
		\textsf{Round}            & \infeasible                                      & -                                              &                            & $3 + 3.25M + 0.75M^2$             & 18.51                      & \sound{}                \\
		\bottomrule
	\end{tabularx}
\end{table}

The analysis of programs with respect to their resource consumption is a fundamental problem in computer science.
Understanding these properties is crucial for selecting appropriate algorithms and data structures, optimising software, and ensuring the efficiency and scalability of applications.
Beyond this, they are also relevant for security-critical aspects, such as preventing timing attacks~\cite{Kocher:1996},
bug finding~\cite{Cadar:etal:2008,Godefroid:etal:2008} or test generation~\cite{Cadar:etal:2008,Tillmann:deHalleux:2008}.
For these reasons, resource analysis has a long-standing tradition~\cite{Wegbreit:1975}.
A central aspect in this context is the \emph{worst-case} runtime complexity of a program,
which characterizes the worst-case execution cost as a function of relevant input properties.

Approaches to automatically determining such worst-case runtime bounds
can be classified into two categories: \emph{static} and \emph{dynamic} methods.
Static methods derive cost bounds by analysing the source of a program,
dynamic methods by observing program behaviour.
Both paradigms have given rise to a variety of techniques.
Without aiming for completeness, we mention static techniques such as abstract interpretation~\cite{Cousot:Cousot:1977}, predicate transformers~\cite{AMS20}, program logics~\cite{Gulwani:2009}, recurrence solving~\cite{Wegbreit:1975} or type-based resource analysis~\cite{Hoffmann:etal:2017}.
Common approaches for dynamic analysis include profiling-based analysis~\cite{Graham:etal:1982},
empirical performance testing~\cite{Raj:1991,Kalibera:Jones:2013},
and genetic algorithm-based input generation~\cite{Bouchachia:2007,Rodrigues:etal:2018}.
While these techniques differ in their mechanisms, they share common trade-offs.
Static methods can provide formal correctness guarantees.
However, they are inherently incomplete and may fail (for sufficiently complex programs).
In contrast, dynamic methods are more flexible and can handle a broader range of programs, but face a different challenge: 
exhaustively exploring all possible executions is generally infeasible due to path explosion, 
preventing formal correctness guarantees for derived bounds.

To address these trade-offs, hybrid approaches~\cite{Wei:2018,Pham:2024,Pham:etal:2025} try to combine the strengths of static and dynamic analysis:
they aim to analyse programs that are beyond the practical reach of purely static methods,
while avoiding the scalability issues of purely dynamic techniques.
In this work, we propose a hybrid method for general resource analysis.
We model resource consumption by means of ticking, i.e. operations of interest in the source code are explicitly marked by \emph{tick} statements.
This yields a flexible, non-monotonic cost model that can represent a wide range of resources, such as runtime operations, memory allocations, and similar quantities.
By applying symbolic execution over an input specification, whose template constraints bound generated structures and whose global assertions restrict admissible inputs, we systematically explore all feasible executions of a program within the constrained input space.
Based on the ticks observed for these executions,
we employ basic statistical methods to determine an empirically sound upper bound to resource consumption, 
i.e. a bound that is guaranteed to never underestimate the costs for the input space covered by the empirical observations.

Our approach offers several advantages.
First, symbolic execution allows to explore complex, data-dependent control flow that is challenging for purely static analysis.
Second, by restricting the search space to a finite domain,
we address the problem of path explosion.
Third, by exploiting the program structure and reasoning symbolically about inputs, 
we avoid redundant exploration and thus improve analysis efficiency.
The bound synthesis based on mixed-integer linear programming (MILP)
that encodes upper bounding constraints
provides a methodological approach that enables the derivation of empirically sound upper bounds.
We focus on (higher-order) \emph{functional programs}, as they allow pure reasoning about resource consumption
without the complications of side-effects or pointers typical in imperative languages.
Nevertheless, our method is language-agnostic in principle, relying on an intermediate abstract syntax tree representation
that captures the essential program structure.

To evaluate our approach, we use benchmarks from two established sources:
the Resource Aware ML (\texttt{RaML}) benchmark suite~\cite{Hoffmann:etal:2017}, comprising $39$ programs,
and the pure OCaml programs from the recent hybrid resource-decomposition approach by Pham et al.~\cite{Pham:etal:2025}.
These benchmarks also include structurally non-trivial, data-dependent control flow for which existing static analyses in some cases fail or infer imprecise bounds.
They therefore serve as stress tests for the underlying reasoning principle and whether it can improve practical worst-case resource analysis.
Table~\ref{tab:experiments1} focuses on the $10$ \texttt{RaML} benchmarks for which \texttt{RaML} reports infeasibility or timeout.
Our method derives empirically sound concrete upper bounds for all $10$ programs.
For $3$ of these (\program{AVLTree}, \program{BubbleSort2}, \program{RedBlackTree}), the hybrid approach by Pham et al.~\cite{Pham:etal:2025} infers correct asymptotic bounds, but without concrete coefficients.

In summary, we make the following contributions:
\begin{itemize}
	\item a symbolic execution-based approach that exhaustively explores a finite input search space;
	\item a formulation of the upper bound problem as a linear programming (LP) task;
	\item an implementation of our approach in the prototype tool \compas{};
	\item an evaluation on a suite of $46$ benchmark programs taken from the literature.
\end{itemize}

We present the core methodology, formal foundations, and experimental evaluation.
For a complete presentation including the symbolic execution rules, detailed proofs, implementation details, and additional experiments, see the extended version~\cite{Frontull:etal:2026:extended}.
The prototype implementation and artifact are archived on Zenodo~\cite{frontull_2026_21323697}.

\paragraph{Outline.}
The paper is structured as follows.
Section~\ref{section:overview} provides an overview of our approach on a motivating example.
Section~\ref{section:symbolic-execution} describes our symbolic execution approach.
Section~\ref{section:linear-programming} details how we formulate the upper bound inference problem as an LP problem.
Section~\ref{section:soundness} establishes the empirical soundness of our approach by combining the correctness of the simulation with the soundness of the MILP formulation.
Section~\ref{section:experimental-evaluation} presents our experimental evaluation on the suite of benchmark programs
and briefly sketches our prototype implementation~\compas{}.
Section~\ref{section:related-work} discusses related work, and Section~\ref{section:conclusion} concludes the paper and outlines future work.

\section{Overview}
\label{section:overview}

\begin{figure}[t]
	\begin{minipage}{0.52\textwidth}
		\begin{lstlisting}[language=ml2]
  (*@$\costbox{dyade}{n,m}{n + n \cdot m}$@*)
  let rec dyade l1 l2 =
      match l1 with
      | [] -> []
      | x :: xs -> tick(1); 
          mult x l2 :: dyade xs l2
    \end{lstlisting}
	\end{minipage}
	\hfill
	\begin{minipage}{0.46\textwidth}
		\begin{lstlisting}[language=ml2]
    (*@$\costbox{mult}{n}{n}$@*)
    let rec mult n l =
        match l with
        | [] -> []
        | x :: xs -> tick(1); 
            n * x :: mult n xs
    \end{lstlisting}
	\end{minipage}
	\caption{Motivating program \program{Dyade} with upper bounds annotations.}\label{fig:dyade-code}
\end{figure}

Consider the example program \program{Dyade} shown in Figure~\ref{fig:dyade-code} taken from the
benchmark suite of \texttt{RaML}.%
\footnote{See~\url{https://www.raml.co/}.}
The main function
\raml{dyade} takes two lists, $l_1$ and $l_2$, consisting of integers and produces a list of lists of integers as output.
For each element $x \in l_1$, the function creates a copy of $l_2$ in which every element is multiplied by $x$.
To illustrate this, consider the following input lists:
\begin{center}
	$l_1 = [1; 2; 3]$ with length $|l_1| = 3$ \qquad
	$l_2 = [4; 5; 6; 7]$ with length $|l_2| = 4$
\end{center}

For these arguments, the execution of \raml{dyade l1 l2} performs three iterations of the outer pattern match.
In the first iteration, the pattern match on $l_1$ triggers \raml{tick(1)} and calls \raml{mult 1 [4; 5; 6; 7]}.
Inside \raml{mult}, the recursion processes each element of the list, resulting in four recursive calls,
each triggering \raml{tick(1)} and performing one multiplication.
This produces $[4; 5; 6; 7]$ at a cost of $1 + 4 = 5$ ticks.
The remaining iterations proceed analogously: the calls \raml{mult 2 [4; 5; 6; 7]} and \raml{mult 3 [4; 5; 6; 7]}
produce the lists $[8; 10; 12; 14]$ and $[12; 15; 18; 21]$, respectively, each incurring the same cost of $5$ ticks.
These intermediate results are combined into the final output, which is:
\[ [[4; 5; 6; 7]; [8; 10; 12; 14]; [12; 15; 18; 21]] \]
For this example, \raml{dyade l1  l2} evaluates with a total cost of $15$ ticks.
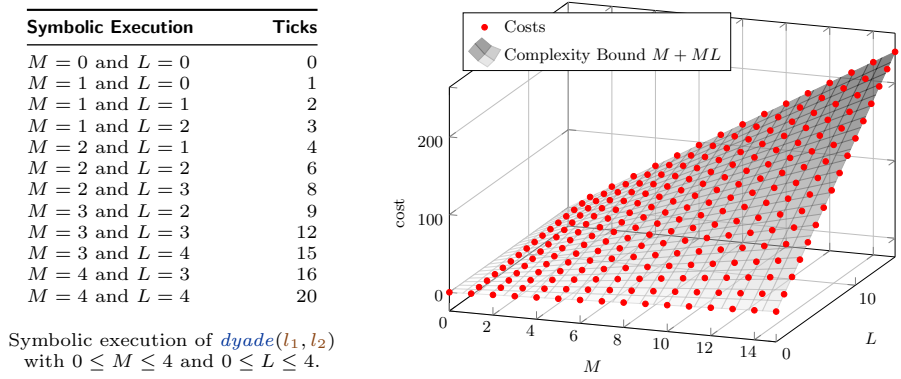
\begin{figure}[t]
	\begin{minipage}{0.45\textwidth}
		\centering
		\scriptsize
		\begin{tabularx}{0.7\linewidth}{@{}Xr@{}}
			\toprule
			\textbf{\textsf{Symbolic Execution}}        & \textbf{\textsf{Ticks}} \\
			\midrule
			$M = 0$ and $L = 0$ & $0$                     \\
			$M = 1$ and $L = 0$ & $1$                     \\
			$M = 1$ and $L = 1$ & $2$                     \\
			$M = 1$ and $L = 2$ & $3$                     \\
			$M = 2$ and $L = 1$ & $4$                     \\
			$M = 2$ and $L = 2$ & $6$                     \\
			$M = 2$ and $L = 3$ & $8$                     \\
			$M = 3$ and $L = 2$ & $9$                     \\
			$M = 3$ and $L = 3$ & $12$                    \\
			$M = 3$ and $L = 4$ & $15$                    \\
			$M = 4$ and $L = 3$ & $16$                    \\
			$M = 4$ and $L = 4$ & $20$                    \\
			                                            &                         \\
		\end{tabularx}

		Symbolic execution of $\IRfunName{dyade}(\IRvar{l_1}, \IRvar{l_2})$\\
		with $0 \leq M \leq 4$ and $0 \leq L \leq 4$.
	\end{minipage}
	\begin{minipage}{0.54\textwidth}
		\begin{tikzpicture}[scale=0.7]
			\begin{axis}[
					view={20}{22},
					xlabel={$M$},
					ylabel={$L$},
					zlabel={cost},
					grid=both,
					width=10cm,
					height=8cm,
					legend pos=north west,
					legend cell align=left
				]
				\input{include/dyade/data.tex}
				\addlegendentry{Costs}
				\addplot3[
					surf,
					domain=0:15,
					domain y=0:15,
					opacity=0.5,
					colormap={gray}{color(0cm)=(white); color(1cm)=(black)},
				]
				{x + x*y};
				\addlegendentry{Complexity Bound $M + ML$}
			\end{axis}
		\end{tikzpicture}
	\end{minipage}
	\caption{Left: Exemplary symbolic executions of \raml{dyade} with the observed worst-case cost. Right: Visualisation of costs for different input sizes. Each point $(M, L, c)$ represents the worst-case cost $c$ observed for inputs of size $M = |\IRvar{l_1}|$ and $L = |\IRvar{l_2}|$. The surface represents the inferred upper bound $M + ML$.}\label{fig:dyade-example}
\end{figure}
This concrete execution, however, is tied to one particular input pair and therefore yields only a single concrete observation:
for the lists $l_1$ and $l_2$ above, with lengths $3$ and $4$, respectively, the cost is $15$.

Repeating this experiment on such randomly sampled inputs would provide additional observations,
but unless all possible behaviours are covered, the worst-case input may remain unobserved.
Such observations cannot support general claims about worst-case behaviour
without understanding the internal mechanisms of the program.

To infer a resource bound, we have to consider all possible executions
and relate input properties with the observed costs.
This requires a systematic exploration of the input space, which we achieve by abstracting concrete inputs to relevant input properties.
For \raml{dyade}, the natural input properties are the lengths of the two lists
given as argument to the \raml{dyade} function.
We therefore seek a bound as a function of
$M = |l_1|$ and $L = |l_2|$.
This identification of the relevant input properties both bounds the search space 
and declares that the collected cost data should be indexed by the pair $(M,L)$.
In the illustrative run of Figure~\ref{fig:dyade-example}, we have constrained the input list lengths to the range
$0 \leq M \leq 4$ and $0 \leq L \leq 4$
and executed symbolically.

For the symbolic analysis of \program{Dyade}, the concrete inputs are replaced by \emph{list descriptors}.
These list descriptors (initially) represent all lists whose length lie within the specified bounds and whose elements satisfy the given element template.
In our example, the descriptor for $l_1$ therefore represents all integer lists of length between $0$ and $4$, and similarly for $l_2$.
During symbolic execution, these descriptors are refined only when the program branches on properties of these inputs.
For instance, when \raml{dyade} pattern matches on $l_1$, the analysis splits into the feasible cases induced by this test:
one path assumes that $l_1$ is empty, and another assumes that $l_1$ is non-empty.
On the non-empty path, the head is represented by a fresh symbolic integer, the tail is represented by a new list descriptor with the corresponding length information, and execution continues with the recursive call.
The calls to \raml{mult} refine the descriptor for $l_2$ in the same way.
At each branch point, SMT solving is used to check whether the path constraints remain satisfiable.%
\footnote{Satisfiability is checked via SMT, where
	we make use of Z3 (see~\url{https://www.microsoft.com/en-us/research/project/z3-3/}).}
When a symbolic path terminates, the executor serializes the accumulated cost together with the structural information stored in the descriptors.

This approach offers two fundamental advantages over purely data-driven black-box approaches:
\begin{enumerate*}[label=($\roman*$)]
	\item an entire equivalence class of concrete executions with identical behaviour is captured by a single symbolic execution;
	\item different executions are explored incrementally (without restarting from the initial state), avoiding redundant computations.
	\item only feasible paths are explored, and infeasible paths are discarded early.
\end{enumerate*}

Figure~\ref{fig:dyade-example} shows these grouped symbolic-execution results; no additional concrete executions are performed.%
\footnote{For identical input properties, only the maximum cost is retained.}
Since all feasible outcomes within the specified bounds are explored, every possible computation path in the bounded input space is represented by some terminating symbolic path.%
\footnote{The extended version~\cite{Frontull:etal:2026:extended} derives the first two observations step by step.}
A sufficiently large collection of systematic symbolic executions can reveal the worst-case complexity, which can be expressed as a function of the input properties.
To accomplish this, we formulate the problem as a linear programming (LP) problem, 
which considers multiple candidate templates (e.g., linear combinations of $1$, $M$, $L$, $ML$) and uses LP to choose coefficients upper-bounding all observations.
For \program{Dyade}, this recovers the exact bound $M + ML = |l_1| + |l_1| \cdot |l_2|$.

\section{Resource Analysis via Symbolic Execution}
\label{section:symbolic-execution}

Symbolic execution executes a program with symbolic values that abstract concrete ones by capturing selected properties.
This allows producing observations for resource-bound inference that are guaranteed to cover all executions within the input domain described by the templates and global assertions.
Since bounds are functions of selected properties, symbolic execution tracks these properties throughout.
Starting from an input specification, the executor replaces input arguments with descriptors generated from the template constraints and seeds the initial state with the global assertions.
Program operations refine descriptors: pattern matching splits states, conditionals add path constraints, and \raml{tick} commands update cost.
When a path terminates, the executor extracts input measures and cost.

We have implemented this for the core language in Figure~\ref{fig:IRsyntax}.
The core language serves as an intermediate representation for functional programs:
it is close to the RaML-style fragment used in our benchmarks, which simplifies comparison with static resource-analysis tools and the validation of inferred bounds, but it is not necessarily tied to it.
Other front ends can in principle translate source programs into the same representation.

Operating on this representation gives the executor a single source-level semantics for descriptors, path constraints, and explicit cost annotations.
The \raml{tick} construct is part of the analysed program and can express the cost model chosen by the user, such as recursive calls, comparisons, allocations, or memory accesses.
While established symbolic execution systems such as KLEE~\cite{Cadar:etal:2008},
SAGE~\cite{Godefroid:etal:2008}, Pex~\cite{Tillmann:deHalleux:2008},
Symbolic PathFinder~\cite{Noller:etal:2018,Luckow:etal:2020}, or binary-analysis frameworks such as angr~\cite{shoshitaishvili:2016} could in principle be adapted to collect similar measurements,%
\footnote{These systems are designed for test generation, bug finding, and coverage-oriented exploration, with interfaces geared toward test cases, path constraints, or low-level states. Recovering quantities such as list length, tree size, or nested structural summaries after compilation would require additional encodings and tool-specific processing. Our executor is therefore less general, but directly produces the data required by the subsequent bound inference.}
our approach requires a custom executor to maintain this source-level focus.

\begin{figure}[t]
	\small
	\begin{langsyntax}[\scriptsize]{Core Language Syntax}
		\begin{minipage}{0.4\textwidth}
			\begin{displaymath}
				\begin{array}{rcl}
					\tau & ::=  & x \mid n \mid b \mid \IRlist{}                                           \\
					     & \mid & \oplus_1(\tau_1) \mid \tau_1 \oplus_2 \tau_2                             \\
					     & \mid & \IRtrue \mid \IRfalse                                                    \\
					     & \mid & \IRifelse{\tau_{\text{cond}}}{\tau_{\text{if}}}{\tau_{\text{else}}}      \\
					     & \mid & \IRcons{\tau_{\text{head}}}{\tau_{\text{tail}}}                          \\
					     & \mid & \IRsplit{\tau_{\text{val}}}{\tau_{\text{nil}}}{h}{t}{\tau_{\text{cons}}} \\
					     & \mid & \IRpack{C}{[\vec{\tau}]}                                                 \\
					     & \mid & \IRunpack{\tau_{\text{val}}}{[(C_1,\vec{x}_1,\tau_1),                    \\
					     &      & \qquad\qquad\qquad\dots,(C_n,\vec{x}_n,\tau_n)]}                         \\
					     & \mid & \IRlet{x}{\tau_1}{\tau_{\text{ret}}}                                     \\
					     & \mid & \IRfun{[\vec{x}]}{\tau_{\text{body}}}                                    \\
					     & \mid & \IRcall{\tau_{f}}{\vec{\tau_{\text{args}}}}                              \\
					     & \mid & \IRtick{c}{\tau_{\text{ret}}}                                            \\
					     & \mid & \IRraise{m}                                                              \\
				\end{array}
			\end{displaymath}
		\end{minipage}
		\hfill
		\begin{minipage}{0.59\textwidth}
			\begin{displaymath}
				\begin{array}{rcll}
					\multicolumn{4}{l}{\text{Metavariables (syntax):}}                                           \\
					x, h, t        & \in & \mathcal{V}                          & \text{(Variables)}             \\
					n              & \in & \mathbb{Z}                           & \text{(Integers)}              \\
					c              & \in & \mathbb{R}                           & \text{(Costs)}                 \\
					m              & \in & \text{String}                        & \text{(Messages)}              \\
					C              & \in & \text{Constructors}                  & \text{(Tags)}                  \\
					\oplus_1       & \in & \{-, \sqrt{}, \sqrt[3]{} \}          & \text{(Unary Operations)}      \\
					\oplus_2       & \in & \{+, -, *, /, \% \} \cup             & \text{(Binary Operations)}     \\
					               & \in & \{<, \leq, =, \neq, \geq, >, \&, |\} & \text{(Boolean Operations)}    \\[1.5ex]
					\multicolumn{4}{l}{\text{Runtime values (execution):}}                                       \\
					\symbolic{v}   & \in & \text{Constants, Z3.Expr}            & \text{(Symbolic)}              \\
					\descriptor{D} & \in & \text{Descriptor}                    & \text{(Structural/Functional)}
				\end{array}
			\end{displaymath}
		\end{minipage}
	\end{langsyntax}
	\caption{Core language syntax and metavariables for both syntactic constructs and runtime values.}
	\label{fig:IRsyntax}
\end{figure}

The core language is operationally untyped:
terms do not carry type annotations, and symbolic execution does not perform type checking.
Types enter the analysis through the input specification.
They are meta-level information used to describe the input arguments
and to determine which kind of symbolic value or descriptor is constructed for each argument.
Symbolic execution operates on two distinct levels of abstraction.
At the \emph{specification level}, the user provides \emph{templates} that describe the shape and constraints of the input space to be explored.
At the \emph{runtime level}, the symbolic executor manipulates \emph{values} derived from these templates.

A \emph{template} $T = (\theta, \templateAssertionSet{})$ specifies the input domain and is defined as a pair consisting of a type $\theta$ and template constraints $\templateAssertionSet{}$.
A type $\theta$ is defined by the grammar
\[
	\theta ::= \DTInt \mid \DTBool \mid \DTList(\theta)
	\mid \DTTuple(C,\theta_1,\ldots,\theta_k)
	\mid \DTTree(L,N,\theta_1,\ldots,\theta_k).
\]
Here, $C$ denotes a tuple or algebraic-data constructor, and $L$ and $N$ denote the leaf and node constructors of a tree type.
The template constraints $\templateAssertionSet{}$ guide input generation: they may specify structural bounds (e.g., minimum and maximum list lengths or node counts), element templates, and optional local restrictions on generated values.
We view $\templateAssertionSet{}$ as a type-specific partial record: $\templateAssertionSet{}[\alpha]$ denotes the value of attribute $\alpha$ when it is present; omitted optional attributes, such as length bounds, are interpreted as unconstrained, whereas attributes required by the chosen template form must be present for the input specification to be well formed.
For instance, in the \program{Dyade} example, the template for the first input list can be written as
\[
	T_{l_1} = \bigl(\DTList(\DTInt),
	\{\mathit{minLen}=0,\mathit{maxLen}=4,\mathit{itemTpl}=(\DTInt,\varnothing)\}\bigr),
\]
where the item template $(\DTInt,\varnothing)$ allows creating unconstrained integer elements.

The type system is intentionally first-order: function types are excluded from input specifications.
Symbolic execution is therefore performed over first-order entry-function inputs: primitives and bounded algebraic data.
Higher-order functions are still supported inside the analysed program through closures and partial applications, but an unknown function is not currently treated as a symbolic input.\footnote{For instance, an entry point such as \raml{map f xs} can be analysed when the function argument is fixed by the analysed wrapper or is a program-defined closure; it is outside the current input model when that function argument itself should range over arbitrary functions.}
Comparable higher-order complexity analyses also make higher-order structure explicit before applying first-order complexity reasoning~\cite{Avanzini:etal:2015}.

During symbolic execution, templates are translated into \emph{runtime values}.
We distinguish two types of runtime values:
\begin{itemize}
	\item \emph{Symbolic Z3 expressions} (e.g., $\symbolic{x}$ representing an unknown integer or boolean)
	\item \emph{Structural descriptors} (list, tuple, or tree descriptors)
\end{itemize}
Primitive input types ($\DTInt$, $\DTBool$) generate symbolic Z3 values, whereas composite input types ($\DTList$, $\DTTuple$, $\DTTree$) generate structural descriptors.

\emph{Symbolic Z3 expressions} can represent fully known, partially known or unknown values.
For example, when a program branches on the first element of a list being greater than 5, the symbolic executor introduces a fresh symbolic variable $\symbolic{x}$ representing the unknown first element and records the path constraint $\symbolic{x} > 5$ in the path assertions for the branch where the condition holds.

\emph{Structural descriptors} are runtime values that track and constrain structural properties of input arguments
and control the exploration of the input space during symbolic execution.
Each descriptor maintains information about the constraints specified in its template (e.g., length bounds) 
and tracks via \emph{input pointers} which parts of the structure have been accessed during the current execution path.
An \emph{input pointer} is a canonical dot-separated identifier rooted at an input argument, for example a particular list element, tuple field, or tree node.

Our implementation supports three kinds of structural descriptors for symbolic inputs: \emph{list descriptors}, \emph{tuple descriptors}, and \emph{tree descriptors}.
In addition, it uses \emph{functional descriptors} at runtime to represent closures and partial applications.

\emph{List descriptors} represent lists with possibly unknown length, bounded by lower and upper length constraints derived from the template. 
A list descriptor is written as $\descriptiveList{o}{p}{m}{M}{T}$, where \(o\) denotes the (current) offset tracking how many elements have been accessed, \(p\) is the input pointer uniquely identifying this list, \(m\) and \(M\) denote the minimum and maximum length bounds from the template, and \(T\) is the element template specifying the type and constraints for list elements.
Thus, the \program{Dyade} template above is initially represented at runtime as $\descriptiveList{0}{\texttt{"l1"}}{0}{4}{T_{\mathit{item}}}$, where $T_{\mathit{item}}=(\DTInt,\varnothing)$.

\emph{Tuple descriptors} represent fixed-structure tuples where each position has a predetermined item type specified by the template. 
A tuple descriptor is written as $\descriptiveTupleS(p, n, C, [T_1,\dots,T_n])$, where $p$ is the input pointer, $n$ is the length (number of items), $C$ specifies the tuple's type name (algebraic data type constructor), and $[T_1,\dots,T_n]$ is the array of element templates. 
This allows tuples with heterogeneous element types (e.g., $[\DTInt, \DTBool, \DTList]$).

\emph{Tree descriptors} represent tree-structured data with recursive references, constrained by node bounds from the template.
A tree descriptor is written as $\descriptiveTreeS(s, r, m, M, L, N, [T_1,\dots,T_n])$, where $r$ is the root pointer, $s$ is the sub-pointer under $r$, $m$ and $M$ are the minimum and maximum node bounds from the template, $L$ and $N$ the leaf and node type names, and $[T_1,\dots,T_n]$ are the element templates.
A separate \texttt{TreeCache} tracks generated \texttt{nodeCount} and \texttt{leafCount} during execution.

\emph{Functional descriptors} encapsulate a function body, the captured environment at the moment the function was defined, and any already provided arguments if the function value was used in partial \textsc{Call}.
This allows for higher-order functions and currying.

\subsection{Dynamic Symbolic Execution}
\label{section:symbolic-execution:semantics}

Our implementation adopts \emph{dynamic symbolic execution}---aka \emph{concolic execution}---which allows the analysis to follow concrete execution paths while simultaneously maintaining symbolic constraints.
The \emph{symbolic execution} of a program $\mathcal{P}$
is initiated with an input specification
$\mathcal{I} = (\mathcal{P}, \IRfunName{main}, \argspecificationList{}, \initialAssertionSet{})$,
where $\mathcal{P}$ is the functional program, $\IRfunName{main}$ is the entry function to analyse,
$\argspecificationList{} = [\iota_1, \ldots, \iota_k]$ the specification of the input arguments, 
and $\initialAssertionSet{}$ is the set of initial global assertions.

Each argument specification $\iota_i = \argspecification{x_i}{\theta_i}{\templateAssertionSet{}_i}$ binds the parameter name $x_i$ to a template with argument type $\theta_i$ and template constraints $\templateAssertionSet{}_i$.
During initialisation, each specification is translated into a symbolic value or structural descriptor bound to its corresponding parameter.
We use the notation $\trcore{\mathcal{P}}$ to denote the representation of a functional program $\mathcal{P}$ in the core language.
The symbolic execution of $\trcore{\mathcal{P}}$
proceeds via a continuation-passing style interpreter that processes frames $\exframe{\tau}{ \sigma}{\kappa}$, where $\tau$ is a core language term, $\sigma$ is the current state, and $\kappa$ is the continuation that either schedules the next frame or forwards the result to the parent frame.
Each execution is characterised by a \emph{state} $\sigma = \exstate{\rho}{\assertionSet{}}{\mathcal{S}}{\mathcal{C}}$ consisting of:

\begin{itemize}
	\item an \emph{environment} $\rho: \text{Var} \to \mathcal{V}$ mapping variables to symbolic values,
	\item an \emph{assertion set} $\assertionSet{}$ containing the current path constraints as Z3 boolean expressions,\footnote{If execution starts with $\initialAssertionSet{}$, then every reached state has assertions $\assertionSet{} = \initialAssertionSet{} \cup \Delta$ for the branch constraints $\Delta$ accumulated along that path. Local value constraints generated from an argument template, such as integer minimum/maximum bounds or optional SMT assertions on a generated element, are stored with that value and merged into $\assertionSet{}$ when the value is brought into the execution state.}
	\item a \emph{structural cache} $\mathcal{S} = (\mathcal{S}_{\text{list}},\mathcal{S}_{\text{tuple}},\mathcal{S}_{\text{tree}})$ storing already generated list elements, tuple fields, and tree nodes by input pointer,
	\item a \emph{cost accumulator} $\mathcal{C} \in \mathbb{R}$ tracking resource consumption.
\end{itemize}

The structural cache makes descriptor-based execution consistent.
While a descriptor describes which parts of an input may still be generated, the cache records which parts have already been generated along the current execution path.
For lists, the cache component $\mathcal{S}_{\text{list}}$ maps input pointers to the symbolic elements already generated for the corresponding list and records whether the end of the list has been fixed.
This ensures that repeated accesses to the same list element reuse the same symbolic value.
The tuple and tree components of the cache play the analogous role for tuple fields and tree nodes.
The analysis starts with
the frame $\initialframe$,
with the starting state $\sigma_0 = \exstate{\rho_{\text{init}}}{\initialAssertionSet{}}{([],[],[])}{0}$
and the initial continuation $\idcontinuation{}$ is the identity $\lambda xyz.xyz$.
The initial environment $\rho_{\text{init}}$ is initialised with closures for all globally defined functions.
For this,
we need an auxiliary function $\mathit{extractFuns}$ that operates on a program $\mathcal{P}$ and
extracts all top-level function bindings, translating each to its core language representation:
\begin{align*}
	\mathit{extractFuns}(\raml{let}~f_1~\vec{x_1}~\raml{=}~e_1\raml{;;}~\dots~\raml{;;}~\raml{let}~f_n~\vec{x_n}~\raml{=}~e_n\raml{;;})
	 & = \bigcup_{i = 1}^n \{(f_i, \vec{x_i}, \trcore{e_i})\}
\end{align*}
$\rho_{\text{init}}$ is then constructed in two steps.
First, the global environment $\rho_{\text{glob}}$ is constructed%
\footnote{The self-reference to $\rho_{\text{glob}}$ in each closure enables recursive and mutually recursive function definitions, as all functions can reference each other through the shared global environment.}
by mapping each top-level function in $\mathcal{P}$ to a closure:
\[
	\rho_{\text{glob}} = \{ f \mapsto \closure{\vec{x}}{\tau_b}{\rho_{\text{glob}}} \mid (f, \vec{x}, \tau_b) \in \mathit{extractFuns}(\mathcal{P}) \}
\]
Second, the \emph{symbolic input descriptors} are constructed from the argument specifications.
We define the auxiliary function $\mathit{buildValue}$ that maps an argument specification $\iota = \argspecification{x}{\theta}{\templateAssertionSet{}}$ to the corresponding symbolic value or descriptor by pattern matching on the type $\theta$ and interpreting the template constraints $\templateAssertionSet{}$:
{\small
\[
	\buildValue{\iota} :=
	\begin{cases}
		\mathsf{Z3.IntConst}(x)                                                                                                                              & \text{if } \theta = \DTInt   \\
		\mathsf{Z3.BoolConst}(x)                                                                                                                             & \text{if } \theta = \DTBool  \\
		\descriptiveList{0}{x}{\templateAssertionSet{}[\mathit{minLen}]}{\templateAssertionSet{}[\mathit{maxLen}]}{\templateAssertionSet{}[\mathit{itemTpl}]} & \text{if } \theta = \DTList  \\
		\descriptiveTuple{x}{\templateAssertionSet{}[\mathit{len}]}{\templateAssertionSet{}[\mathit{type}]}{\templateAssertionSet{}[\mathit{items}]}         & \text{if } \theta = \DTTuple \\[1ex]
		\begin{aligned}
			\descriptiveTree{\mathit{subPointer}(x,\templateAssertionSet{}[\mathit{root}])}{\templateAssertionSet{}[\mathit{root}]}{\templateAssertionSet{}[\mathit{minNodes}]}
			{ \\\templateAssertionSet{}[\mathit{maxNodes}]}{\templateAssertionSet{}[\mathit{leaf}]}{\templateAssertionSet{}[\mathit{node}]}{\templateAssertionSet{}[\mathit{items}]}
		\end{aligned}
		                                                                                                                                                     & \text{if } \theta = \DTTree
	\end{cases}
\]
}
For primitive types (\DTInt, \DTBool), $\mathit{buildValue}$ creates symbolic Z3 expressions and attaches any local value constraints, such as integer minimum/maximum bounds or optional SMT assertions specified on that argument template. For composite types (\DTList, \DTTuple, \DTTree), it constructs structural descriptors with the appropriate parameters, such as length or node bounds.
These descriptor bounds are enforced lazily when execution inspects the descriptor.
For lists, $\templateAssertionSet{}[\mathit{itemTpl}]$ is the template passed recursively to $\mathit{buildValue}$ when a fresh element is forced, which also covers nested lists.
Generated values are inserted into the corresponding component of the structural cache, so later accesses to the same input pointer reuse the same symbolic value.
Similar mechanisms apply to tuple unpacking and tree traversal, where descriptors guide which structural patterns are explored and how deeply recursion proceeds. This descriptor-based approach enables our method to systematically enumerate all execution paths corresponding to inputs of varying shapes and sizes, up to the specified bounds.
The input argument environment is constructed by applying $\mathit{buildValue}$ to each of the $k$ argument specifications $\argspecificationList{} = [\iota_1, \ldots, \iota_k]$:
\[
	\rho_{\text{input}} = \bigcup_{i=1}^{k} \{ x_i \mapsto \buildValue{\iota_i} \}
	\quad \text{where } \iota_i = \argspecification{x_i}{\theta_i}{\templateAssertionSet{}_i} \in \argspecificationList{}
\]
Global assertions are separate from the template constraints. They are constructed by $\mathit{buildAssertions}$ on the top-level input specification, which allows the user to specify SMT-LIB conditions referencing arguments and their subcomponents through input pointers rooted at the argument names. The referenced pointers are declared as Z3 constants according to the template type found at that pointer and used in \texttt{Z3.ast\_from\_string}, turning the SMT-LIB text into internal Z3 expressions. Thus, a global assertion may constrain a component whose symbolic value has not yet been materialised: when symbolic execution later forces the corresponding descriptor, the generated value is created at the same input pointer. \initialAssertionSet{} is initialised as the set of these top-level conditions.%
\footnote{The implementation also supports SMT assertions attached to nested argument templates. In the presentation above, these are treated as local value constraints contributed by the corresponding template constraints $\templateAssertionSet{}$.}
The final initial environment $\rho_{\text{init}} = \rho_{\text{glob}} \cup \rho_{\text{input}}$ combines the global function environment with the input descriptors
where the union operator $\cup$ merges the two environments, giving precedence to $\rho_{\text{input}}$ in case of name conflicts (though in well-formed programs, function names and argument names are distinct).
Starting with the initial frame, the symbolic execution engine explores all feasible execution paths
$\Pi = \initialframe \leadsto^* \tau(\symbolic{v},\assertionSet{}, \mathcal{S}, \mathcal{C})$
by repeatedly applying the execution rules $\leadsto$ defined for each core language construct.
The symbolic executor systematically explores this state space, ensuring all reachable execution paths satisfying the template constraints and initial global assertions are explored.
Whenever there is a choice of continuation, we choose the identity $\idcontinuation{}$ for simplicity.
For the sake of brevity, we provide the full set of execution rules in the Appendix~\ref{appendix:execution-rules}.

\subsection{Extracting Relevant Information From Symbolic Executions}
\label{section:symbolic-execution:grouping}

Every successfully completed execution path $\pi \in \Pi$ terminates with a final frame
$(\symbolic{v},\assertionSet{}, \mathcal{S}, \mathcal{C})$
where $\symbolic{v}$ is the symbolic result value, $\assertionSet{}$ are the accumulated path constraints,
$\mathcal{S}$ is the structural state tracking input properties, and $\mathcal{C} \in \mathbb{N}$ is the accumulated cost.
We associate each path $\pi$ with a \emph{dimension vector}\
$x_\pi = (d_1,\ldots,d_m)$, which captures the relevant input size descriptors\
for the explored execution. The components of $x_\pi$ are extracted from the\
structural state $\mathcal{S}$ according to the input specification $\argspecificationList{}$. Specifically,\
for each input argument, we select the properties relevant for cost estimation,\
such as the number of accessed elements for list arguments\
($S_{\mathrm{list}}[l.\mathrm{length}]$), the number of visited nodes for tree\
arguments ($S_{\mathrm{tree}}[t.\mathrm{nodes}]$), and the concrete values of\
integer arguments using the respective input pointers. Thus, $x_\pi$ represents an abstraction of the input instance\
containing only the dimensions relevant to the inferred cost bound.
For arguments with nested structures (e.g., lists of lists or trees with lists), $\mathcal{S}$ contains separate input pointers for each sub-structure accessed during execution (e.g., $\texttt{a.1.length}$, $\texttt{a.2.length}$, etc.).
To obtain a fixed-length dimension vector, we aggregate these nested dimensions by computing the \emph{maximum} over all sub-structures of the same type.
This extraction process produces one primary component per top-level input argument plus additional secondary components for the maximum sizes of nested structures, when present.
Each feasible path $\pi$ yields an observation $D_\pi = (\vec{x}_\pi, \mathcal{C})$ pairing its dimension vector with the observed cost.
The complete set of observations is
\[
	\mathcal{D} = \{ (\vec{x}_\pi, \mathcal{C}) \mid \pi \in \Pi \}.
\]
Multiple paths may produce the same dimension vector $\vec{x}$ but differ in cost (e.g., due to different branches taken or data-dependent control flow).
For resource analysis, we require the \emph{worst-case cost} for each distinct input shape.
We therefore group observations by dimension vector and retain only the maximum cost for each unique $\vec{x}$:
\[
	\mathcal{D}_{\text{max}} = \bigl\{(\vec{x}, c_{\text{max}}) \mid \vec{x} \text{ occurs in } \mathcal{D} \text{ with } c_{\text{max}} = \max\{c \mid (\vec{x}, c) \in \mathcal{D}\}\bigr\}.
\]
This grouping ensures that $\mathcal{D}_{\text{max}}$ contains exactly one observation per unique input shape, associating each shape with its worst-case observed cost.
The set $\mathcal{D}_{\text{max}}$ is then used for bound inference.
\section{Bound Estimation via Mixed Integer Linear Programming}
\label{section:linear-programming}
As described in Section~\ref{section:symbolic-execution:grouping}, symbolic execution produces a set of worst-case executions $\mathcal{D}_{\text{max}}$ where each pair $(\vec{x}, c) \in \mathcal{D}_{\text{max}}$ associates a unique input shape $\vec{x}$ (dimension vector) with the maximum observed cost $c$ for that shape.
We introduce a sorted list $l_{\mathcal{D}_{\text{max}}}$ containing these observations in lexicographic order by input shape, which enables constraints on residuals (described below).
Let $N = |\mathcal{D}_{\text{max}}|$ denote the number of observations.
In the following, we formulate the bound inference problem as a mixed integer linear program.

\subsection{Problem Formulation}

We formulate the problem of inferring an upper bound from $\maxCosts{\mathcal{D}}$ as a constrained optimization problem $\mathit{InferBound}(\maxCosts{\mathcal{D}})$ that finds coefficients for a candidate bound \(f(\vec{x})\) such that it upper-bounds all observed costs.
Candidate bounds are represented as linear combinations of basis functions:
\[
	f(\mathbf{x}) = \sum_{i=0}^{n} a_i \phi_i(\mathbf{x}),
\]
where $\phi_i(\mathbf{x})$ are predefined basis functions (e.g., $1$, $x$, $x^2$, $x\log x$, $xy$) and $a_i$ are the coefficients to be inferred.
We define constraints for:
\begin{enumerate*}[label=($\roman*$)]
	\item non-negative coefficients to ensure the bound is non-decreasing,
	\item upper bounding the costs, and
	\item residuals with respect to $l_{\maxCosts{\mathcal{D}}}$.
\end{enumerate*}
\paragraph{Non-negative Coefficients}
To ensure non-negativity of the bound coefficients, we require
$
	a_i \geq 0 \text{ for all } i \in \{0, 1, \ldots, n\}
$.
This ensures that the bound is monotonically non-decreasing as input size increases.
While restrictive for some programs, this assumption is reasonable for many algorithms when considering worst-case scenarios.

\paragraph{Upper Bounding Constraints}
For each observation \((\vec{x}_i, c_i)\), we require:
\[f(\vec{x}_i) \geq c_i \quad \text{for all } i \in \{0, 1, \ldots, N\}.\]
For the minimization formulation (standard form for linear programming), this is equivalently written as $-f(\vec{x}_i) \leq -c_i$.
\begin{figure}[t]
	\centering
	\begin{subfigure}[b]{0.46\textwidth}
		\centering
		\input{include/residuals/1.tex}
		\subcaption{Upper Bound without Constraints on Residuals.}\label{fig:residuals-1}
	\end{subfigure}
	\hfill
	\begin{subfigure}[b]{0.46\textwidth}
		\centering
		\input{include/residuals/2.tex}
		\subcaption{Upper Bound with Constraints on Residuals.}\label{fig:residuals-2}
	\end{subfigure}
	\caption{Effect of Constraints on Residuals. The red line represents the upper bound, while the black dots indicate the observed costs. In (a), without constraints on residuals, the bound deviates from the data before aligning with the final observation (residuals at some point decrease with increasing input size). In (b), with constraints on residuals, the bound is forced to stay closer to the costs, resulting in a tighter fit and improved generalisation.}
	\label{fig:residuals}
\end{figure}
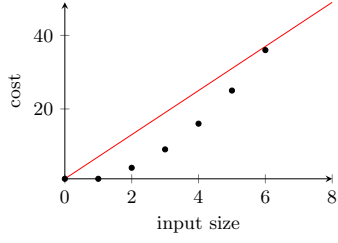
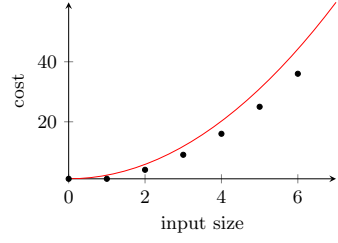

\paragraph{Constraints on Residuals}
To encourage the bound to diverge from the costs as the input size increases, we introduce constraints on the residuals.
This serves to promote better generalisation.
A motivation for this is shown in Figure~\ref{fig:residuals}, where two valid upper bounds are depicted: one without constraints on the residuals (Figure~\ref{fig:residuals-1}) and one with such constraints (Figure~\ref{fig:residuals-2}).
The \emph{residual} for observation $(\vec{x}_i, c_i)$ is defined as $r_i = f(\vec{x}_i) - c_i$, representing how much the bound exceeds the cost. We require that residuals do not decrease as input size increases (encouraging the bound to diverge from observations), unless the bound exactly matches an observation. This can be expressed as:
\[f(\vec{x}_{j}) - f(\vec{x}_{i}) \geq c_{j} - c_i \quad \text{or} \quad f(\vec{x}_{j}) = c_{j} \quad \text{for all } i < j\]
This requires the introduction of binary variables to distinguish between the two cases, leading to a MILP.
We use the big-M formulation to encode this condition:

\begin{align}
	 & (f(\vec{x}_{j}) - c_{j}) - (f(\vec{x}_{i}) - c_i) \geq -M \, b_i,                                                                   \\
	 & f(\vec{x}_{j}) - c_{j} \leq M \, (1 - b_i),                                                                                         \\
	 & \text{for all } (\vec{x}_i, c_i), (\vec{x}_j, c_j) \in l_{\maxCosts{\mathcal{D}}} \text{ with } i < j \text{ and } b_i \in \{0,1\},
\end{align}

Here, each $b_i$ is a binary variable and $M$ is a ``sufficiently large constant'' which we set to the maximum observed cost.%
\footnote{The choice of \(M\) is crucial; it must be large enough to not restrict the feasible region but not so large as to cause numerical instability in the solver.}
When $b_i = 0$, the second constraint is inactive, and the first enforces non-decreasing residuals.
When $b_i = 1$, the first constraint is inactive, and the second enforces $f(\vec{x}_{j}) \leq c_{j}$, which combined with the upper-bounding constraint $f(\vec{x}_{j}) \geq c_{j}$ yields exact matching.
This formulation ensures that the bound either maintains or increases its deviation from the observed costs as the input size increases, unless it exactly matches an observed cost.
The downside is that these constraints transform the linear program into a mixed-integer linear program, which is NP-hard in general~\cite{Karp:1972}. Since we add \(N-1\) binary variables, the search space grows exponentially with the number of observations.
We did not encounter performance issues in our experiments due to the relatively small number of executions, but this could become a bottleneck for larger datasets.

\subsection{Candidate Bounds}

There may be multiple bounds that satisfy the constraints outlined above.
Depending on the problem, a linear, quadratic, or even cubic bound may be required to meet the observations under the given restrictions. However, since only linear combinations of decision variables are allowed, we construct the bound as a combination of predefined basis functions (constant, identity, square, cube, log, etc. and combinations of these). These basis functions serve as building blocks that the solver can combine to form bounds of arbitrary degree.
To enable this, the corresponding basis function values are precomputed and provided to the solver in the form of a matrix. Each column of the matrix represents a basis function, and each row corresponds to an observation. The solver then determines the optimal coefficients for each basis function (some of which may be zero if the corresponding component is not relevant).
For example, to test whether the costs can be described by a bound of the form
$
	f(x, y) = c + x + x y + y^2,
$
the matrix $ [[1, x_1, x_1 y_1, y_1^2], [1, x_2, x_2 y_2, y_2^2], \dots, [1, x_n, x_n y_n, y_n^2]]$ can be passed to the solver:
The user can specify a list of potential candidate bounds to be tested, which can include bounds of various degrees and numbers of variables.

\subsection{Solving the Mixed Integer Linear Program}

We search for the best bound by iterating over potential candidate bounds for each degree \(d\) applicable to the number of variables in the input specification.
The algorithm selects the solution with minimal residual sum (difference between bound and costs).
To avoid overfitting, we prioritise bounds with fewer non-zero coefficients and minimise $k + \sum_{i=0}^{k} a_i$.
We use the CBC solver~\cite{Forrest:etal:2024} to solve the generated MILPs.
If a feasible solution is found, the solver returns the coefficients $a_0, a_1, \ldots, a_k$, which define the upper bound.

\section{Empirical Soundness of the Approach}
\label{section:soundness}

The empirical soundness of our approach rests on two foundations:

\begin{enumerate}
	\item \emph{Correctness of the simulation}: Every symbolic execution path corresponds to a concrete execution, and the symbolic execution systematically explores all execution paths up to a bounded input size.
	\item \emph{Soundness of MILP}: The bound derived via the MILP formulation provably upper-bounds all costs obtained from symbolic execution.
\end{enumerate}

Due to space constraints, we present only the statements of the following theorems here; all proofs are deferred to Appendix~\ref{appendix:proofs}.

\begin{definition}[Concretisation]
	\label{def:concretisation}
	A \emph{model} $M$ is a mapping from symbolic variables to concrete values produced by Z3 when path constraints $\assertionSet{}$ are satisfiable. We write $M \models \assertionSet{}$ when $M$ satisfies all constraints in $\assertionSet{}$.
	The \emph{concretisation} function $\gamma_M$ maps symbolic values to concrete values under model $M$.
\end{definition}

\begin{lemma}[Completeness: Symbolic $\implies$ Concrete]
	\label{lem:completeness}
	Given an input specification $\mathcal{I} = (\mathcal{P}, \IRfunName{main}, \argspecificationList{}, \initialAssertionSet{})$.
	For any path
	\[
		\exframe{\IRcall{\IRfunName{main}}{\argspecificationList{}}}{\exstate{\rho_{\text{init}}}{\initialAssertionSet{}}{([],[],[])}{0}}{\idcontinuation{}} \leadsto^* \kappa(v_s, \assertionSet{}, \mathcal{S}, \mathcal{C})
	\]
	where the path constraints are satisfiable with model $M \models \assertionSet{}$, there exists a concrete execution on the concretised environment $\rho_c = \{\, x \mapsto \gamma_M(\rho(x)) \mid x \in \text{dom}(\rho_{\text{init}}) \,\}$ that produces the concrete value $v_c = \gamma_M(v_s)$ with cost exactly $\mathcal{C}$.
\end{lemma}

\begin{definition}[Template from Concrete Values]
	\label{def:exact-template}
	For a concrete value $v$ of type $\theta$, we define $T_\theta(v)$ as the template whose template constraints $\templateAssertionSet{}$ fix the structural properties of $v$ exactly.
\end{definition}

\begin{lemma}[Soundness: Concrete $\implies$ Symbolic]
	\label{lem:soundness}
	Given an input specification $(\mathcal{P}, \IRfunName{main}, \argspecificationList{}, \initialAssertionSet{})$.
	For any concrete evaluation $\mathcal{P}(x_1, \ldots, x_k) \Downarrow^c v$, there exists a symbolic execution
	\[
		\exframe{\IRcall{\IRfunName{main}}{\argspecificationList{}}}{\exstate{\rho_{\text{init}}}{\initialAssertionSet{}}{([],[],[])}{0}}{\idcontinuation{}}\leadsto^* \kappa(v_s, \assertionSet{}, \mathcal{S}, \mathcal{C})
	\]
	where $\rho_{\text{input}} = \{x_i \mapsto \buildValue{\iota_i} \mid i = 1, \ldots, k\}$, such that $v = \gamma_M(v_s)$ where $M \models \assertionSet{}$, and the symbolic cost $\mathcal{C}$ equals the concrete evaluation cost $c$.
\end{lemma}

Together, Lemma~\ref{lem:soundness} and Lemma~\ref{lem:completeness}
establish that symbolic execution with bounded descriptors provides an \emph{empirically sound and complete} characterization
of program behaviour within those bounds.
This means the maximum cost across all symbolic executions is precisely the worst-case cost for all inputs within the specified size bounds.

\begin{theorem}[Soundness of Upper Bound]
	\label{thm:bound-soundness}
	Let $\mathcal{M} = \{(\vec{x}_i, c_i) \mid i = 1, \ldots, N\}$ be a finite set of costs, where $\vec{x}_i \in \mathbb{N}^k$ represents an input shape (a vector of input dimension sizes) and $c_i \in \mathbb{R}^+$ is the cost for that shape.
	If $\mathit{InferBound}(\mathcal{M})$ admits a feasible solution with a bound $f : \mathbb{N}^N \to \mathbb{R}$, then $f$ is an upper bound on all costs:
	\[
		f(\vec{x}_i) \geq c_i \quad \text{for all } (\vec{x}_i, c_i) \in \mathcal{M}
	\]
\end{theorem}

To obtain the empirical soundness of our approach,
we combine the results of Lemma~\ref{lem:soundness} and Theorem~\ref{thm:bound-soundness}
to show that if $\mathit{InferBound}(\mathcal{M})$ admits a feasible solution,
then the derived bound is sound for all concrete inputs of size up to $n$.

\begin{corollary}[Empirical Soundness]\label{cor:overall-soundness}
	Let $\mathcal{P}$ be a program and $\mathcal{I} = (\mathcal{P},\allowbreak \IRfunName{main},\allowbreak \argspecificationList{},\allowbreak \initialAssertionSet{})$
	an input specification.
	Moreover, let $\Pi$ be
	all feasible symbolic execution paths for $\IRfunName{main}$ in $\trcore{\mathcal{P}}$
	and $\mathcal{D}_{\text{max}}$
	the set of worst-case costs grouped by input shape as defined in Section~\ref{section:symbolic-execution:grouping}.
	If $\mathit{InferBound}(\mathcal{D}_{\text{max}})$ admits a feasible solution with complexity bound $f(\vec{x})$,
	then $f(\vec{x})$ is a \emph{sound upper bound} for the worst-case execution cost of $\mathcal{P}$
	for all concrete inputs of size $k \in \{1, \ldots, N\}$.
\end{corollary}

\paragraph{Proof Sketches.}
Lemma~\ref{lem:completeness} is shown by induction on the symbolic execution derivation.
Each operational rule either performs the same computation as the concrete semantics under concretisation, or adds a path condition that is satisfied by the chosen model \(M\).
For symbolic conditionals, satisfiability of the selected branch condition ensures that the corresponding concrete execution takes the same branch.
For structural descriptors, the cache records generated list, tuple, and tree components, so repeated accesses concretise consistently.
Ticks are accumulated only by the \textsc{Tick} rule and are therefore preserved step by step.

Lemma~\ref{lem:soundness} uses the converse construction.
Given a concrete input within the input specification, choose the template that fixes its shape and instantiate the fresh symbolic values with the concrete components.
Whenever execution branches on a descriptor, the bounded split rules include the case matching the concrete structure; whenever execution branches on a symbolic condition, the concrete valuation satisfies the corresponding path constraint.
Thus the concrete run is represented by some feasible symbolic path with the same accumulated tick cost.

Theorem~\ref{thm:bound-soundness} follows directly from the MILP constraints: feasibility entails \(f(\vec{x}_i)\geq c_i\) for every observation.
Combining with Lemmas~\ref{lem:soundness} and~\ref{lem:completeness}, \(\mathcal{D}_{\text{max}}\) contains the true worst-case cost for every input shape explored.
Therefore any feasible inferred bound upper-bounds all concrete executions in the bounded domain, yielding Corollary~\ref{cor:overall-soundness}.

\section{Experimental Evaluation}
\label{section:experimental-evaluation}

\begin{table}[!t]
	\centering
	\caption{Experimental results for all $46$ benchmark programs. Bounds from \texttt{RaML}, the hybrid approach by Pham et al.~\cite{Pham:etal:2025} (where available), and our method.}
	\label{tab:experiments2}
	\begingroup
	\scriptsize
	\setlength{\tabcolsep}{2pt}
	\renewcommand{\arraystretch}{0.68}
	\begin{tabularx}{\textwidth}{@{}l >{\raggedright\arraybackslash}X >{\raggedright\arraybackslash}X >{\raggedright\arraybackslash}X r c@{}}
		\toprule
		\textbf{\textsf{Program}}           & \textbf{\textsf{RaML}}~\cite{Hoffmann:etal:2017} & \textbf{\textsf{Hybrid}~\cite{Pham:etal:2025}} & \textbf{\textsf{Our Method}}                       & \textbf{\textsf{Time [s]}} & \textbf{\textsf{Sound}} \\
		\midrule
		\rowcolor{gray!10}
		\texttt{\textsf{AVLTree}}           & \infeasible                                      & $O((M + L) \log(M))$                           & $3 + 0.64M^3 + 1.36M^2 + ML + 2L$                  & 136.28                     & \sound{visual}          \\
		\texttt{\textsf{BFS}}               & $4M$                                             & -                                              & $3M$                                               & 71.00                      & \sound{manual}          \\
		\rowcolor{gray!10}
		\texttt{\textsf{BinarySearchTree}}  & $6 + 2L + LM - 0.5M + 3.5M^2$                    & $O((M + L)M)$                                  & $6 + 2L + LM + 2.5M^2$                             & 103.36                     & \sound{visual}          \\
		\texttt{\textsf{BitVectors}}        & $1 + M$                                          & -                                              & $1 + M$                                            & 25.72                      & \sound{eq}              \\
		\rowcolor{gray!10}
		\texttt{\textsf{BubbleSort}}        & \infeasible                                      & -                                              & $1 + M^2$                                          & 146.82                     & \sound{manual}          \\
		\texttt{\textsf{BubbleSort2}}       & \infeasible                                      & $O(M^2)$                                       & $2 + 0.78M^2 + 0.11M^3$                            & 161.41                     & \sound{manual}          \\
		\rowcolor{gray!10}
		\texttt{\textsf{Concat}}            & $M + ML$                                         & -                                              & $0.5M + ML + 0.5M^2$                               & 61.62                      & \sound{visual}          \\
		\texttt{\textsf{DivBySub}}          & \infeasible                                      & -                                              & $1 + 2M$                                           & 41.58                      & \sound{manual}          \\
		\rowcolor{gray!10}
		\texttt{\textsf{Duplicates}}        & $-0.5M_1 M + 0.5M_1 M^2 + 0.5M + 0.5M^2$         & -                                              & $M + 12M_1$                                        & 21.94                      & \unsound{manual}        \\
		\texttt{\textsf{Dyade}}             & $M + ML$                                         & -                                              & $M + ML$                                           & 58.77                      & \sound{eq}              \\
		\rowcolor{gray!10}
		\texttt{\textsf{Eratosthenes}}      & $0.5M + 0.5M^2$                                  & -                                              & $0.5M + 0.5M^2$                                    & 242.50                     & \sound{eq}              \\
		\texttt{\textsf{EvenOdd}}           & $1 + 0.5M$                                       & -                                              & $1 + M$                                            & 6.26                       & \sound{coeff}           \\
		\rowcolor{gray!10}
		\texttt{\textsf{EvenOddTail}}       & $3 + 7.75M + 0.75M^2$                            & -                                              & $1 + 2.36M + 1.64M^2$                              & 7.28                       & \sound{coeff}           \\
		\texttt{\textsf{FibMemo}}           & \infeasible                                      & -                                              & $2 + 11.95M + 1.35M^2$                             & 2.26                       & \sound{manual}          \\
		\rowcolor{gray!10}
		\texttt{\textsf{Flatten}}           & $0.5M^2 M_1 + 0.5M^2 M_1^2 + M_1 M + M$          & -                                              & $0.94M + 0.06M^2 + 0.33M M_1^2 + 2M^2 M_1$         & 269.61                     & \sound{manual}          \\
		\texttt{\textsf{FoldSum}}           & \infeasible                                      & -                                              & $3 + 4M$                                           & 4.63                       & \sound{manual}          \\
		\rowcolor{gray!10}
		\texttt{\textsf{Id}}                & \infeasible                                      & -                                              & $2 + 2M$                                           & 4.51                       & \sound{manual}          \\
		\texttt{\textsf{InsertionSort2}}    & $M + M^2$                                        & -                                              & $1.5M + 0.5M^2$                                    & 135.61                     & \sound{manual}          \\
		\rowcolor{gray!10}
		\texttt{\textsf{ISort}}             & $0.5M + 0.5M^2$                                  & -                                              & $0.5M + 0.5M^2$                                    & 19.75                      & \sound{eq}              \\
		\texttt{\textsf{ISortFold}}         & $0.5M + 0.5M^2$                                  & -                                              & $0.5M + 0.5M^2$                                    & 19.83                      & \sound{eq}              \\
		\rowcolor{gray!10}
		\texttt{\textsf{LCS}}               & $1 + L + 3LM + M$                                & -                                              & $0.5 + 1.5L + 2.25LM + 0.25M + 0.25M^2 + 0.25ML^2$ & 205.89                     & \sound{visual}          \\
		\texttt{\textsf{ListSort}}          & $-0.5M_1 M + 0.5M_1 M^2 + 0.5M + 0.5M^2$         & -                                              & $0.77M + 0.6M^2 M_1 + 0.23M^3$                     & 358.68                     & \sound{manual}          \\
		\rowcolor{gray!10}
		\texttt{\textsf{MapAppend}}         & $M_1 M + M$                                      & -                                              & $5M$                                               & 14.72                      & \sound{manual}          \\
		\texttt{\textsf{MapPlus}}           & $1 + 2M$                                         & -                                              & $1 + 2M$                                           & 15.61                      & \sound{eq}              \\
		\rowcolor{gray!10}
		\texttt{\textsf{MergeSort}}         & $-1.5M + 1.5M^2$                                 & -                                              & $1.54M \log(M)$                                    & 41.31                      & \sound{manual}          \\
		\texttt{\textsf{MergeSort2}}        & $1 + 3.96M + 0.29M^2$                            & $O(M\log(M))$                                  & $1 + 0.78M + 0.61M^2$                              & 12.87                      & \sound{visual}          \\
		\rowcolor{gray!10}
		\texttt{\textsf{MergeSortDC}}       & \infeasible                                      & -                                              & $2 + 13.57M \log(M)$                               & 44.84                      & \sound{manual}          \\
		\texttt{\textsf{MinSort}}           & $1.5M + 0.5M^2$                                  & -                                              & $0.5M + 0.5M^2$                                    & 167.15                     & \sound{visual}          \\
		\rowcolor{gray!10}
		\texttt{\textsf{MSS}}               & $5 + 4M$                                         & -                                              & $5 + 3M$                                           & 35.41                      & \sound{visual}          \\
		\texttt{\textsf{Queue}}             & $2.5M + 3.5M^2$                                  & -                                              & $3M + 3M^2$                                        & 10.20                      & \sound{visual}          \\
		\rowcolor{gray!10}
		\texttt{\textsf{QuickSelect}}       & $1 + 2M + M^2$                                   & -                                              & $1.5M^2$                                           & 140.09                     & \sound{visual}          \\
		\texttt{\textsf{QuickSort}}         & $1 + 3M + M^2$                                   & -                                              & $1 + 3M + M^2$                                     & 285.87                     & \sound{eq}              \\
		\rowcolor{gray!10}
		\texttt{\textsf{QuickSort2}}        & $1 + 3M + M^2$                                   & $O(M^2)$                                       & $1 + 2.27M + 1.24M^2$                              & 298.96                     & \sound{coeff}           \\
		\texttt{\textsf{RationalPotential}} & $M$                                              & -                                              & $M$                                                & 16.04                      & \sound{eq}              \\
		\rowcolor{gray!10}
		\texttt{\textsf{RedBlackTree}}      & \timeout{60s}                                    & $O((M + L) \log(M))$                           & $3 + M^2 + ML + 2M + 2L$                           & 144.11                     & \sound{visual}          \\
		\texttt{\textsf{Rev}}               & $1 + 1.5M + 0.5M^2$                              & -                                              & $1 + 1.5M + 0.5M^2$                                & 6.06                       & \sound{eq}              \\
		\rowcolor{gray!10}
		\texttt{\textsf{RevDL}}             & \infeasible                                      & -                                              & $1 + M$                                            & 4.46                       & \sound{manual}          \\
		\texttt{\textsf{RevFletf}}          & $2 + M$                                          & -                                              & $3 + M$                                            & 5.14                       & \sound{coeff}           \\
		\rowcolor{gray!10}
		\texttt{\textsf{RevFoldl}}          & $2 + 2M$                                         & -                                              & $2 + 2M$                                           & 4.82                       & \sound{eq}              \\
		\texttt{\textsf{Round}}             & \infeasible                                      & -                                              & $3 + 3.25M + 0.75M^2$                              & 18.51                      & \sound{manual}          \\
		\rowcolor{gray!10}
		\texttt{\textsf{SplayTree}}         & $3 + 3L + LM + 3.5M + 0.5M^2$                    & $O((M + L)M)$                                  & $3.29 + 1.71M^2 + 2.57L + 0.14L^2 + 0.76ML$        & 232.64                     & \sound{visual}          \\
		\texttt{\textsf{SplitAndSort}}      & $1.5M + 1.5M^2$                                  & -                                              & $2M + 2M^2$                                        & 170.48                     & \sound{coeff}           \\
		\rowcolor{gray!10}
		\texttt{\textsf{Subtrees}}          & $0.5M + 0.5M^2$                                  & -                                              & $0.5M + 0.5M^2$                                    & 30.13                      & \sound{eq}              \\
		\texttt{\textsf{Sum}}               & $2 + 2M$                                         & -                                              & $2 + 2M$                                           & 5.44                       & \sound{eq}              \\
		\rowcolor{gray!10}
		\texttt{\textsf{SumSqs3}}           & $9 + 7.33M + 2M^2 + 0.67M^3$                     & -                                              & $3 + 15.96M^2 + 0.04M^3$                           & 2.22                       & \sound{visual}          \\
		\texttt{\textsf{Tuples}}            & $-0.67M + 2.75M^2 - 1.33M^3 + 0.25M^4$           & -                                              & $0.81M + 0.19M^4$                                  & 16.11                      & \sound{visual}          \\
		\bottomrule
	\end{tabularx}
	\endgroup
\end{table}

We evaluate our approach on $46$ benchmark programs from two established sources:
the \texttt{RaML}~\cite{Hoffmann:etal:2017} benchmark suite ($39$ programs),
and all $7$ pure OCaml programs
from the hybrid approach by Pham et al.~\cite{Pham:etal:2025}
(\textsf{AVLTree}, \textsf{BinarySearchTree}, \textsf{BubbleSort2}, \textsf{MergeSort2}, \textsf{QuickSort2}, \textsf{RedBlackTree}, and \textsf{SplayTree}).%
\footnote{The remaining examples from~\cite{Pham:etal:2025} (\textsf{Dijkstra}, \textsf{Prim}, \textsf{BellmanFord}, \textsf{HeapSort} and \textsf{HuffmanCode})
	use RaML's specialised array library (\texttt{Rarray}). We leave the extension of our translator to support these constructs for future work.}
For each benchmark, we configure the symbolic execution with input specifications that bound the input sizes.
These bounds were determined through automatic incremental testing, balancing two competing objectives:
maintaining analysis time within reasonable limits while ensuring sufficient measurement points are generated to capture the program's resource behaviour.

All results are presented in Table~\ref{tab:experiments2}
where we report the inferred bounds from \texttt{RaML}, the hybrid approach~\cite{Pham:etal:2025} (where available\footnote{We report the $7$ overlapping pure OCaml benchmarks available in the evaluation of~\cite{Pham:etal:2025}. Their methodology requires manual annotation of programs with resource guards, a process that demands human expertise to identify appropriate decomposition points and instrument the code accordingly.}),
and our method, along with the total analysis time (symbolic execution + MILP solving).
We quantify resource consumption as the number of \emph{ticks}, consistent with \texttt{RaML}'s cost model,
where each tick represents a unit of computational work explicitly marked in the source code.
For timing measurements, we report for \texttt{RaML} the time taken in Module Mode with degree 3,
and for our approach the sum of the time for symbolic execution and MILP-based bound synthesis.
In the reported bounds, primary components are denoted by $M,L,\ldots$ in argument order.
Secondary components inherit the symbol of their top-level argument with an index: $M_1,M_2,\ldots$ are secondary dimensions of the first argument, $L_1,L_2,\ldots$ of the second, and so on.
For example, for an input list of lists, $M$ denotes the length of the outer list and $M_1$ denotes the maximum length of an inner list. The asymptotic results taken from~\cite{Pham:etal:2025} are rewritten to this notation for readability.

\paragraph{Comparison and Validation.}
The \emph{Sound} column records the validation method of the reported bound.
The marker \sound{eq} means that our bound coincides with the reference bound, \sound{coeff} means that it is validated by coefficient-wise domination, \sound{visual} means that plots support the comparison and the remaining argument is by inspection of the displayed bound, and \sound{manual} means that a program-specific semantic argument is used.
The marker \unsound{manual} denotes a case where this validation fails.
We validate soundness by comparing our bounds against \texttt{RaML}'s statically-derived bounds where available,
and against the known ground truth bounds reported in~\cite{Pham:etal:2025}.
\texttt{RaML} successfully derived worst-case upper bounds for $34$ programs.
Among these, our predictions match those obtained by \texttt{RaML} in $20$ instances; $12$ of them exactly (\sound{eq}), and $8$ with identical structure but differing coefficients (\sound{coeff}).
For $14$ programs, our predictions deviated structurally from those generated by \texttt{RaML}.
For the remaining $10$ programs, \texttt{RaML} was not able to derive a bound or timed out.
\begin{figure}[t]
	\centering
	\begin{subfigure}[b]{0.33\textwidth}
		\centering
		\includegraphics[width=.9\textwidth]{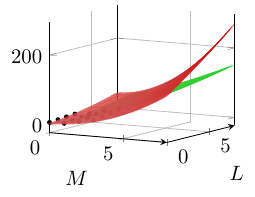}
		\subcaption{\program{BinarySearchTree}.}\label{fig:bst-example}
	\end{subfigure}
	\hfill
	\begin{subfigure}[b]{0.3\textwidth}
		\centering
		\includegraphics[width=\textwidth]{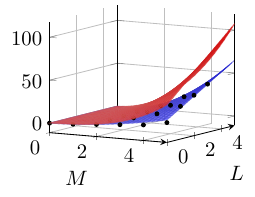}
		\subcaption{\program{ListSort}.}\label{fig:listsort-example}
	\end{subfigure}
	\hfill
	\begin{subfigure}[b]{0.3\textwidth}
		\centering
		\includegraphics[width=\textwidth]{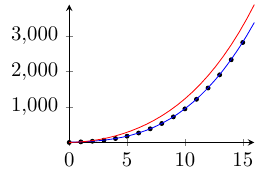}
		\subcaption{\program{SumSqs3}.}\label{fig:sum_sqs3-example}
	\end{subfigure}
	\caption{Three representative examples from Table~\ref{tab:experiments2} for which the plots support the soundness validation. The black dots represent the observed costs, the blue surface the static analysis bound, the red surface our prediction, and the green surface the ground truth bound taken from the artifact of \cite{Pham:etal:2025}.}
	\label{fig:plots-examples-main}
\end{figure}
For the programs \program{Concat}, \program{ListSort} and
\program{SumSqs3} (in Figure~\ref{fig:plots-examples-main}),
as well as the additional plotted examples included in the artifact and the extended version~\cite{Frontull:etal:2026:extended}, the plots provide visual evidence that the inferred upper bounds either closely approximate or conservatively exceed the reference bounds determined by \program{RaML}.
For the remaining examples, particularly those where static analysis could not infer a bound (e.g. \program{BubbleSort}), plots alone are insufficient to validate soundness, as the bounds follow the trend of the data points but manual reasoning (\sound{manual}) is required to confirm they correctly capture worst-case behaviour.
In the following, we discuss a selection of representative examples that we validated manually through detailed analysis of the program semantics. Further benchmark plots and case studies are included in the artifact and the extended version~\cite{Frontull:etal:2026:extended}.

\paragraph{BFS and InsertionSort2.}
For \program{BFS}, ticks occur only in the auxiliary list reversal used to move the accumulated next-level queue into the current queue.
The worst case is a skew tree in which each level reversal processes two queued children and the empty-list case, giving $3$ ticks per node and hence the bound $3M$.
For \program{InsertionSort2}, the first insertion sort has worst-case cost $\sum_{i=1}^{M} i = 0.5M^2+0.5M$; the second pass sorts an already sorted list, so each insertion stops after one comparison, adding $M$ ticks.
This yields $0.5M^2+1.5M$, matching the inferred bound.

\paragraph{bubblesort.} The \program{BubbleSort} implementation uses an auxiliary
function \raml{scan\allowbreak\_and\_swap} that traverses the list,
swapping adjacent elements if they are out of order, and returns a boolean indicating whether any swaps were made. The main function \raml{bubble\_sort xs} repeatedly calls this helper function until no swaps are needed, indicating that the list is sorted.
In the worst case, the swap operation needs to be performed for every pair of adjacent elements in the list.
This would be the case for a list sorted in descending order which would require \(n\) passes through the list, with each pass involving \(n\) swaps. This results in a total of \(n^2\) ticks, except for $n = 0$, where the cost is $1$. $1 + n^2$, as predicted by our approach, is thus a valid upper bound on the worst-case execution time of \raml{bubble\_sort}.
\texttt{RaML}, on the other hand, successfully infers the $1 + n$ bound for the \raml{scan\_and\_swap} function, but fails to combine this with the outer recursive calls in \raml{bubble\_sort}.

\paragraph{duplicates.}
The \program{Duplicates} program removes duplicate elements from a list of lists using the \raml{nub l} function. For each element $x$ in the input list, \raml{nub} recursively removes all occurrences of $x$ from the remaining elements by calling \raml{remove x xs}, which in turn uses \raml{eq} to compare $x$ with each remaining element. The \raml{eq} function performs structural equality on lists, executing 1 tick per element comparison. For a list of $M$ sublists each with maximum length $M_1$, the worst-case occurs when all elements are distinct, requiring $\sum_{i=1}^{M-1} (M-i) \cdot M_1 = \Theta(M_1 M^2)$ ticks for pairwise comparisons. \texttt{RaML} correctly infers the quadratic bound $-0.5M_1 M + 0.5M_1 M^2 + 0.5M + 0.5M^2$. Our approach, however, infers the linear bound $M + 12M_1$,
which is not sound against this reference bound.
This example illustrates a limitation of our approach: the selected input dimensions, bounded observations, and bound template did not expose the interaction between the number of sublists and pairwise equality checks strongly enough to synthesise the required $M_1M^2$ term.

\paragraph{mergesort2.}
The \program{MergeSort2} benchmark, taken from~\cite{Pham:etal:2025}, is a variant of the merge sort algorithm with theoretical complexity \(O(M \log(M))\).
For this example, our approach infers a quadratic bound $1 + 0.78M + 0.61M^2$ (which is sound) rather than the logarithmic bound. This occurs because within our observation range, the quadratic model achieves a better fit (lower sum of residuals) and is therefore preferred by the MILP objective.
The hybrid approach in~\cite{Pham:etal:2025} successfully identifies the correct logarithmic asymptotic bound for this benchmark. However, as noted in their work, their data-driven analysis achieves $59.5\%$ soundness proportion for the sampled bounds, which differs from the exhaustive bounded guarantee targeted here.

\paragraph{Tool.}
We have implemented our approach in a prototype tool called \compas{}.
The tool is a command-line application that integrates an OCaml-to-core-language translator, a symbolic execution engine with Z3-based constraint solving, and a Python-based MILP solver for bound inference.
The source code and experimental artifact are released under the Creative Commons Attribution 4.0 International license (CC BY 4.0)\footnote{\url{https://creativecommons.org/licenses/by/4.0/}} and archived on Zenodo~\cite{frontull_2026_21323697}; implementation details are documented in the extended version~\cite{Frontull:etal:2026:extended}.

\section{Related Work}
\label{section:related-work}

Complexity analysis techniques span a spectrum from purely static to purely dynamic approaches, with recent hybrid methods attempting to combine their complementary strengths.
In this section, we discuss related work.

\paragraph{Static Analysis.}
Static methods derive complexity bounds from program structure without execution.
A prominent example is Automatic Amortized Resource Analysis (AARA), pioneered by Hoffmann et al.~\cite{Hoffmann:2012:ARA,Hoffmann:2012:RaML,Hoffmann:etal:2017}, which uses type-based analysis and linear programming to infer resource bounds for higher-order functional programs, building on Tarjan's amortised analysis~\cite{Tarjan:1985}. See~\cite{leutgeb2021atlas,leutgeb2022automated,WMSZ26} for related approaches.
Other static techniques include size-change abstraction for imperative programs~\cite{Zuleger:etal:2011}, transformations from higher-order to first-order systems~\cite{Avanzini:etal:2015}, and predicate transformers~\cite{AMS20,Avanzini:2023}.
Static analysis methods have also been developed in the context of term rewriting, a versatile abstract computational
model, cf.~\cite{Avanzini:etal:2016}.
\paragraph{Symbolic Execution for Complexity Analysis.}
Symbolic execution~\cite{King:1976} is a well-known technique for program analysis.
It has been successfully applied to various domains.
One of the most prominent examples of symbolic execution is the application to error finding in C/C++ code within Microsoft~\cite{Bush:etal:2000}.
Other applications include software testing~\cite{Cadar:etal:2011}, security and code analysis~\cite{Baldoni:etal:2018}.
KLEE~\cite{Cadar:etal:2008} is a popular dynamic symbolic execution engine initially designed at Stanford University and now primarily maintained by the Software Reliability Group at Imperial College London, with a large community spanning both academia and industry.
KLEE has been used and extended in diverse areas such as high-coverage test generation, automated debugging, exploit generation, and wireless sensor networks, among others.
SAGE~\cite{Godefroid:etal:2008} presents a whitebox fuzz testing approach that combines symbolic execution with constraint solving to generate new inputs that exercise different control paths, successfully discovering vulnerabilities in large Windows applications including the MS07-017 ANI vulnerability.
Pex~\cite{Tillmann:deHalleux:2008} automatically produces test suites with high code coverage for .NET programs using dynamic symbolic execution similar to path-bounded model-checking, learning program behaviour by monitoring execution traces and using constraint solvers to generate new test inputs.
Similarly, \emph{fuzzing}~\cite{Boehme:etal:2021} is a related technique as well. Fuzzing is used to uncover software bugs, by changing the input over and over again in order to reach all possible control flows.
Luckow et al.~\cite{Luckow:etal:2020} use guided symbolic execution with learned path policies to identify worst-case behaviours, exploring exhaustively at small input sizes and generalizing to larger ones.
Pasareanu et al.~\cite{Pasareanu:etal:2016} are close in spirit as they also enumerate symbolic paths and associate them with accumulated costs (in the setting of side-channel analysis).
In our setting, the accumulated costs serve a different purpose:
they are related to source-level input measures which become the variables of the inferred resource bounds.
While such measures could be encoded separately in a side-channel analysis, they are not the quantities retained and optimized by Pasareanu et al.'s formulation, which focuses on distinguishable observations.
Badger~\cite{Noller:etal:2018} combines fuzzing with symbolic execution, while SlowFuzz~\cite{Petsios:etal:2017} uses fuzzing to detect algorithmic complexity vulnerabilities.
Unlike complexity-testing approaches that search for worst cases through heuristics, we exhaustively enumerate all executions within bounded input sizes using SMT solving, guaranteeing discovery of true worst cases in the analysed range.

\paragraph{Dynamic Analysis.}
Purely dynamic methods infer bounds from observed executions.
They can be categorised into symbolic execution~\cite{Noller:etal:2018} and actual/dynamic execution~\cite{Nguyen:2020}, where symbolic tools like WISE~\cite{Burnim:2009} strive to abstract away concrete input data values.
Singularity~\cite{Wei:2018} is closely related, but follows a different route:
it models worst-case input generation as an optimal program-synthesis problem~\cite{Bornholt:etal:2016} and uses genetic programming to synthesise input patterns represented as recurrent computation graphs.
The authors show that this approach can effectively discover worst-case complexity for a range of algorithms and scales to larger real-world applications, where it can expose performance bugs and availability vulnerabilities.
However, the search remains heuristic and does not provide soundness guarantees for the inferred worst-case behaviour.
At the same time, this suggests a complementary direction for our approach:
the symbolic executions and concrete models produced by our bounded analysis, in particular those attaining maximal costs for small input sizes, could be used as structured examples from which to learn worst-case input patterns.
Such learned patterns could guide exploration at larger sizes, while the exhaustive symbolic phase would still provide validation within selected finite ranges.
Dynaplex~\cite{Ishimwe:etal:2021} learns recurrence relations from execution traces and solves them to obtain closed-form complexity bounds, handling even non-polynomial cases like $O(n^{1.58})$.
Similar data-driven approaches for term rewriting systems~\cite{Frontull:etal:2025} are limited to univariate problems and provide no soundness guarantees.
Machine learning has been applied to data-driven program analysis using disjunctive models~\cite{Jeon:etal:2019}.
While dynamic methods can handle complex programs that resist static analysis, they typically lack guarantees about bound validity.
Our approach differs by using symbolic execution to systematically cover all paths up to bounded sizes, then applying MILP-based bound synthesis with explicit upper-bounding constraints, yielding provably valid upper bounds within the analysed range.

\paragraph{Hybrid Approaches.}
Recent work has combined static and dynamic techniques to leverage their complementary strengths.
Pham et al.~\cite{Pham:2024} introduced Hybrid AARA, which integrates AARA's type-based analysis with Bayesian inference over runtime measurements to infer statistically sound probability distributions over resource bounds.
Their resource decomposition framework~\cite{Pham:etal:2025} extends this idea by enabling heterogeneous analysis techniques to be applied to different program components while providing provable composition guarantees.
The framework requires manually instrumenting the source program with resource guards.
AARA/\texttt{RaML} then analyses the instrumented program to derive an overall cost bound parameterised by these guards, while data-driven techniques infer bounds for the corresponding guarded resource components.
Although~\cite{Pham:etal:2025} successfully infers the correct asymptotic complexity classes for challenging programs, its coefficient-level guarantees are expressed as posterior soundness probabilities rather than as exhaustive bounds over the input specifications considered here.
Like these hybrid approaches, we combine systematic exploration (via symbolic execution) with empirical bound synthesis (via linear programming).
However, we target a different trade-off: our guarantee is exhaustive over explicitly bounded input specifications, rather than probabilistic over sampled or resource-decomposed measurements, at the cost of being limited to analysable input sizes.
Our method derives bounds with precise coefficients in a single automated step,
without requiring manual resource annotations and independent of any other static analysis.
Among the $7$ overlapping pure OCaml benchmarks, our method agrees asymptotically with the function-call tick-guard bounds reported in~\cite{Pham:etal:2025} for \program{BinarySearchTree} and \program{QuickSort2}.

\section{Conclusion}
\label{section:conclusion}
In this work, we have presented a novel approach to the resource analysis of functional programs that combines symbolic execution with linear programming-based inference of upper bounds. Our method collects symbolic cost data and translates bounding and monotonicity requirements into linear constraints on bound templates. By solving the resulting linear programs, we obtain sound and verifiable upper bounds on program resource consumption.
Beyond the core methodology, we have demonstrated how this framework can be applied to representative functional programs, showing that the inferred bounds are both precise and computationally tractable.

We see our work as a foundation for several promising future directions.
Detecting recurring patterns in program structures could allow for more efficient exploration and analysis, potentially reducing the computational cost. Such pattern detection could be guided by the concrete examples of worst-case inputs that our method generates.
Our method also opens avenues for deriving closed-form expressions that describe algorithmic behaviour by analysing the collected symbolic cost data.
Given the effectiveness of our method, it could also serve as a diagnostic tool to identify code segments that may lead to suboptimal performance in IDEs.
Extending the approach to explicitly handle multiple resource types simultaneously would enable comprehensive resource profiling, analysing trade-offs between different resources (e.g., time vs. space).

While our presentation and examples focus primarily on worst-case runtime complexity, the tick-based cost model is fully general and supports analysis of arbitrary resources.
The support for negative costs naturally extends to modelling resource lifecycles, such as tracking heap memory usage patterns including peak allocation points.
Another promising approach is to integrate large language models (LLMs) into the symbolic execution workflow. Thanks to their strong capabilities in pattern recognition and code comprehension, LLMs could be used to guide the exploration strategy by prioritising execution paths where the worst-case behaviour is more likely, thereby narrowing the search space and improving the efficiency of the analysis. Beyond this, LLMs could also help to derive generalisations from the limited set of worst-case executions.

\subsubsection{Acknowledgements} We would like to thank the anonymous reviewers for their work
and invaluable suggestions, which greatly improved our presentation.
This work was partly supported by the FWF project 
``Automated Sublinear Amortised Resource Analysis of Data Structures'' (AUTOSARD), No. P3662.

\appendix
\renewcommand{\theHsection}{appendix.\Alph{section}}
\section{Dynamic Symbolic Execution Rules}
\label{appendix:execution-rules}
\begin{figure}[t]
	\begin{simulationrules}
		\inferrule*[lab={\IRLInt}]
		{ }
		{\exframe{\IRint{n}}{\exstate{\rho}{\assertionSet{}}{\mathcal{S}}{\mathcal{C}}}{\kappa} \leadsto \kappa(\text{Z3.Int.val}(\IRint{n}), \assertionSet{}, \mathcal{S}, \mathcal{C})}
		\inferrule*[lab={\IRLBool}]
		{ }
		{\exframe{\IRbool{b}}{\exstate{\rho}{\assertionSet{}}{\mathcal{S}}{\mathcal{C}}}{\kappa} \leadsto \kappa(\text{Z3.Bool.val}(\IRbool{b}), \assertionSet{}, \mathcal{S}, \mathcal{C})}
		
    \inferrule*[lab={\IRLUnaryOp}]
		{\phantom{\top} \\\\ \exframe{\tau}{\sigma}{\idcontinuation{}} \leadsto (\symbolic{v}, \assertionSet{}', \mathcal{S}', \mathcal{C}') \\\\
			\diamond \in \{\,\lnot,\; -_1,\; \sqrt{\cdot},\; \sqrt[3]{\cdot}\,\}}
		{\exframe{\diamond \tau}{\sigma}{\kappa} \leadsto
			\kappa(\unopmap{\diamond}{\symbolic{v}},\; \assertionSet{}',\, \mathcal{S}',\, \mathcal{C}')}
		\inferrule*[lab={\IRLOp}]
		{\exframe{\tau_1}{\sigma}{\idcontinuation{}} \leadsto (\symbolic{v_1}, \assertionSet{}_1, \mathcal{S}_1, \mathcal{C}_1) \\\\
			\exframe{\tau_2}{\exstate{\rho}{\assertionSet{}_1}{\mathcal{S}_1}{\mathcal{C}_1}}{\idcontinuation{}} \leadsto (\symbolic{v_2}, \assertionSet{}', \mathcal{S}', \mathcal{C}') \\\\
			\bullet \in \{+,\,-_2,\,*,\,/,\,\mathbin{\%},\,<,\,{\leq},\,{=},\,{\neq},\,{\geq},\,>,\,{\land},\,{\lor}\}}
		{\exframe{\tau_1 \bullet\, \tau_2}{\sigma}{\kappa} \leadsto
			\kappa(\biopmap{\bullet}{\symbolic{v_1}}{\symbolic{v_2}},\; \assertionSet{}',\, \mathcal{S}',\, \mathcal{C}')}
		
    \inferrule*[lab={\IRLVar}]
		{\IRvar{x} \in \text{dom}(\rho)}
		{\exframe{\IRvar{x}}{\exstate{\rho}{\assertionSet{}}{\mathcal{S}}{\mathcal{C}}}{\kappa} \leadsto \kappa(\exstate{\rho(\IRvar{x})}{\assertionSet{}}{\mathcal{S}}{\mathcal{C}})}
		\inferrule*[lab={\IRLTick}]
		{c \in \mathbb{R}}
		{\exframe{\IRtick{c}{\tau}}{\exstate{\rho}{\assertionSet{}}{\mathcal{S}}{\mathcal{C}}}{\kappa} \leadsto \exframe{\tau}{\exstate{\rho}{\assertionSet{}}{\mathcal{S}}{\mathcal{C} + c}}{\kappa}}
		
    \inferrule*[lab={\IRLLet}]
		{ \exframe{\tau_1}{\sigma}{\idcontinuation{}} \leadsto (\symbolic{v}, \assertionSet{}', \mathcal{S}', \mathcal{C}')}
		{\exframe{\IRlet{\IRvar{x}}{\tau_1}{\tau_2}}{\sigma}{\kappa} \leadsto
			\exframe{\tau_2}{\exstate{\rho[\IRvar{x} \mapsto \symbolic{v}]}{\assertionSet{}'}{\mathcal{S}'}{\mathcal{C}'}}{\kappa}}
	\end{simulationrules}
	\caption{Execution rules for primitive operations, variable binding, and tick operations.}
	\label{fig:execution-rules-1}
\end{figure}
\paragraph{Primitive operations} such as integer and boolean constants, variable lookups, and unary/binary operations are handled by straightforward rules that evaluate the term and update the state accordingly.
Integer and boolean constants are lifted to Z3 values and variables are looked up in the environment $\rho$.
The \IRLTick{} rule performs a pure cost accumulation operation. It takes a cost value $c \in \mathbb{R}$ (which may be a concrete number or a symbolic expression) and adds it to the current accumulated cost $\mathcal{C}$, producing a new accumulated cost $\mathcal{C} + c$. Critically, this rule has no premises: executing a tick operation always succeeds immediately without performing any computation or evaluation. The environment $\rho$ and assertion set $\assertionSet{}$ as well as structure cache $\mathcal{S}$ remain unchanged, as ticking affects only the cost dimension of the execution state.
All source-language operators are lifted to Z3 by the operator mapping $\opmap{\cdot}$ defined below.
Unary and binary operations are then handled uniformly by two rules:
The lifting $\opmap{\cdot}$ maps each source operator to its Z3 counterpart, as illustrated in Figure~\ref{fig:operator-mapping}.
\begin{figure}[b]
	\scriptsize
	\[
		\begin{array}{rcl@{\qquad}rcl@{\qquad}rcl@{\qquad}rcl}
			\opmap{\lnot}           & = & \mathrm{Z3.Not}                     &
			\opmap{-_1}             & = & \mathrm{Z3.Neg}                     &
			\opmap{\sqrt{\cdot}}    & = & \mathrm{Z3.Sqrt}                    &
			\opmap{\sqrt[3]{\cdot}} & = & \mathrm{Z3.Cbrt}                      \\
			\opmap{+}               & = & \mathrm{Z3.Add}                     &
			\opmap{-_2}             & = & \mathrm{Z3.Sub}                     &
			\opmap{*}               & = & \mathrm{Z3.Mul}                     &
			\opmap{/}               & = & \mathrm{Z3.Div}                       \\
			\opmap{\mathbin{\%}}    & = & \mathrm{Z3.Mod}                     &
			\opmap{<}               & = & \mathrm{Z3.LT}                      &
			\opmap{\leq}            & = & \mathrm{Z3.LE}                      &
			\opmap{>}               & = & \mathrm{Z3.GT}                        \\
			\opmap{\geq}            & = & \mathrm{Z3.GE}                      &
			\opmap{=}               & = & \mathrm{Z3.Eq}                      &
			\opmap{\neq}            & = & \mathrm{Z3.Not}\circ \mathrm{Z3.Eq} &
			\opmap{\land}           & = & \mathrm{Z3.And}                       \\
			\opmap{\lor}            & = & \mathrm{Z3.Or}
		\end{array}
	\]
	\caption{The operator mapping $\opmap{\cdot}$ defines how source language operators are lifted to their Z3 equivalents.}
	\label{fig:operator-mapping}
\end{figure}
The rule for variable binding (\texttt{let}) captures the semantics of introducing a new variable in the current scope:
The premise states that the bound expression $\tau_1$ is evaluated in the current state $\sigma = \exstate{\rho}{\assertionSet{}}{\mathcal{S}}{\mathcal{C}}$. This evaluation produces a value $v$ and potentially modifies the assertion set, structure cache and accumulated cost, yielding updated values $\assertionSet{}'$, $\mathcal{S}'$ and $\mathcal{C}'$. Importantly, these primed values ($\assertionSet{}'$, $\mathcal{S}'$ and $\mathcal{C}'$) are \emph{outputs} of evaluating $\tau_1$, not inputs; they represent the state after $\tau_1$ has been fully evaluated.
Once $\tau_1$ has been evaluated, the conclusion of the rule extends the environment $\rho$ with the new binding \IRvar{x}, and continues by evaluating the body $\tau_2$ in this extended environment $\rho[\IRvar{x} \mapsto \symbolic{v}]$, with the updated assertions $\assertionSet{}'$, structure cache $\mathcal{S}'$, and cost $\mathcal{C}'$ carried forward. This models lexical scoping and allows for variable shadowing, as a new binding for \IRvar{x} can temporarily override an existing one.
\begin{figure}[t]
	\begin{simulationrules}
		\inferrule*[lab={\IRLFun}]
		{ }
		{\exframe{\IRfun{[\vec{x}]}{\tau}}{\sigma}{\kappa} \leadsto \kappa(\closure{\vec{x}}{\tau}{\rho}, \assertionSet{}, \mathcal{S}, \mathcal{C})}
		
    \inferrule*[lab={\IRLCallInit}]
		{\exframe{\tau_f}{\sigma}{\idcontinuation{}} \leadsto (\closure{\IRvar{\vec{x}}}{\tau_b}{\rho_c}, \assertionSet{}', \mathcal{S}', \mathcal{C}')}
		{\exframe{\IRcall{\tau_f}{\vec{\tau_a}}}{\sigma}{\kappa} \leadsto
			\exframe{\IRcall{\closure{\IRvar{\vec{x}}}{\tau_b}{\rho_c}}{\vec{\tau_a}}}{\exstate{\rho}{\assertionSet{}'}{\mathcal{S}'}{\mathcal{C}'}}{\kappa}}
		
    \inferrule*[lab={\IRLCallStep}]
		{i < n \\
			\exframe{\tau_i}{\sigma}{\idcontinuation{}} \leadsto (\symbolic{v_i}, \assertionSet{}_i, \mathcal{S}_i, \mathcal{C}_i)}
		{\exframe{\IRcall{\closure{\IRvar{\vec{x}}}{\tau_b}{\rho_c}}{[\symbolic{v_0}, \dots, \symbolic{v_{i-1}}, \tau_i, \dots, \tau_{n-1}]}}{\sigma}{\kappa} \leadsto \\\\
			\exframe{\IRcall{\closure{\IRvar{\vec{x}}}{\tau_b}{\rho_c}}{[\symbolic{v_0}, \dots, \symbolic{v_i}, \tau_{i+1}, \dots, \tau_{n-1}]}}{\exstate{\rho}{\assertionSet{}_i}{\mathcal{S}_i}{\mathcal{C}_i}}{\kappa}}
		
    \inferrule*[lab={\IRLCallEval}]
		{cl = \closure{\IRvar{[\IRvar{x_1},\dots,\IRvar{x_n}]}}{\tau_b}{\rho_c} }
		{\exframe{\IRcall{cl}{[\symbolic{v_1}, \dots, \symbolic{v_n}]}}{\sigma}{\kappa} \leadsto \\
			\exframe{\tau_b}{\exstate{\rho_c[\IRvar{x_1} \mapsto \symbolic{v_1}, \dots, \IRvar{x_n} \mapsto \symbolic{v_n}]}{\assertionSet{}}{\mathcal{S}}{\mathcal{C}}}{\kappa}}
		
      \inferrule*[lab={\IRLCallPartial}]
		{cl = \closure{\IRvar{[\IRvar{x_1},\dots,\IRvar{x_n}]}}{\tau_b}{\rho_c} \\
			m < n}
		{\exframe{\IRcall{cl}{[\symbolic{v_1}, \dots, \symbolic{v_m}]}}{\sigma}{\kappa} \leadsto
			\kappa(\closure{\IRvar{[\IRvar{x_{m+1}},\dots,\IRvar{x_n}]}}{\tau_b}{\rho_c[\IRvar{x_1} \mapsto \symbolic{v_1}, \dots, \IRvar{x_m} \mapsto \symbolic{v_m}]},\assertionSet{}, \mathcal{S}, \mathcal{C})}
	\end{simulationrules}
	\caption{Execution rules for function calls and closures.}
	\label{fig:execution-rules-2}
\end{figure}
\paragraph{Function Operations.}
Function operations involve creating closures and invoking them through a sequence of rules:
{\IRLFun{}.} Creates a closure $\closure{\vec{x}}{\tau}{\rho}$ that captures the formal parameters $\vec{x}$, function body $\tau$, and current environment $\rho$. No computation is performed.
	{\IRLCallInit{}.} Evaluates the function expression $\tau_f$ to obtain a closure value. This allows first-class functions (functions stored in variables or returned from other functions).
	{\IRLCallStep{}.} Evaluates arguments left-to-right. When argument $i$ is not yet a value, it is evaluated with continuation $\idcontinuation{}$, producing value $v_i$ and updated state. The call expression progresses with $\tau_i$ replaced by $v_i$.
	{\IRLCallEval{}.} Once all arguments are values, the closure is unpacked: the function body $\tau_b$ is evaluated in the closure's captured environment $\rho_c$ extended with parameter bindings $\IRvar{x_i} \mapsto \symbolic{v_i}$. Using $\rho_c$ (not the current environment $\rho$) implements lexical scoping: variables are resolved according to where the function was defined, not where it is called.
	{\IRLCallPartial{}.} When fewer arguments ($m$) are provided than the closure expects ($n$), a new closure is returned with the first $m$ parameters bound in its environment, awaiting the remaining $n - m$ arguments. This enables currying.
\paragraph{Control Flow.}
Conditional branching handles both concrete and symbolic conditions through four rules:
{\IRLIfElseTrue{} / \IRLIfElseFalse{}.} When the condition $\tau_c$ evaluates to $\symbolic{v_c}$ that simplifies to $\top$ (respectively $\bot$), only the then-branch (respectively else-branch) is explored. No branching occurs.
	{\IRLIfElseThen{} / \IRLIfElseElse{}.} When $\symbolic{v_c}$ is symbolic and does not simplify to a concrete boolean, the SMT solver checks satisfiability: if $\text{Z3.sat}(\assertionSet{}' \cup \{\symbolic{v_c}\})$ holds, the then-branch is explored with assertions $\assertionSet{}' \cup \{\symbolic{v_c}\}$; if $\text{Z3.sat}(\assertionSet{}' \cup \{\neg \symbolic{v_c}\})$ holds, the else-branch is explored with assertions $\assertionSet{}' \cup \{\neg \symbolic{v_c}\}$. When both are satisfiable, symbolic execution forks into two independent paths, each carrying its respective path constraint. This enables exhaustive exploration of all feasible control-flow paths.
The \IRLRaise{} rule handles exceptions by immediately terminating the current execution path with an exception message $m$. No continuation is invoked, and no further computation occurs on this path. This models exceptional control flow, which is important for handling error conditions in the analysed programs.
\begin{figure}[t]
	\begin{simulationrules}
		\inferrule*[lab={\IRLIfElseTrue}]
		{\exframe{\tau_c}{\sigma}{\idcontinuation{}} \leadsto (\symbolic{v_c}, \assertionSet{}', \mathcal{S}', \mathcal{C}') \\\\
			\text{Z3.simplify}(\symbolic{v_c}) = \top}
		{\exframe{\IRifelse{\tau_c}{\tau_1}{\tau_2}}{\sigma}{\kappa} \leadsto
			\exframe{\tau_1}{\exstate{\rho}{\assertionSet{}'}{\mathcal{S}'}{\mathcal{C}'}}{\kappa}}
		\inferrule*[lab={\IRLIfElseFalse}]
		{\exframe{\tau_c}{\sigma}{\idcontinuation{}} \leadsto (\symbolic{v_c}, \assertionSet{}', \mathcal{S}', \mathcal{C}') \\\\
			\text{Z3.simplify}(\symbolic{v_c}) = \bot}
		{\exframe{\IRifelse{\tau_c}{\tau_1}{\tau_2}}{\sigma}{\kappa} \leadsto
			\exframe{\tau_2}{\exstate{\rho}{\assertionSet{}'}{\mathcal{S}'}{\mathcal{C}'}}{\kappa}}
		
    \inferrule*[lab={\IRLIfElseThen}]
		{\exframe{\tau_c}{\sigma}{\idcontinuation{}} \leadsto (\symbolic{v_c}, \assertionSet{}', \mathcal{S}', \mathcal{C}') \\\\
			\text{Z3.simplify}(\symbolic{v_c}) \notin \{\top, \bot\} \\\\
			\text{Z3.sat}(\assertionSet{}' \cup \{\symbolic{v_c}\})}
		{\exframe{\IRifelse{\tau_c}{\tau_1}{\tau_2}}{\sigma}{\kappa} \leadsto \\\\
			\exframe{\tau_1}{\exstate{\rho}{\assertionSet{}' \cup \{\symbolic{v_c}\}}{\mathcal{S}'}{\mathcal{C}'}}{\kappa}}
		\inferrule*[lab={\IRLIfElseElse}]
		{\exframe{\tau_c}{\sigma}{\idcontinuation{}} \leadsto (\symbolic{v_c}, \assertionSet{}', \mathcal{S}', \mathcal{C}') \\\\
			\text{Z3.simplify}(\symbolic{v_c}) \notin \{\top, \bot\} \\\\
			\text{Z3.sat}(\assertionSet{}' \cup \{\neg \symbolic{v_c}\})}
		{\exframe{\IRifelse{\tau_c}{\tau_1}{\tau_2}}{\sigma}{\kappa} \leadsto \\\\
			\exframe{\tau_2}{\exstate{\rho}{\assertionSet{}' \cup \{\neg \symbolic{v_c}\}}{\mathcal{S}'}{\mathcal{C}'}}{\kappa}}
	\end{simulationrules}
	\caption{Execution rules for control flow constructs.}
	\label{fig:execution-rules-control-flow}
\end{figure}
\paragraph{List Operations.}
The empty list \IRnil{} maps to a distinguished value $\widehat{\text{nil}}$.
List operations support both construction and pattern matching on concrete lists and symbolic list descriptors:
{\IRLCons{}.} Constructs a list by evaluating head $\tau_1$ to $\symbolic{v_h}$ and tail $\tau_2$ to $\symbolic{v_t}$, producing $\symbolic{v_h} :: \symbolic{v_t}$.
	{\IRLSplitNil{} / \IRLSplitCons{}.} Pattern matching on concrete lists: when $\tau_v$ evaluates to $\symbolicNil$, the empty branch $\tau_e$ is executed; when $\tau_v$ evaluates to $\symbolic{v_h} :: \symbolic{v_t}$, variables $\IRvar{h}$ and $\IRvar{t}$ are bound and the non-empty branch $\tau_n$ is executed.
	{\IRLSplitVarEmpty{}.} When $\tau_v$ evaluates to a list descriptor $\descriptiveList{i}{p}{m}{M}{T}$ with offset $i \geq m$, the empty branch can be explored. An end-of-list marker $\EOL{}$ is cached at position $i$ in $\mathcal{S}_{\text{list}}$, terminating this list at length $i$.
	{\IRLSplitVarNonemptyFresh{} / \IRLSplitVarNonemptyCached{}.} When offset $i < M$, the non-empty branch can be explored. If position $i$ is not cached, a fresh symbolic value $\symbolic{v_h}$ is created (via $\buildValue{\argspecification{f}{\theta}{\phi}}$) and cached; if already cached (and not $\EOL{}$), the cached value is reused. The head is bound to $\IRvar{h}$, the tail becomes descriptor $\descriptiveList{i+1}{p}{m}{M}{T}$ with incremented offset, and the non-empty branch $\tau_n$ is executed. Together, these rules enable exploration of all list lengths in $[m, M]$, with the cache ensuring consistent symbolic values across multiple accesses to the same position.
\begin{figure}[t]
	\begin{simulationrules}
		\inferrule*[lab={\IRLNil}]
		{ }
		{\exframe{\IRlist{}}{\exstate{\rho}{\assertionSet{}}{\mathcal{S}}{\mathcal{C}}}{\kappa} \leadsto \kappa(\symbolicNil, \assertionSet{}, \mathcal{S}, \mathcal{C})}
		\inferrule*[lab={\IRLSplitNil}]
		{\exframe{\tau_v}{\sigma}{\idcontinuation{}} \leadsto (\symbolicNil, \assertionSet{}', \mathcal{S}', \mathcal{C}')}
		{\exframe{\IRsplit{\tau_v}{\tau_e}{\IRvar{h}}{\IRvar{t}}{\tau_n}}{\sigma}{\kappa} \leadsto
			\exframe{\tau_e}{\exstate{\rho}{\assertionSet{}'}{\mathcal{S}'}{\mathcal{C}'}}{\kappa}}
		
      \inferrule*[lab={\IRLCons}]
		{\exframe{\tau_1}{\sigma}{\idcontinuation} \leadsto (\symbolic{v_h}, \assertionSet{}_1, \mathcal{S}_1, \mathcal{C}_1) \\\\
			\exframe{\tau_2}{\exstate{\rho}{\assertionSet{}_1}{\mathcal{S}_1}{\mathcal{C}_1}}{\idcontinuation} \leadsto (\symbolic{v_t}, \assertionSet{}', \mathcal{S}', \mathcal{C}')}
		{\exframe{\IRcons{\tau_1}{\tau_2}}{\sigma}{\kappa} \leadsto
			\kappa(\symbolic{v_h} :: \symbolic{v_t}, \assertionSet{}', \mathcal{S}', \mathcal{C}')}
		\inferrule*[lab={\IRLSplitCons}]
		{\exframe{\tau_v}{\sigma}{\idcontinuation{}} \leadsto (\symbolic{v_h} :: \symbolic{v_t}, \assertionSet{}', \mathcal{S}', \mathcal{C}')}
		{\exframe{\IRsplit{\tau_v}{\tau_e}{\IRvar{h}}{\IRvar{t}}{\tau_n}}{\sigma}{\kappa} \leadsto \\\\
			\exframe{\tau_n}{\exstate{\rho[\IRvar{h} \mapsto \symbolic{v_h}, \IRvar{t} \mapsto \symbolic{v_t}]}{\assertionSet{}'}{\mathcal{S}'}{\mathcal{C}'}}{\kappa}}
	\end{simulationrules}
	\begin{simulationrules}
		\inferrule*[lab={\IRLSplitVarEmpty}]
		{\exframe{\tau_v}{\sigma}{\idcontinuation{}} \leadsto (\descriptor{d_l}, \assertionSet{}', (\mathcal{S_{\text{list}}}',\mathcal{S_{\text{tuple}}}',\mathcal{S_{\text{tree}}}'), \mathcal{C}') \\
		\descriptor{d_l} = \descriptiveList{i}{p}{m}{M}{T} \\
		m \leq i \\
		f = \text{concat}(p, i) \\
		f \notin dom(\mathcal{S_{\text{list}}}') \lor \mathcal{S_{\text{list}}}'(f) = \EOL{}}
		{\exframe{\IRsplit{\tau_v}{\tau_e}{\IRvar{h}}{\IRvar{t}}{\tau_n}}{\sigma}{\kappa} \leadsto
		\exframe{\tau_e}{\exstate{\rho}{\assertionSet{}'}{(\mathcal{S_{\text{list}}}'[f \mapsto \EOL{}],\mathcal{S_{\text{tuple}}}',\mathcal{S_{\text{tree}}}')}{\mathcal{C}'}}{\kappa}}
		
    \inferrule*[lab={\IRLSplitVarNonemptyFresh}]
		{\exframe{\tau_v}{\sigma}{\idcontinuation{}} \leadsto (\descriptor{d_l}, \assertionSet{}', (\mathcal{S_{\text{list}}}',\mathcal{S_{\text{tuple}}}',\mathcal{S_{\text{tree}}}'), \mathcal{C}') \\
		\descriptor{d_l} = \descriptiveList{i}{p}{m}{M}{T} \\
		i < M \\
		f=\text{concat}(p, i) \\
		f \notin dom(\mathcal{S_{\text{list}}}') \\
		\symbolic{v_h}=\buildValue{\argspecification{f}{\theta}{\phi}} \\
		\descriptor{d_l'} = \descriptiveList{i+1}{p}{m}{M}{T}}
		{\exframe{\IRsplit{\tau_v}{\tau_e}{\IRvar{h}}{\IRvar{t}}{\tau_n}}{\sigma}{\kappa} \leadsto
		\exframe{\tau_n}{\exstate{\rho[\IRvar{h} \mapsto \symbolic{v_h}, \IRvar{t} \mapsto \descriptor{d_l'}]}{\assertionSet{}'}{(\mathcal{S_{\text{list}}}'[f \mapsto \symbolic{v_h}],\mathcal{S_{\text{tuple}}}',\mathcal{S_{\text{tree}}}')}{\mathcal{C}'}}{\kappa}}
		
    \inferrule*[lab={\IRLSplitVarNonemptyCached}]
		{\exframe{\tau_v}{\sigma}{\idcontinuation{}} \leadsto (\descriptor{d_l}, \assertionSet{}', (\mathcal{S_{\text{list}}}',\mathcal{S_{\text{tuple}}}',\mathcal{S_{\text{tree}}}'), \mathcal{C}') \\
			\descriptor{d_l} = \descriptiveList{i}{p}{m}{M}{T} \\
			i < M \\
			\mathcal{S_{\text{list}}}'(\text{concat}(p,i))=\symbolic{v_h} \\
			\symbolic{v_h} \ne \EOL{} \\
			\descriptor{d_l'} = \descriptiveList{i+1}{p}{m}{M}{T}}
		{\exframe{\IRsplit{\tau_v}{\tau_e}{\IRvar{h}}{\IRvar{t}}{\tau_n}}{\sigma}{\kappa} \leadsto
			\exframe{\tau_n}{\exstate{\rho[\IRvar{h} \mapsto \symbolic{v_h}, \IRvar{t} \mapsto \descriptor{d_l'}]}{\assertionSet{}'}{(\mathcal{S_{\text{list}}}',\mathcal{S_{\text{tuple}}}',\mathcal{S_{\text{tree}}}')}{\mathcal{C}'}}{\kappa}}
	\end{simulationrules}
	\caption{Execution rules for list operations.}
	\label{fig:execution-rules-lists}
\end{figure}
\paragraph{Tuples and Trees.}
The execution rules for tuples and trees follow a similar structure to lists but with fixed arity (for tuples) and branching structure (for trees).
The key distinction is that tuples have a predetermined number of components accessed via \texttt{unpack}, while trees support recursive pattern matching on node and leaf constructors.
\paragraph{Descriptors.}
Algebraic data types and tuples are handled through construction and destructuring rules, which are defined in Figure~\ref{fig:execution-rules-data-structures} and Figure~\ref{fig:execution-rules-tree-data-structures}.
{\IRLPackStep{} / \IRLPackVal{}.} Construction evaluates components left-to-right via \IRLPackStep: each component $\tau_i$ is evaluated to $\symbolic{v_i}$ with updated state, then replaced in the pack expression. Once all components are values, IRLPackVal produces $C(\symbolic{v_1}, \dots, \symbolic{v_n})$ where $C$ is either a tuple constructor or ADT variant tag.
	{\IRLUnpack{}.} Pattern matching on concrete constructed values: scrutineer $\tau_v$ evaluates to $C(\symbolic{v_1}, \dots, \symbolic{v_n})$, variables $\IRvar{x_1}, \dots, \IRvar{x_n}$ are bound to the components, and body $\tau_r$ executes.
	{\IRLUnpackTupleCached{} / \IRLUnpackTupleFresh{}.} When $\tau_v$ evaluates to tuple descriptor $\descriptiveTuple{p}{n}{C}{\vec{T}}$, if pointer $p$ is cached in $\mathcal{S}_{\text{tuple}}'$, retrieve cached values $C(\symbolic{v_1}, \dots, \symbolic{v_n})$ and bind to variables; otherwise, generate fresh symbolic values via $\buildValue{\argspecification{\text{concat}(p, i)}{T_i}{\phi_i}}$ for each component, cache the constructed tuple at $p$, and bind variables.
	{\IRLUnpackTreeLeafCached{} / \IRLUnpackTreeNodeCached{}.} When $\tau_v$ evaluates to tree descriptor $\descriptiveTree{s}{r}{m}{M}{L}{N}{\vec{T}}$ with pointer $p = \text{concat}(r, s)$ cached, retrieve cached value: if $\mathcal{S_{\text{tree}}}'(p) = L()$ execute leaf branch $\tau_L$; if $\mathcal{S_{\text{tree}}}'(p) = N(\symbolic{v_1}, \dots, \symbolic{v_n})$ bind node components and execute $\tau_N$.
	{\IRLUnpackTreeLeafFresh{} / \IRLUnpackTreeNodeFresh{}.} When pointer $p$ is not cached, explore both branches based on node count bounds: if minimum node count $m$ is satisfied, cache $L()$ and execute leaf branch; if maximum node count $M$ is not exceeded, generate fresh node values $\buildValue{\argspecification{\text{concat}(p, i)}{T_i}{\phi_i}}$, cache $N(\symbolic{v_1}, \dots, \symbolic{v_n})$, and execute node branch. This enables exploration of all tree structures within specified bounds.
\begin{figure}[t]
	\begin{simulationrules}
		\inferrule*[lab={\IRLPackStep}]
		{i < n \\
			\exframe{\tau_i}{\sigma}{\idcontinuation{}} \leadsto (\symbolic{v_i}, \assertionSet{}_i, \mathcal{S}_i, \mathcal{C}_i)}
		{\exframe{\IRpack{C}{[\symbolic{v_0}, \dots, \symbolic{v_{i-1}}, \tau_i, \dots, \tau_{n-1}]}}{\sigma}{\kappa} \leadsto
			\exframe{\IRpack{C}{[\symbolic{v_0}, \dots, \symbolic{v_i}, \tau_{i+1}, \dots, \tau_{n-1}]}}{\exstate{\rho}{\assertionSet{}_i}{\mathcal{S}_i}{\mathcal{C}_i}}{\kappa}}
		
    \inferrule*[lab={\IRLPackVal}]
		{\phantom{.}}
		{\exframe{\IRpack{C}{[\symbolic{v_1}, \dots, \symbolic{v_n}]}}{\sigma}{\kappa} \leadsto \\\\
			\kappa(C(\symbolic{v_1}, \dots, \symbolic{v_n}), \assertionSet{}, \mathcal{S}, \mathcal{C})}
		\inferrule*[lab={\IRLUnpack}]
		{\exframe{\tau_v}{\sigma}{\idcontinuation{}} \leadsto (C(\symbolic{v_1}, \dots, \symbolic{v_n}), \assertionSet{}', \mathcal{S}', \mathcal{C}')}
		{\exframe{\IRunpack{\tau_v}{[\dots,(C,[\IRvar{x_1}, \dots, \IRvar{x_n}],\tau_r),\dots]}}{\sigma}{\kappa} \leadsto \\\\
			\exframe{\tau_r}{\exstate{\rho[\IRvar{x_1} \mapsto \symbolic{v_1}, \dots, \IRvar{x_n} \mapsto \symbolic{v_n}]}{\assertionSet{}'}{\mathcal{S}'}{\mathcal{C}'}}{\kappa}}
		
      \inferrule*[lab={\IRLUnpackTupleCached}]
		{\exframe{\tau_v}{\sigma}{\idcontinuation{}} \leadsto (\descriptor{d_t}, \assertionSet{}', (\mathcal{S_{\text{list}}}',\mathcal{S_{\text{tuple}}}',\mathcal{S_{\text{tree}}}'), \mathcal{C}') \\
			\descriptor{d_t} = \descriptiveTuple{p}{n}{C}{\vec{T}} \\
			\mathcal{S_{\text{tuple}}}'(p) = C(\symbolic{v_1}, \dots, \symbolic{v_n}) }
		{\exframe{\IRunpack{\tau_v}{[\dots,(C,[\IRvar{x_1}, \dots, \IRvar{x_n}],\tau_r),\dots]}}{\sigma}{\kappa} \leadsto \\\\
			\exframe{\tau_r}{\exstate{\rho[\IRvar{x_1} \mapsto \symbolic{v_1}, \dots, \IRvar{x_n} \mapsto \symbolic{v_n}]}{\assertionSet{}'}{(\mathcal{S_{\text{list}}}',\mathcal{S_{\text{tuple}}}',\mathcal{S_{\text{tree}}}')}{\mathcal{C}'}}{\kappa}}
		
    \inferrule*[lab={\IRLUnpackTupleFresh}]
		{\exframe{\tau_v}{\sigma}{\idcontinuation{}} \leadsto (\descriptor{d_t}, \assertionSet{}', (\mathcal{S_{\text{list}}}',\mathcal{S_{\text{tuple}}}',\mathcal{S_{\text{tree}}}'), \mathcal{C}') \\
		\descriptor{d_t} = \descriptiveTuple{p}{n}{C}{[T_1, \dots, T_n]} \\
		p \notin dom(\mathcal{S_{\text{tuple}}}') \\
		\symbolic{v_i} = \buildValue{ \argspecification{\text{concat}(p, i)}{\theta_i}{\phi_i} }}
		{\exframe{\IRunpack{\tau_v}{[\dots,(C,[\IRvar{x_1}, \dots, \IRvar{x_n}],\tau_r),\dots]}}{\sigma}{\kappa} \leadsto \\
		\exframe{\tau_r}{\exstate{\rho[\IRvar{x_1} \mapsto \symbolic{v_1}, \dots, \IRvar{x_n} \mapsto \symbolic{v_n}]}{\assertionSet{}'}{(\mathcal{S_{\text{list}}}',\mathcal{S_{\text{tuple}}}'[p \mapsto C(\symbolic{v_1}, \dots, \symbolic{v_n})],\mathcal{S_{\text{tree}}}')}{\mathcal{C}'}}{\kappa}}
	\end{simulationrules}
	\caption{Execution rules for tuples and algebraic data types.}
	\label{fig:execution-rules-data-structures}
\end{figure}
\begin{figure}[t]
	\begin{simulationrules}
		\inferrule*[lab={\IRLUnpackTreeLeafCached}]
		{\exframe{\tau_v}{\sigma}{\idcontinuation{}} \leadsto (\descriptor{d_t}, \assertionSet{}', (\mathcal{S_{\text{list}}}',\mathcal{S_{\text{tuple}}}',\mathcal{S_{\text{tree}}}'), \mathcal{C}') \\
			\descriptor{d_t} = \descriptiveTree{s}{r}{m}{M}{L}{N}{[T_1, \dots, T_n]} \\
			p = \text{concat}(r, s) \\
			\mathcal{S_{\text{tree}}}'(p) = L()}
		{\exframe{\IRunpack{\tau_v}{[\dots,(L,[],\tau_L),\dots,(N,[\IRvar{x_1}, \dots, \IRvar{x_n}],\tau_N),\dots]}}{\sigma}{\kappa} \leadsto \\
			\exframe{\tau_L}{\exstate{\rho}{\assertionSet{}'}{(\mathcal{S_{\text{list}}}',\mathcal{S_{\text{tuple}}}',\mathcal{S_{\text{tree}}}')}{\mathcal{C}'}}{\kappa}}
		
    \inferrule*[lab={\IRLUnpackTreeNodeCached}]
		{\exframe{\tau_v}{\sigma}{\idcontinuation{}} \leadsto (\descriptor{d_t}, \assertionSet{}', (\mathcal{S_{\text{list}}}',\mathcal{S_{\text{tuple}}}',\mathcal{S_{\text{tree}}}'), \mathcal{C}') \\
			\descriptor{d_t} = \descriptiveTree{s}{r}{m}{M}{L}{N}{[T_1, \dots, T_n]} \\
			p = \text{concat}(r, s) \\
			\mathcal{S_{\text{tree}}}'(p) = N(\symbolic{v_1}, \dots, \symbolic{v_n})}
		{\exframe{\IRunpack{\tau_v}{[\dots,(L,[],\tau_L),\dots,(N,[\IRvar{x_1}, \dots, \IRvar{x_n}],\tau_N),\dots]}}{\sigma}{\kappa} \leadsto \\
			\exframe{\tau_N}{\exstate{\rho[\IRvar{x_1} \mapsto \symbolic{v_1}, \dots, \IRvar{x_n} \mapsto \symbolic{v_n}]}{\assertionSet{}'}{(\mathcal{S_{\text{list}}}',\mathcal{S_{\text{tuple}}}',\mathcal{S_{\text{tree}}}')}{\mathcal{C}'}}{\kappa}}
		
    \inferrule*[lab={\IRLUnpackTreeLeafFresh}]
		{\exframe{\tau_v}{\sigma}{\idcontinuation{}} \leadsto (\descriptor{d_t}, \assertionSet{}', (\mathcal{S_{\text{list}}}',\mathcal{S_{\text{tuple}}}',\mathcal{S_{\text{tree}}}'), \mathcal{C}') \\
		\descriptor{d_t} = \descriptiveTree{s}{r}{m}{M}{L}{N}{[T_1, \dots, T_n]} \\
		p = \text{concat}(r, s) \\
		p \notin dom(\mathcal{S_{\text{tree}}}') \\
		m \leq \#\{x | \mathcal{S_{\text{tree}}}'(\text{concat(r, x)}) = N(\dots)\}
		}
		{\exframe{\IRunpack{\tau_v}{[\dots,(L,[],\tau_L),\dots,(N,[\IRvar{x_1}, \dots, \IRvar{x_n}],\tau_N),\dots]}}{\sigma}{\kappa} \leadsto \\
		\exframe{\tau_L}{\exstate{\rho}{\assertionSet{}'}{(\mathcal{S_{\text{list}}}',\mathcal{S_{\text{tuple}}}',\mathcal{S_{\text{tree}}}'[p \mapsto L()])}{\mathcal{C}'}}{\kappa}}
		
    \inferrule*[lab={\IRLUnpackTreeNodeFresh}]
		{\exframe{\tau_v}{\sigma}{\idcontinuation{}} \leadsto (\descriptor{d_t}, \assertionSet{}', (\mathcal{S_{\text{list}}}',\mathcal{S_{\text{tuple}}}',\mathcal{S_{\text{tree}}}'), \mathcal{C}') \\
		\descriptor{d_t} = \descriptiveTree{s}{r}{m}{M}{L}{N}{[T_1, \dots, T_n]} \\
		p = \text{concat}(r, s) \\
		p \notin dom(\mathcal{S_{\text{tree}}}') \\
		\#\{x | \mathcal{S_{\text{tree}}}'(\text{concat(r, x)}) = N(\dots)\} \leq M \\
		\symbolic{v_i} = \buildValue{\argspecification{\text{concat}(p, i)}{\theta_i}{\phi_i}}
		}
		{\exframe{\IRunpack{\tau_v}{[\dots,(L,[],\tau_L),\dots,(N,[\IRvar{x_1}, \dots, \IRvar{x_n}],\tau_N),\dots]}}{\sigma}{\kappa} \leadsto \\
		\exframe{\tau_N}{\exstate{\rho[\IRvar{x_1} \mapsto \symbolic{v_1}, \dots, \IRvar{x_n} \mapsto \symbolic{v_n}]}{\assertionSet{}'}{(\mathcal{S_{\text{list}}}',\mathcal{S_{\text{tuple}}}',\mathcal{S_{\text{tree}}}'[p \mapsto N(\symbolic{v_1},\dots\,\symbolic{v_n})])}{\mathcal{C}'}}{\kappa}}
	\end{simulationrules}
	\caption{Execution rules for tree data structures.}
	\label{fig:execution-rules-tree-data-structures}
\end{figure}
\section{Implementation Details}
\label{appendix:implementation}
This section provides detailed implementation information for the three main components of our toolchain: (i) the OCaml-to-core-language translator, which converts source programs into our intermediate representation; (ii) the symbolic execution engine, which explores all feasible execution paths; and (iii) the statistical analysis engine, which uses mixed integer linear programming to infer polynomial complexity bounds from the collected observations.
\subsection{Implementation of the OCaml Translation}\label{section:tool:translation}
To automate the analysis of existing complexity example programs, we implemented a translation component from OCaml to our core language.
Figure~\ref{fig:IRtransformationInference} defines the transformation $\trcore{e} = \tau$ that maps an OCaml expression $e$ to its core language AST $\tau$. The definition is inductive on the structure of expressions, covering literals, variables, function definitions and applications, control flow constructs, data structures, and operations. Each OCaml construct is systematically translated into a corresponding core language representation, preserving the semantics while enabling uniform analysis in the core language.
Note that while the transformation focuses on a subset of OCaml constructs relevant to our examples, it can be extended to cover additional language features as needed for broader applicability.
\begin{figure}[t]
	\scriptsize
	\[
		\begin{array}{c}
			\inferrule
			{n \in \mathbb{Z}}
			{n \Rightarrow \IRconst{n}}
			\qquad
			\inferrule
			{b \in \{\texttt{true},\texttt{false}\}}
			{b \Rightarrow b}
			\qquad
			\inferrule
			{x \in \mathcal{V}}
			{x \Rightarrow \IRvar{x}}
			\qquad
			\inferrule
			{f \in \mathcal{F}}
			{f \Rightarrow \IRfunName{f}}
			\\[1.2em]
			\inferrule
			{\oplus_1 \in \{-, \sqrt{}, \sqrt[3]{}, \lnot\} \\
				e_1 \Rightarrow \tau_1}
			{\oplus_1 e_1
				\Rightarrow
				\IRUnaryOp{\oplus_1}{\tau_1}}
			\\[1.2em]
			\inferrule
			{\oplus_2 \in \{+, -, \dots, \texttt{mod}\}     \\
			e_1 \Rightarrow \tau_1                          \\
				e_2 \Rightarrow \tau_2}
			{e_1 \oplus_2 e_2
				\Rightarrow
				\IROp{\tau_1}{\oplus_2}{\tau_2}}
			\\[1.2em]
			\inferrule
			{c \in \mathbb{R}                               \\
				e \Rightarrow \tau}
			{\texttt{let \_ = tick}(c)\texttt{ in }e
				\Rightarrow
				\IRtick{c}{\tau}}
			\qquad
			\inferrule
			{m \in \texttt{String}}
			{\texttt{raise } m
				\Rightarrow
				\IRraise{m}}
			\qquad
			\inferrule
			{\trcore{e_1} \Rightarrow \tau_1                \\
				\trcore{e_2} \Rightarrow \tau_2}
			{\trcore{e_1; e_2} \Rightarrow
				\trcore{\text{let \_ = } e_1 \text{ in } e_2}}
			\\[1.2em]
			\inferrule
			{\trcore{e_c} \Rightarrow \tau_c                \\
			\trcore{e_1} \Rightarrow \tau_1                 \\
				\trcore{e_2} \Rightarrow \tau_2}
			{\trcore{\text{if } e_c \text{ then } e_1 \text{ else } e_2}
				\Rightarrow
				\IRifelse{\tau_c}{\tau_1}{\tau_2}}
			\\[1.2em]
			\inferrule
			{\trcore{e_1} \Rightarrow \tau_1                \\
				\trcore{e_2} \Rightarrow \tau_2}
			{\trcore{\text{let } x = e_1 \text{ in } e_2}
				\Rightarrow
				\IRlet{x}{\tau_1}{\tau_2}}
			\\[1.2em]
			\inferrule
			{\trcore{e_b} \Rightarrow \tau_b}
			{\trcore{\text{let [rec] } f\ x_1\dots x_n = e_b}
				\Rightarrow
				\IRfun{[x_1,\dots,x_n]}{\tau_b}}
			\\[1.2em]
			\inferrule
			{\trcore{e_b} \Rightarrow \tau_b                \\
				\trcore{e} \Rightarrow \tau}
			{\trcore{\text{let [rec] } f\ x_1\dots x_n = e_b \text{ in } e}
				\Rightarrow
				\IRlet{f}{\IRfun{[x_1,\dots,x_n]}{\tau_b}}{\tau}}
			\\[1.2em]
			\inferrule
			{\trcore{e_f} \Rightarrow \tau_f                \\
			\trcore{e_{a_1}} \Rightarrow \tau_1             \\
			\cdots                                          \\
				\trcore{e_{a_n}} \Rightarrow \tau_n}
			{\trcore{e_f\ e_{a_1}\dots e_{a_n}}
				\Rightarrow
				\IRcall{\tau_f}{[\tau_1,\dots,\tau_n]}}
			\\[1.2em]
			\inferrule
			{\trcore{e_1} \Rightarrow \tau_1                \\
				\trcore{e_2} \Rightarrow \tau_2}
			{\trcore{e_1 :: e_2}
				\Rightarrow
				\IRcons{\tau_1}{\tau_2}}
			\\[1.2em]
			\inferrule
			{\trcore{e_v} \Rightarrow \tau_v                \\
			\trcore{e_1} \Rightarrow \tau_1                 \\
				\trcore{e_2} \Rightarrow \tau_2}
			{\trcore{\text{match } e_v \text{ with }
					\mid [] \to e_1
					\mid h::t \to e_2}
				\Rightarrow
				\IRsplit{\tau_v}{\tau_1}{h}{t}{\tau_2}}
			\\[1.2em]
			\inferrule
			{\trcore{e_1} \Rightarrow \tau_1                \\
			\cdots                                          \\
				\trcore{e_n} \Rightarrow \tau_n}
			{\trcore{C(e_1,\dots,e_n)}
				\Rightarrow
				\IRpack{C}{[\tau_1,\dots,\tau_n]}}
			\\[1.2em]
			\inferrule
			{\trcore{e_v} \Rightarrow \tau_v                \\
			\trcore{e_1} \Rightarrow \tau_1                 \\
			\cdots                                          \\
				\trcore{e_n} \Rightarrow \tau_n}
			{\trcore{\text{match } e_v \text{ with }
					\mid C_1(\vec{x}_1)\to e_1
					\mid \dots
					\mid C_n(\vec{x}_n)\to e_n}
				\Rightarrow
				\IRunpack{\tau_v}
				{[(C_1,\vec{x}_1,\tau_1),\dots,(C_n,\vec{x}_n,\tau_n)]}}
			\\[1.2em]
			\inferrule
			{\trcore{e_1} \Rightarrow \tau_1}
			{\trcore{\oplus_1 e_1}
				\Rightarrow
				\IRUnaryOp{\oplus_1}{\tau_1}}
			\qquad
			\inferrule
			{\trcore{e_1} \Rightarrow \tau_1                \\
				\trcore{e_2} \Rightarrow \tau_2}
			{\trcore{e_1 \oplus_2 e_2}
				\Rightarrow
				\IROp{\tau_1}{\oplus_2}{\tau_2}}
			\\[1.2em]
			\inferrule
			{\trcore{e} \Rightarrow \tau}
			{\trcore{\text{let \_ = tick}(c)\text{ in }e}
				\Rightarrow
				\IRtick{c}{\tau}}
			\qquad
			\inferrule
			{}
			{\trcore{\text{raise } m}
				\Rightarrow
				\IRraise{m}}
		\end{array}
	\]
	\caption{Translation relation from OCaml expressions to the core IR.}
	\label{fig:IRtransformationInference}
\end{figure}
The translator is implemented in OCaml and leverages the language's native parser (the \texttt{Parsetree} module from the OCaml compiler) to obtain the abstract syntax tree (AST) of the input program. The translated core language AST is then serialised to JSON using the Yojson library, enabling interoperability with the JavaScript-based symbolic execution engine.
Core language terms ($\tau$) represent the static program structure as defined by the grammar in Figure~\ref{fig:IRsyntax}.
The translation proceeds by structural recursion over the OCaml AST. Most OCaml constructs have direct correspondences in our core language:
\paragraph{Operators.} Binary arithmetic operators (\code{+}, \code{-}, \code{*}, \code{/}) and comparison operators (\code{<}, \code{<=}, \code{=}, \code{!=}, \code{>=}, \code{>}) are translated directly to their core language equivalents. Some operators require renaming: \code{mod} becomes \code{\%}, \code{\&\&} becomes \code{\&}, and \code{||} becomes \code{|}. Unary operators are similarly handled, with \code{not} translated to \code{!} and \code{sqrt} preserved. The translator also recognises special function applications like \code{Raml.tick(n)} and extracts them as explicit cost annotations in the core language representation.
\paragraph{Functions.} Function definitions are handled by traversing \texttt{Pexp\_fun} chains to extract the complete parameter list and function body. Recursive functions (marked with the \texttt{rec} flag in OCaml's AST) are detected and appropriately represented. Functions may be defined at the top level or within \code{let} bindings, and both cases are handled uniformly. The translator produces a core language function definition with an explicit parameter list and body.
\paragraph{Let Bindings.} Let expressions (\texttt{Pexp\_let}) are translated to \code{let} operations in the core language. The translator handles both single and multiple sequential bindings, threading the environment through each binding. Pattern-based let bindings (e.g., \code{let (x, y) = ...}) are desugared into \code{unpack} operations. Anonymous let bindings (\code{let \_ = ...}) that contain tick operations are recognised and translated to explicit cost annotations.
\paragraph{Pattern Matching.} The OCaml \code{match} construct requires special handling due to its flexibility. The translator analyses the pattern structure of each match case and generates different core language operations depending on the pattern type:
\begin{itemize}
	\item \textbf{List patterns:} A match on list patterns (empty list \code{[]} vs. cons pattern \code{h::t}) is translated to a \code{split} operation. The \code{split} construct includes an \code{ifEmpty} branch for the empty list case, and an \code{else} branch for the non-empty case that binds the head element to variable \code{h} and the tail to variable \code{t}.
	\item \textbf{Tuple patterns:} A match on a tuple pattern is translated to an \code{unpack} operation. The tuple components are bound to the pattern variables specified in the match arm, and execution continues with the body expression.
	\item \textbf{Algebraic data type (ADT) patterns:} A match on variant constructors is also translated to \code{unpack}. Each constructor pattern becomes a separate case in the core language, with the variant's fields bound to the corresponding pattern variables. When multiple variant constructors are matched, the translator generates a multi-branch \code{unpack} that associates each constructor name with its corresponding binding pattern and body.
\end{itemize}
\paragraph{Conditionals and Sequences.} OCaml's \code{if-then-else} expressions are translated to \code{ifelse} operations in the core language. The translator ensures that both branches are present (OCaml allows omitting the else branch in certain contexts, which is not supported by our core language for symbolic execution). Sequential composition (\texttt{Pexp\_sequence}) is translated to nested \code{let} bindings, where the first expression is bound to an anonymous variable before evaluating the second. If the first expression is a tick operation, it is extracted and inserted as an explicit cost annotation.
\paragraph{Data Constructors.} OCaml's data constructors are translated to \code{pack} operations. The translator handles list constructors (\code{[]}, \code{::}), tuple constructors, unit \code{()}, booleans (\code{true}, \code{false}), and user-defined variant constructors. Each constructor is tagged with its name (for variants) or a structural marker (for built-in types), and the component values are provided as a list.
The translator is designed to be conservative: constructs that are not yet supported (such as polymorphic pattern matching that mixes different types in a single match, or advanced OCaml features like first-class modules) are detected and reported with appropriate error messages. The current implementation covers the subset of OCaml commonly used in algorithm implementations and complexity benchmarks. Future extensions to handle additional OCaml features can be incorporated incrementally as needed.
The transformation of the OCaml program \texttt{dyade} into the core language is illustrated in Figure~\ref{fig:dyade-translation}.
\begin{figure}[tb]
	\begin{inferencetree}[\scriptsize]
		\begin{prooftree}
			\AxiomC{\raml{n} $\in \mathbb{Z}$}
			\UnaryInfC{\IRinfer{n}{\IRconst{n}}}
			\AxiomC{\raml{x} $\in \mathcal{V}$}
			\UnaryInfC{\IRinfer{x}{\IRvar{x}}}
			\AxiomC{\raml{*} $\in \{+,*, \dots\}$}
			\TrinaryInfC{\IRinfer{n * x}{\IROp{\IRconst{n}}{*}{\IRvar{x}}}}
			\AxiomC{\raml{n} $\in \mathbb{Z}$}
			\UnaryInfC{\IRinfer{n}{\IRconst{n}}}
			\AxiomC{\raml{xs} $\in \mathcal{V}$}
			\UnaryInfC{\IRinfer{xs}{\IRvar{xs}}}
			\AxiomC{\raml{mult} $\in \mathcal{F}$}
			\UnaryInfC{\IRinfer{mult}{\IRfunName{mult}}}
			\TrinaryInfC{\IRinfer{mult\ n\ xs}{\IRcall{\IRfunName{mult}}{\IRlist{\IRconst{n}, \IRvar{xs}}}}}
			\BinaryInfC{\IRinfer{n * x\ ::\ mult\ n\ xs}{\IRcons{\IROp{\IRconst{n}}{*}{\IRvar{x}}}
					{\IRcall{\IRfunName{mult}}{\IRlist{\IRconst{n}, \IRvar{xs}}}}}}
			\UnaryInfC{\large $\diamond$}
		\end{prooftree}
		\begin{prooftree}
			\AxiomC{\raml{l} $\in \mathcal{V}$}
			\UnaryInfC{\IRinfer{l}{\IRvar{l}}}
			\AxiomC{}
			\UnaryInfC{\IRinfer{[]}{\IRlist{}}}
			\AxiomC{$1 \in \mathbb{R}$}
			\AxiomC{\large $\diamond$}
			\BinaryInfC{\IRinfer{let\ \_\ tick(1)\ in\ n * x\ ::\ mult\ n\ xs}{\tau_{2}}}
			\UnaryInfC{\IRinfer{tick(1);\ n * x\ ::\ mult\ n\ xs}{\tau_{2}}}
			\TrinaryInfC{\IRinfer{match\ l\ with\ |\ [] -> []\ |\ x :: xs\ ->\ tick(1);\ n * x\ ::\ mult\ n\ xs}{\tau_{\text{mult}}}}
			\UnaryInfC{\IRinfer{P_{\text{mult}}}{\IRfun{\IRlist{\IRconst{n}, \IRvar{l}}}{\tau_{\text{mult}}}}}
		\end{prooftree}
		\begin{align*}
			\tau_{2}           & =
			\IRtick{1}{\IRcons{\IROp{\IRconst{n}}{*}{\IRvar{x}}}{\IRcall{\IRfunName{mult}}{\IRlist{\IRconst{n}, \IRvar{xs}}}}} \\
			\tau_{\text{mult}} & = \IRsplit{\IRvar{l}}{\IRlist{}}{\IRvar{x}}{\IRvar{xs}}{\tau_{2}}                             \\
			P_{\text{mult}}    & = \begin{aligned}[t]
				                        & \raml{let rec mult n l = match l with}                \\
				                        & \quad \raml{| [] -> []}                               \\
				                        & \quad \raml{| x :: xs -> tick(1); n * x :: mult n xs}
			                       \end{aligned}
		\end{align*}
		\begin{prooftree}
			\AxiomC{\raml{x} $\in \mathcal{V}$}
			\UnaryInfC{\IRinfer{{x}}{\IRvar{x}}}
			\AxiomC{\raml{l2} $\in \mathcal{V}$}
			\UnaryInfC{\IRinfer{{l2}}{\IRvar{l_2}}}
			\AxiomC{\raml{mult} $\in \mathcal{F}$}
			\UnaryInfC{\IRinfer{mult}{\IRfunName{mult}}}
			\TrinaryInfC{\IRinfer{{mult\ x\ l2}}{\IRcall{\IRfunName{mult}}{\IRlist{\IRvar{x}, \IRvar{l_2}}}}}
			\AxiomC{\raml{x} $\in \mathcal{V}$}
			\UnaryInfC{\IRinfer{{x}}{\IRvar{x}}}
			\AxiomC{\raml{l2} $\in \mathcal{V}$}
			\UnaryInfC{\IRinfer{{l2}}{\IRvar{l_2}}}
			\AxiomC{\raml{dyade} $\in \mathcal{F}$}
			\UnaryInfC{\IRinfer{dyade}{\IRfunName{dyade}}}
			\TrinaryInfC{\IRinfer{{dyade\ xs\ l2}}{\IRcall{\IRfunName{dyade}}{\IRlist{\IRvar{xs}, \IRvar{l_2}}}}}
			\BinaryInfC{\IRinfer{{mult\ x\ l2\ ::\ dyade\ xs\ l2}}{\IRcons{\IRcall{\IRfunName{mult}}{\IRlist{\IRvar{x}, \IRvar{l_2}}}}{\IRcall{\IRfunName{dyade}}{\IRlist{\IRvar{xs}, \IRvar{l_2}}}}}}
			\UnaryInfC{\large $\diamond$}
		\end{prooftree}
		\begin{prooftree}
			\AxiomC{\raml{l1} $\in \mathcal{V}$}
			\UnaryInfC{\IRinfer{{l1}}{\IRvar{l_1}}}
			\AxiomC{}
			\UnaryInfC{\IRinfer{{[]}}{\IRlist{}}}
			\AxiomC{$1 \in \mathbb{R}$}
			\AxiomC{\large $\diamond$}
			\BinaryInfC{\IRinfer{let\ \_\ tick(1)\ in\ \ mult\ x\ l2\ ::\ dyade\ xs\ l2}{\tau_{1}}}
			\UnaryInfC{\IRinfer{tick(1);\ mult\ x\ l2\ ::\ dyade\ xs\ l2}{\tau_{1}}}
			\TrinaryInfC{\IRinfer{match\ l1\ with\ |\ [] -> []\ |\ x :: xs\ ->\ tick(1);\ mult\ x\ l2\ ::\ dyade\ xs\ l2}{\tau_{\text{dyade}}}}
			\UnaryInfC{\IRinfer{P_{\text{dyade}}}{\IRfun{\IRlist{\IRvar{l_1}, \IRvar{l_2}}}{\tau_{\text{dyade}}}}}
		\end{prooftree}
		\begin{align*}
			\tau_{1}            & = \IRtick{1}{\IRcons{\IROp{\IRconst{n}}{*}{\IRvar{x}}}{\IRcall{\IRfunName{mult}}{\IRlist{\IRconst{n}, \IRvar{xs}}}}} \\
			\tau_{\text{dyade}} & = \IRsplit{\IRvar{l_1}}{\IRlist{}}{\IRvar{x}}{\IRvar{xs}}{\tau_{1}}                                                  \\
			P_{\text{dyade}}    & = \begin{aligned}[t]
				                         & \raml{let rec dyade l1 l2 = match l1 with}                  \\
				                         & \quad \raml{| [] -> []}                                     \\
				                         & \quad \raml{| x :: xs -> tick(1); mult x l2 :: dyade xs l2}
			                        \end{aligned}
		\end{align*}
	\end{inferencetree}
	\caption{The transformation of the OCaml program \texttt{dyade} into the core language.}
	\label{fig:dyade-translation}
\end{figure}
\subsection{Implementation of the Symbolic Execution Engine}
\label{section:tool:symbolic}
Once the program has been translated to the intermediary language --- or been given directly in that format as a JSON file --- symbolic execution starts.
As mentioned before, this engine also requires information about the main function, its arguments and other additional assertions, also in the form of a JSON file.
The engine itself is written in JSDoc-annotated, TypeScript-checked JavaScript, and uses the WebAssembly version of Z3~\cite{DeMoura:2008} to solve those assertions, which is the reason why input assertions are given in an extended SMT-LIB format.
This format allows referencing input arguments and their properties, like \code{length} or \code{nodes}, and also invoking arbitrary JavaScript code.
For example, if the program input are two integers named \code{x} and \code{y}, the assertion \code{(>= \$\{\$('x')\} \$\{\$('y')\})} makes sure that only inputs where $x>=y$ are considered.
The outer \code{\$} is equal to string interpolation used in JavaScript's template literals, the inner \code{\$} is a predefined function that allows referencing input arguments.
Input arguments and their nested arguments can also be annotated with assertions.
In the implementation, the current input pointer is exposed to assertions through the special variable \code{path}.
Concretely, \code{path} is a JavaScript array containing all pointer segments up to the current argument.
For example, if a program takes as input a list of integers called \code{myList}, and those integers are annotated, \code{path} would be \code{['myList', 5]} for the fifth element of the list.
This allows to further constraint the input, e.g. to a sorted list, by annotating the list's integer items with the assertion
\[
	\code{(< \$\{\$(path)\} \$\{\$('myList',path.at(-1)+1)\})}.
\]
Every time a new item at offset $n$ gets referenced, it is ensured that $myList[n] < myList[n+1]$.
Using assertions, expert knowledge about the shape of the input data can be incorporated.
The output of the symbolic execution engine is another JSON file, containing an array of all possible executions that satisfy the assertions.
Each execution contains the cost, all gathered assertions, the resulting value and the model of the input arguments.
Those results are then grouped as explained in Section~\ref{section:symbolic-execution:grouping}, and then passed on to the statistical engine.
\paragraph{Work-List Algorithm and Frame-Based Execution.}
The symbolic execution engine implements a work-list algorithm using a frame-based execution model.
Each frame represents a pending computation and contains: (1) the \code{command} to execute (e.g., \code{ifelse}, \code{split}, \code{let}), (2) a \code{variables} map binding names to symbolic values, (3) a \code{state} object tracking assertions and the evolution of symbolic data structures, (4) a \code{continuation} callback to process the result, and (5) an optional \code{upperFrame} reference to the parent frame.
In the implementation, input pointers are represented by the type \code{Path}.
The main execution loop iterates over an array \code{pendingFrames}, processing each frame by calling \code{handleFrame}, which dispatches to the appropriate handler based on the command's operation type.
This approach avoids explicit stack management while maintaining the compositional structure of the semantics.
Initially, the entry frame for the main function is added to \code{pendingFrames}.
As execution proceeds, handlers push new frames onto this array, effectively exploring all feasible execution paths.
\paragraph{Symbolic Value Construction.}
Input arguments are transformed into symbolic values by the \code{buildValue} function.
For integer arguments without concrete values, Z3 integer constants are created using \code{Z3.Int.const(name)}, where \code{name} is derived from the argument's pointer (e.g., \code{"myList.3"} for the third element of list \code{myList}).
Similarly, boolean arguments become Z3 boolean constants via \code{Z3.Bool.const(name)}.
When minimum or maximum constraints are specified on integers, corresponding assertions (e.g., \code{constraint.ge(minValue)}) are added to the symbolic value's assertion set.
Lists are represented as one of three types: \code{Fixed} (concrete length with item templates), \code{Templated} (concrete length with uniform item template), or \code{Variable} (symbolic length with bounds).
Variable-length lists track their current offset and lazily instantiate elements as they are accessed.
Tuples and trees follow similar patterns, recursively building nested symbolic structures.
Custom assertions from the input specification are parsed and converted to Z3 expressions by \code{buildAssertions}, which evaluates assertion templates in an environment providing \code{\$} for selector-based argument references.
\paragraph{State Management and Path Constraints.}
The \code{state} object threads through execution and contains four components: \code{assertions} (a set of Z3 boolean expressions representing path constraints), \code{lists} (a map from list identifiers to their partially-decided values and completion status), \code{tuples} (a map from tuple identifiers to value arrays), and \code{trees} (a map from tree identifiers to node structures).
When a list is split for the first time at a symbolic position, an empty values array is created.
As execution proceeds and more elements are accessed, new symbolic values are instantiated from the item template and added to the array.
The \code{isComplete} flag indicates whether the list's length has been decided.
Similarly, trees track node and leaf counts, updating them as the structure is explored.
Assertions accumulate monotonically: each branching operation extends the assertion set with new constraints specific to that path.
\paragraph{Path Splitting in Conditionals.}
The \code{handleIfElse} function implements path splitting for conditional branches.
First, it evaluates the condition expression, yielding a symbolic boolean value containing a Z3 expression.
This expression is simplified using Z3's built-in simplification.
If the result is the constant \code{True} or \code{False}, only the corresponding branch is executed.
Otherwise, the engine performs a satisfiability check: a Z3 solver instance is created, the current assertions and the condition are added, and \code{check()} is called.
If the result is \code{sat}, a new frame for the \code{if} branch is pushed with the condition added to the assertions.
The same process is repeated for the negated condition and the \code{else} branch.
This results in exploring both paths when both are feasible, effectively doubling the number of execution paths at each symbolic branch.
\paragraph{Path Splitting in List Operations.}
Symbolic list splitting is handled by the \code{handleSplit} function.
After evaluating the list expression, the handler distinguishes between concrete and symbolic cases.
For concrete lists (\code{Nil} or \code{Cons}), the appropriate branch is chosen deterministically.
For symbolic lists of type \code{Variable}, the handler checks the current offset against the cached values.
If the offset is less than the number of cached values, the head is retrieved and the non-empty branch is executed.
If the offset equals the number of cached values and the list is not yet marked complete, a choice point arises.
The handler pushes two frames: one for the empty branch (updating the state to mark the list complete) and one for the non-empty branch (instantiating a new symbolic element from the item template and extending the values array).
This branching continues until minimum and maximum length constraints, if present, are satisfied.
\paragraph{Continuation Handling and Result Propagation.}
Each handler defines continuations that process results and propagate them upward.
For example, \code{handleLet} first pushes a frame to evaluate the bound expression with a continuation that receives the result, extends the variable environment with the binding, and pushes a new frame for the body.
The continuation for the body simply forwards the result with accumulated costs to the upper frame.
Costs are tracked explicitly: each command may specify a cost increment via \code{getCosts(command)}, and results carry a \code{costs} field that accumulates these increments.
The continuation pattern ensures that state, costs, and values flow correctly through nested computations, mirroring the semantic rules' premise-conclusion structure.
\paragraph{Termination and Output Generation.}
Execution terminates when all frames in \code{pendingFrames} have been processed.
For each complete execution path that reaches a return point, the engine invokes the entry frame's continuation with the final result.
This result includes the returned value, total cost, final state, and the complete set of path constraints (assertions).
The engine then calls Z3's solver to compute a satisfying model for the assertions, providing concrete values for all symbolic inputs consistent with that path.
The final output is a JSON array where each element represents one feasible execution, containing the cost, serialised assertions, return value, and a model showing example input values.
This output is subsequently grouped by cost and input characteristics before being passed to the statistical analysis component.
\subsection{Implementation of the Statistical Analysis Engine}
\label{section:tool:statistical}
We have implemented the linear programming approach outlined in Section~\ref{section:linear-programming} using the \texttt{PuLP} library~\cite{Mitchell:etal:2011}, which provides a convenient interface to define and solve linear programs in Python.
For solving the linear programs, we use the open-source solver \texttt{CBC}~\cite{Forrest:etal:2024}, which is included with \texttt{PuLP}.
The statistical analysis engine takes as input the JSON output from the symbolic execution engine (containing costs and input shapes) and produces upper bounds through a template-based MILP solving approach.
\paragraph{Data Preprocessing and Grouping.}
The engine first processes the symbolic execution results to extract data points $(\vec{x}_i, c_i)$ where $\vec{x}_i$ represents the input shape (e.g., list lengths, tree node counts) and $c_i$ is the cost.
When multiple executions produce the same input shape, the maximum cost is retained to ensure worst-case coverage.
The implementation supports both pre-grouped data (where results are already organised by input characteristics) and raw symbolic execution output, which is automatically converted by extracting variable names and values from the Z3 models.
The resulting data is organised into a DataFrame with columns for each input variable and a \code{cost} column, sorted by input variables in ascending order to facilitate monotonicity checking.
\paragraph{Basis Functions and Feature Construction.}
The engine supports a library of basis functions that can be composed to form features.
The available basis functions include: \code{id} (identity $x$), \code{square} ($x^2$), \code{cube} ($x^3$), \code{quartic} ($x^4$), \code{exp} ($2^x$), \code{log} ($\log_2(x)$), and \code{idlog} ($x \log_2(x)$).
Each basis function is implemented as a lambda function paired with a string formatter for generating human-readable representations.
Given input data matrix $X \in \mathbb{R}^{m \times n}$ (where $m$ is the number of observations and $n$ the number of input variables), the \code{make\_poly\_features} function constructs a feature matrix $\Phi$ according to a template specification.
The template specifies which basis functions to apply via the \code{base} field (univariate terms) and how to combine them via the \code{combine} field (interaction terms).
For example, a template with \code{base: ["id", "square"]} and \code{combine: [[["id"], ["id"]]]} produces features $[1, x, y, x^2, y^2, x \cdot y]$ for two variables.
The feature matrix construction begins with a bias column of ones, applies each base function to each variable column, and then computes products for the specified combinations.
The function returns both the feature matrix $\Phi$ and a parallel list of exponent descriptors tracking which basis functions contributed to each feature.
\paragraph{Template-Based Bound Search.}
Bound templates are defined in a separate JSON file, organised by number of variables and bound degree.
This design enables modularity: new forms can be added without modifying the solver code.
Each template entry specifies the \code{base} functions, \code{combine} interaction patterns, and a human-readable description.
A quadratic template for two variables, for example, includes linear and quadratic terms as well as the interaction term $x \cdot y$.
The engine iterates over degrees (typically $1$ through $4$) and within each degree tries all available templates.
This idea is given Algorithm~\ref{alg:best-bound}.
An optional degree estimation heuristic based on divided differences can be used to determine a minimum starting degree.
The \code{get\_min\_degree} function computes monotonic convex data points, builds a divided difference table, and applies a correlation-based trend analysis to estimate the degree from the data structure.
This heuristic reduces search time by skipping degrees that are unlikely to fit the observed cost pattern.
\begin{algorithm}[t]
	\caption{Find Best Bound}\label{alg:best-bound}
	\begin{algorithmic}[1]
		\Require Set of potential degrees $\mathcal{D}$, objective functions $\mathcal{O}$, costs $\mathcal{M} = (X, y)$
		\Ensure Best fit and corresponding cost
		\State $\textit{best\_fit} \gets \text{None}$
		\State $\textit{best\_cost} \gets \infty$
		\For{$d \in \mathcal{D}$}
		\For{\textbf{each} $\zeta \in \mathcal{O}[d]$}
		\State $\Phi \gets \textsc{MakePolyFeatures}(X, \zeta)$
		\State $\textit{fit} \gets \textsc{SolveProblem}(\Phi, y, \textit{best\_cost})$
		\If{$\textit{fit}$}
		\State $\textit{best\_fit} \gets \textit{fit.bound}$
		\State $\textit{best\_cost} \gets \textit{fit.cost}$
		\EndIf
		\EndFor
		\EndFor
		\State \Return $(\textit{best\_fit}, \textit{best\_cost})$
	\end{algorithmic}
\end{algorithm}
\paragraph{MILP Formulation for Sound Upper Bounds.}
The core bound fitting is performed by \code{solve\_problem}, which constructs and solves a Mixed Integer Linear Program.
For a feature matrix $\Phi \in \mathbb{R}^{m \times k}$ (where $k$ is the number of features) and cost vector $y \in \mathbb{R}^m$, the formulation creates:
\begin{itemize}
	\item \textbf{Coefficient variables}: One non-negative LP variable $w_j \geq 0$ for each feature $j$, named according to the term it represents (e.g., \code{x * log(x)}).
	\item \textbf{Binary soundness variables}: One binary variable $z_i \in \{0, 1\}$ for each data point $i$, used to encode the soundness requirement that residuals be non-positive.
\end{itemize}
The fundamental soundness constraint ensures that the upper-bounds all costs:
\[
	\sum_{j=1}^{k} w_j \Phi_{i,j} \geq c_i \quad \forall i \in \{1, \ldots, m\}
\]
This is equivalently written as $\sum_j -w_j \Phi_{i,j} + c_i \leq 0$, defining the residual $r_i = c_i - f(\vec{x}_i)$ which must be non-positive for soundness.
To enforce monotonicity and conservativeness across data points, the formulation employs binary variables and the Big-M technique.
For each pair of consecutive data points (when sorted by input size), constraints are added to ensure that residuals are either monotone non-increasing or exactly zero:
\[
	r_i - r_{i-1} \leq M \cdot z_i + (c_{i-1} - c_i)
\]
where $M$ is a sufficiently large constant (set to $100 \times \max(y)$).
This constraint is satisfied either when $z_i = 1$ (allowing the residual difference), or when $r_i \leq r_{i-1}$ adjusted for cost differences.
Additionally, if $r_i \neq 0$, then:
\[
	-M \cdot (1 - z_i) \leq r_i \leq M \cdot (1 - z_i)
\]
forcing $z_i = 0$ when the residual is non-zero, and allowing $z_i = 1$ only when the bound exactly matches the cost.
For univariate bounds (degree 1), only consecutive points are compared; for higher degrees, all pairwise comparisons are enforced to maintain global monotonicity.
The objective function minimizes the sum of coefficients weighted by the number of features:
\[
	\text{minimize} \quad \sum_{j=1}^{k} w_j \cdot k
\]
This encourages sparse, simple bounds while satisfying all soundness constraints.
When a previous feasible solution with residual sum $S$ exists, an additional constraint $\sum_i r_i \geq S$ is added to find strictly better solutions in subsequent iterations.
\paragraph{Solver Invocation and Result Selection.}
For each template, the MILP problem is passed to the CBC solver via PuLP's \code{PULP\_CBC\_CMD} interface with suppressed output.
If the solver returns an \code{Optimal} status, the coefficient values are extracted, the bound is evaluated on the input data, and residuals are computed.
The solution is marked as \emph{sound} if all residuals are non-positive (within a tolerance of $\epsilon = 10^{-9}$).
The engine tracks the best solution across all templates and degrees, selecting the one with the smallest residual sum.
After exploring all templates, the best bound is formatted into a human-readable string using the exponent descriptors and basis function string formatters.
\paragraph{Output and Visualization.}
The final bound, along with metadata (degree, residual sum, soundness flag, number of data points), is written to a JSON file for consumption by the web interface.
For single-variable bounds, 2D plots are generated showing costs, the fitted curve, and optionally a reference bound from RaML.
For two-variable bounds, 3D surface plots visualise the bound over a mesh grid with scatter points for costs.
Additionally, the MILP problem is exported to an LP file for inspection, and the data points are saved in both Markdown table format and a LaTeX-compatible coordinate format for inclusion in figures.
\begin{table}[t]
	\centering
	\caption{Overview of the objective functions $\mathcal{O}$ considered in this work, categorised by bound degree and number of variables. The coefficients $a_i$ denote real-valued parameters to be optimised.}
	\label{tab:objective-functions}
	\small
	\begin{tabularx}{\textwidth}{clX}
		\toprule
		\textbf{Degree} & \textbf{Variables} & \textbf{Objective Function}                                                                      \\
		\midrule
		1               & $x$                & $a_0 + a_1 x$                                                                                    \\
		                & $x$                & $a_0 + a_1 \log(x)$                                                                              \\
		                & $x$                & $a_0 + a_1 x + a_2 \log(x)$                                                                      \\
		                & $x, y$             & $a_0 + a_1 x + a_2 y$                                                                            \\
		                & $x, y$             & $a_0 + a_1 x + a_2 y + a_3 \log(x) + a_4 \log(y)$                                                \\
		                & $x, y, z$          & $a_0 + a_1 x + a_2 y + a_3 z$                                                                    \\
		\midrule
		2               & $x$                & $a_0 + a_1 x + a_2 x^2$                                                                          \\
		                & $x$                & $a_0 + a_1 \log(x) + a_2 x \log(x)$                                                              \\
		                & $x, y$             & $a_0 + a_1 x + a_2 y + a_3 x^2 + a_4 y^2 + a_5 x y$                                              \\
		                & $x, y$             & $a_0 + a_1 x + a_2 y + a_3 x y$                                                                  \\
		                & $x, y, z$          & $a_0 + a_1 x + a_2 y + a_3 z + a_4 x^2 + a_5 y^2 + a_6 z^2 + a_7 xy + a_8 xz + a_9 yz$           \\
		\midrule
		3               & $x$                & $a_0 + a_1 x + a_2 x^2 + a_3 x^3$                                                                \\
		                & $x, y$             & $a_0 + a_1 x + a_2 y + a_3 x^2 + a_4 y^2 + a_5 x^3 + a_6 y^3 + a_7 xy + a_8 x^2 y + a_9 x y^2$   \\
		                & $x, y, z$          & $a_0 + a_1 x + a_2 y + a_3 z + a_4 x^2 + a_5 y^2 + a_6 z^2 + a_7 x^3 + a_8 y^3 + a_9 z^3 +$      \\&& \quad $a_{10} x y^2 + a_{11} x^2 y + a_{12} x z^2 + a_{13} x^2 z + a_{14} y z^2 + a_{15} y^2 z$ \\
		\midrule
		4               & $x$                & $a_0 + a_1 x + a_2 x^2 + a_3 x^3 + a_4 x^4$                                                      \\
		                & $x, y$             & $a_0 + a_1 x + a_2 y + a_3 x^2 + a_4 y^2 + a_5 x^3 + a_6 y^3 + a_7 x^4 + a_8 y^4 + a_{13} x y^3$ \\
		                & $x, y, z$          & $a_0 + a_1 x + a_2 y + a_3 z + a_4 x^2 + a_5 y^2 + a_6 z^2 + a_7 x^3 + a_8 y^3 + a_9 z^3 +$      \\&& \quad $a_{10} x^4 + a_{11} y^4 + a_{12} z^4 + a_{13} x y^3 + a_{14} x^3 y + a_{15} x z^3 + a_{16} x^3 z + a_{17} y z^3 + a_{18} y^3 z$ \\
		\bottomrule
	\end{tabularx}
\end{table}
\section{Proofs}
\subsection{Correctness of Symbolic Execution}
\label{appendix:proofs}
\begin{proof}[Proof sketch for Lemma~\ref{lem:soundness}]
	We prove that any complete symbolic execution path with satisfiable constraints corresponds to a concrete execution with the same cost.
	\paragraph{Initial state.} The symbolic execution begins with initial environment $\rho_{\text{init}}$ constructed from the input specification $\argspecificationList{}$ by applying $\buildValue{\argspecification{x_i}{\theta_i}{\phi_i}}$ for each argument. Under a satisfying model $M \models \assertionSet{}$, concretisation $\gamma_M$ maps these symbolic values to concrete inputs that respect the specified type and constraint structure.
	\paragraph{Step-by-step correspondence.} We prove by structural induction on the symbolic execution derivation that each symbolic transition corresponds to a valid concrete transition under concretisation $\gamma_M$.
	Operations like \textsc{Let}, \textsc{Call}, \textsc{Fun}, \textsc{Cons}, \textsc{Tick}, and \textsc{Raise} behave identically in symbolic and concrete execution.
	Their concretisation is immediate: $\gamma_M$ commutes with these operations.
	For arithmetic operations like \textsc{BinOp}, \textsc{CmpOp}, \textsc{BoolOp}, \textsc{Neg}, etc., soundness relies on the correctness of Z3's expression evaluation. If $v_1 \oplus v_2$ yields a Z3 expression in symbolic execution, then $\gamma_M(v_1 \oplus v_2) = \gamma_M(v_1) \oplus_{\text{concrete}} \gamma_M(v_2)$ by Z3's semantic guarantees. This ensures that symbolic arithmetic expressions evaluate to the same concrete results as direct computation.
	For \textsc{IfElse}, three cases arise:
	\begin{itemize}
		\item \emph{Definite true/false}: If $\text{simplify}(v_c)$ yields $\top$ or $\bot$, the path is deterministic and concretisation trivially preserves the branch taken.
		\item \emph{Symbolic branching}: When $\text{simplify}(v_c) \notin \{\top, \bot\}$, the engine explores both branches, adding $v_c$ (respectively $\neg v_c$) to path constraints. For any model $M \models \assertionSet{} \cup \{v_c\}$, we have $\gamma_M(v_c) = \texttt{true}$, ensuring the concrete execution follows the same branch. The same holds for the else-branch with $\neg v_c$.
	\end{itemize}
	For \textsc{Split}, the analysis considers all structural possibilities:
	\begin{itemize}
		\item \emph{Concrete lists}: Rules \textsc{Split-Nil} and \textsc{Split-Cons} match the concrete list structure, and concretisation preserves the decomposition.
		\item \emph{Symbolic lists}: Rules \textsc{Split-Sym-Empty} and \textsc{Split-Sym-Nonempty} add length constraints $\text{len}(v_l) = 0$ or $\text{len}(v_l) > 0$ to $\assertionSet{}$. For any satisfying model $M$, the length constraint determines the concrete list structure, ensuring the correct branch is taken. Fresh symbolic variables $v_h, v_t$ introduced in the non-empty case are instantiated by $M$ to match the concrete head and tail.
	\end{itemize}
	The same reasoning applies to \textsc{Unpack} for algebraic data types.
	\paragraph{Cost preservation.} The cost accumulator $\mathcal{C}$ is only modified by the \textsc{Tick} rule, which directly corresponds to tick operations in the source program. Since each symbolic step preserves cost through concretisation, the final symbolic cost $\mathcal{C}$ equals the concrete execution cost.
	By induction on the complete execution trace from the initial call, every symbolic path with satisfiable constraints $M \models \assertionSet{}$ corresponds to a concrete execution under concretisation $\gamma_M$, producing the same final value and cost.
\end{proof}
\begin{proof}[Proof sketch for Lemma~\ref{lem:completeness}]
	The key insight is that structural descriptors induce \emph{exhaustive case splitting} within bounds.
	\paragraph{Initial environment construction.}
	Given concrete inputs $I_1, \ldots, I_k$ and the input specification $\argspecificationList{}$, the initial environment is constructed as $\rho_{\text{input}} = \{x_i \mapsto \buildValue{\argspecification{x_i}{\theta_i}{\phi_i}} \mid i = 1, \ldots, k\}$, where each template $T_{\theta_i}(I_i)$ precisely captures the structure of the corresponding concrete input $I_i$. This ensures that symbolic execution starts with symbolic values that directly represent the concrete inputs.
	\paragraph{Exhaustive path exploration.}
	When symbolic execution encounters a descriptor during pattern matching (e.g., via \textsc{Split-Var-Empty}, \textsc{Split-Var-Nonempty}), it systematically explores all structural possibilities:
	\begin{itemize}
		\item For a list descriptor with length bounds $[\texttt{minSize}, \texttt{maxSize}]$, the split rules create one path for each possible length $\ell \in \{\texttt{minSize}, \ldots, \texttt{maxSize}\}$.
		\item Each path adds appropriate length constraints to $\assertionSet{}$ (e.g., $\text{len}(v) = \ell$) and creates fresh symbolic variables for list elements.
		\item For any concrete input tuple $(I_1, \ldots, I_k)$ where each $I_i$ has structure within the specified bounds, there exists a path whose constraints are satisfiable, and the SMT solver produces a model $M$ that maps symbolic variables to the concrete values.
	\end{itemize}
	\paragraph{Completeness and cost preservation.}
	By construction, branching on descriptors partitions the space of all concrete input tuples within bounds into disjoint cases, one per symbolic path. Since every concrete input tuple satisfies exactly one case's constraints, every concrete execution from the initial call to $\IRfunName{main}$ is represented. The cost equality $\mathcal{C} = c$ follows from Lemma~\ref{lem:soundness}: once we identify the symbolic path representing the concrete inputs, concretisation shows that the symbolic and concrete costs match exactly.
\end{proof}
Together, Lemma~\ref{lem:soundness} and Lemma~\ref{lem:completeness} establish that symbolic execution with bounded descriptors provides a \emph{sound and complete} characterization of program behaviour within those bounds. For complexity analysis, this means the maximum cost across all symbolic paths is precisely the worst-case cost for all inputs within the specified size bounds.
\paragraph{Cost Soundness and Upper Bounds.}
The cost accumulator $\mathcal{C}$ is central to deriving complexity bounds. We establish that costs are correctly tracked through symbolic execution.
For any symbolic execution path with cost $\mathcal{C}$ and satisfiable constraints $\assertionSet{}$ with model $M \models \assertionSet{}$, the corresponding concrete execution (via concretisation $\gamma_M$) has cost exactly $\mathcal{C}$.
This follows directly from Lemma~\ref{lem:soundness} and the fact that the \textsc{Tick} rule is the only operation that modifies $\mathcal{C}$, and it corresponds exactly to tick operations in the source program. Path splitting does not affect tick accumulation---each path independently accumulates its own cost.
Consequently, when we report $\text{max}\{\mathcal{C} \mid \text{path with cost } \mathcal{C}\}$, we obtain the maximum cost over all concrete executions within bounds, yielding a \emph{sound upper bound} on worst-case complexity.
\subsection{Soundness of the Estimated Bound}
\label{section:soundness-bound}
\begin{proof}[Proof sketch for the Soundness]
	$f(\vec{x})$ is constructed as a linear combination of basis functions $\{f_1(\vec{x}), f_2(\vec{x}), \ldots, f_m(\vec{x})\}$:
	\[
		f(\vec{x}) = \sum_{j=1}^{m} a_j \cdot f_j(\vec{x})
	\]
	where the coefficients $a_j$ are the decision variables of the MILP.
	The MILP formulation imposes the following constraints on these coefficients:
	\paragraph{(1) Non-negativity constraints:} For each coefficient $a_j$, we have $a_j \geq 0$. This ensures that the bound is non-decreasing (assuming non-negative basis functions), which is appropriate for worst-case complexity analysis.
	\paragraph{(2) Upper-bounding constraints:} For each observation $(\vec{x}_i, c_i) \in \mathcal{M}$, the MILP includes the constraint:
	\[
		f(\vec{x}_i) \geq c_i
	\]
	This is the fundamental requirement that the bound must upper-bound every observed cost.
	In the standard minimization formulation used by LP solvers, this is equivalently expressed as:
	\[
		-f(\vec{x}_i) \leq -c_i
	\]
	\paragraph{(3) Residual monotonicity constraints (optional):} If residual constraints are enabled, for costs ordered lexicographically as $l_{\mathcal{M}} = [(\vec{x}_1, c_1), \ldots, (\vec{x}_N, c_N)]$ with $i < j$, we have:
	\[
		(f(\vec{x}_j) - c_j) \geq (f(\vec{x}_i) - c_i) \quad \text{or} \quad f(\vec{x}_j) = c_j
	\]
	encoded via binary variables $b_i \in \{0, 1\}$ using the big-M method. These constraints do not weaken the upper-bounding property; they only restrict \emph{how} the bound can deviate from the costs, promoting better generalisation.
	When the MILP solver finds a \emph{feasible solution}, it has determined values for all decision variables (coefficients $a_j$ and binary variables $b_i$) that simultaneously satisfy all constraints. In particular, constraint (2) is satisfied, which directly implies:
	\[
		f(\vec{x}_i) \geq c_i \quad \text{for all } i \in \{1, \ldots, N\}
	\]
	Therefore, if the MILP admits a feasible solution, the resulting bound $f(\vec{x})$ is guaranteed to be an upper bound on all costs in $\mathcal{M}$.
\end{proof}
Theorem~\ref{thm:bound-soundness} establishes soundness \emph{with respect to the costs} $\mathcal{M}$. To obtain soundness with respect to the \emph{actual program behaviour}, we must ensure that $\mathcal{M}$ accurately captures the worst-case costs for all relevant input shapes. This is addressed by combining the soundness of MILP with the correctness of the simulation (Lemma~\ref{lem:soundness}) in the empirical soundness result (Corollary~\ref{cor:overall-soundness}).
\section{Dyade Example} \label{appendix:dyade}
To illustrate the frame-based symbolic execution, consider evaluating \IRfunName{dyade}.
For the \texttt{dyade} example, the set of global function definitons is as follows:
{
\[
	\begin{aligned}
		\mathit{extractFuns}( & \raml{let}~mult~n~l~\raml{=}~e_m\raml{;;}~\raml{let}~dyade~l1~l2~\raml{=}~e_d)
		\\
		                      & = \{
		(\texttt{mult}, [n, l], \trcore{e_m}),\;
		(\texttt{dyade}, [l1, l2], \trcore{e_d})
		\}
	\end{aligned}
\]
}
We may execute the function for any input sizes and constraints; in this example, we illustrate it with symbolic lists of length $n=2$ and $m=1$.
The input argument specification for this example is formally given as $\mathcal{I} = (\mathcal{P}, \IRfunName{dyade}, [\iota_1, \iota_2], \emptyset)$, with:
\begin{align*}
	\iota_1
	 & = \argspecification{\texttt{l1}}{T_1}{\emptyset}
	\\
	 & \quad \text{where }
	T_1 =
	(\DTList,\{
	\texttt{minLength}=0,\,
	\texttt{maxLength}=2,\,
	\texttt{itemTemplate}=T_{\text{int}}
	\})
	\\[0.5ex]
	\iota_2
	 & = \argspecification{\texttt{l2}}{T_2}{\emptyset}
	\\
	 & \quad \text{where }
	T_2 =
	(\DTList,\{
	\texttt{minLength}=0,\,
	\texttt{maxLength}=1,\,
	\texttt{itemTemplate}=T_{\text{int}}
	\})
\end{align*}
and $T_{\text{int}} = (\DTInt, \emptyset)$ denotes the template for unbounded integer values.
This specification declares that the first argument \texttt{l1} is a list of integers with maximum length 2, and the second argument \texttt{l2} is a list of integers with maximum length 1.
All function definitions from the program are transformed into closures and placed in the initial environment. Based on the input specification (which specifies a main function name and arguments with their types and optional constraints), the tool constructs \emph{variable-length list descriptors}. For lists with $\texttt{maxLength}=2$ and $\texttt{maxLength}=1$, these are \emph{structural descriptors} (not symbolic Z3 values) that lazily generate elements during execution. For this example:
\begin{center}
	\[
		\rho_{\text{init}} = \begin{cases}
			\IRfunName{mult} \mapsto \closure{\IRlist{\IRvar{x}, \IRvar{l}}}{\tau_{\text{mult}}}{\rho_{\text{init}}}       \\
			\IRfunName{dyade} \mapsto \closure{\IRlist{\IRvar{l_1}, \IRvar{l_2}}}{\tau_{\text{dyade}}}{\rho_{\text{init}}} \\
		\end{cases}
	\]
	\[
		\rho_{\text{args}} = \begin{cases}
			\IRvar{l_1} \mapsto \descriptiveList{0}{\texttt{"l1"}}{0}{2}{T_{\text{int}}} \\
			\IRvar{l_2} \mapsto \descriptiveList{0}{\texttt{"l2"}}{0}{1}{T_{\text{int}}} \\
		\end{cases}
	\]
\end{center}
where $\tau_{\text{dyade}}$ is the body of \texttt{dyade} shown above, and $\rho_{\text{init}}$ is the initial environment at the time of function definition (containing references to the \raml{mult} and \raml{dyade} functions).
The initial environment $\rho_0$ is then constructed by combining $\rho_{\text{init}}$ with $\rho_{\text{args}}$, which includes the list descriptors for the input arguments.
The tool constructs a synthetic call to the main function:
\[
	\exframe{\IRcall{\IRfunName{dyade}}{[\IRvar{l_1}, \IRvar{l_2}]}}{\sigma_0}{\idcontinuation{}}
\]
where $\sigma_0 = \exstate{\rho_0}{\assertionSet{}_0}{\mathcal{S}_0}{0}$, $\assertionSet{}_0$ contains any initial assertions from the input specification, and $\idcontinuation{}$ is the top-level continuation that collects results.
This synthetic call frame is added to the work-list, and execution begins.
\begin{enumerate}
	\item \textbf{Function call (Rules \textsc{Call-Step} and \textsc{Call-Eval}).}
	      The arguments are evaluated left-to-right using \textsc{Call-Step}:
	      \begin{itemize}
		      \item First, rule \textsc{Var} evaluates $\IRvar{l_1}$ to $\rho_0(\texttt{l1}) = \descriptiveList{0}{\text{"l1"}}{0}{2}{T_{\text{int}}}$. Rule \textsc{Call-Step} then updates the call with this list descriptor as the first argument.
		      \item Next, rule \textsc{Var} evaluates $\IRvar{l_2}$ to $\rho_0(\texttt{l2}) = \descriptiveList{0}{\text{"l2"}}{0}{1}{T_{\text{int}}}$. Rule \textsc{Call-Step} updates the call with this list descriptor as the second argument.
		      \item Looking up \IRfunName{dyade} in $\rho_0$ yields: $\rho_0(\IRfunName{dyade}) = \closure{[\IRvar{l_1}, \IRvar{l_2}]}{\tau_{\text{dyade}}}{\rho_{\text{init}}}$ where $\tau_{\text{dyade}} = \IRsplit{l1}{[]}{x}{xs}{\tau_{\text{body}}}$ is the body of \IRfunName{dyade}.
	      \end{itemize}
	      \begin{figure}[t]
		      \scriptsize
		      \begin{inferencetree}[\scriptsize]
			      \begin{prooftree}
				      \AxiomC{}
				      \RightLabel{\textsc{Cons-Val}}
				      \UnaryInfC{
					      $\exframe{\IRcons{[\symbolic{v_1} \cdot \symbolic{w_1}]}{ [[\symbolic{v_2} \cdot \symbolic{w_1}]]}}{\exstate{\rho_2}{\emptyset}{\emptyCache{}}{3}}{\idcontinuation{}} \leadsto
						      ([\symbolic{v_1} \cdot \symbolic{w_1}] ::  [[\symbolic{v_2} \cdot \symbolic{w_1}]], \emptyset, 3)$
				      }
			      \end{prooftree}
			      \begin{prooftree}
				      \AxiomC{$\exframe{\IRcall{\IRfunName{dyade}}{[\IRvar{xs}, \IRvar{l_2}]}}{\exstate{\rho_2}{\emptyset}{\emptyCache{}}{2}}{\kappa_t} \leadsto \kappa_t(v_t, \emptyset, 3)$}
				      \RightLabel{\textsc{Cons-Step}}
				      \UnaryInfC{
					      $\exframe{\IRcons{[\symbolic{v_1} \cdot \symbolic{w_1}]}{\IRcall{\IRfunName{dyade}}{[\IRvar{xs}, \IRvar{l_2}]}}}{\exstate{\rho_2}{\emptyset}{\emptyCache{}}{2}}{\idcontinuation{}} \leadsto$
					      $\exframe{\IRcons{[\symbolic{v_1} \cdot \symbolic{w_1}]}{v_t}}{\exstate{\rho_2}{\emptyset}{\emptyCache{}}{3}}{\idcontinuation{}}$
				      }
			      \end{prooftree}
			      \begin{itemize}[label=]
				      \item $\rho_2(\IRvar{xs}) = \descriptor{d'_{l1}} = \descriptiveList{1}{0}{2}{T_{\text{int}}}{\emptyset}$ (tail descriptor at offset 1)
				      \item Recursive call to $\IRcall{\IRfunName{dyade}}{[\descriptor{d'_{l1}}, \descriptor{d_{l2}}]}$ splits $\descriptor{d'_{l1}}$ to create $\symbolic{v_2} = \texttt{Z3.Int.const("l1.2")}$
				      \item Result: $v_t = [[\symbolic{v_2} \cdot \symbolic{w_1}]]$ with cost $\mathcal{C} = 3$
			      \end{itemize}
			      \begin{prooftree}
				      \AxiomC{$\exframe{\IRcall{\IRfunName{mult}}{[\IRvar{x}, \IRvar{l_2}]}}{\exstate{\rho_2}{\emptyset}{\emptyCache{}}{1}}{\kappa_h} \leadsto \kappa_h(v_h, \emptyset, 2)$}
				      \RightLabel{\textsc{Cons-Step}}
				      \UnaryInfC{
					      $\exframe{\IRcons{\IRcall{\IRfunName{mult}}{[\IRvar{x}, \IRvar{l_2}]}}{\tau_t}}{\exstate{\rho_2}{\emptyset}{\emptyCache{}}{1}}{\idcontinuation{}} \leadsto$
					      $\exframe{\IRcons{v_h}{\tau_t}}{\exstate{\rho_2}{\emptyset}{\emptyCache{}}{2}}{\idcontinuation{}}$
				      }
			      \end{prooftree}
			      \begin{itemize}[label=]
				      \item $\rho_2(\IRvar{x}) = \symbolic{v_1}$, $\rho_2(\IRvar{l_2}) = \descriptor{d_{l2}} = \descriptiveList{0}{0}{1}{T_{\text{int}}}{\emptyset}$
				      \item Calling $\IRcall{\IRfunName{mult}}{[\symbolic{v_1}, \descriptor{d_{l2}}]}$ splits $\descriptor{d_{l2}}$ to create $\symbolic{w_1} = \texttt{Z3.Int.const("l2.1")}$
				      \item Result: $v_h = [\symbolic{v_1} \cdot \symbolic{w_1}]$ with cost $\mathcal{C} = 2$
				      \item $\tau_t = \IRcall{\IRfunName{dyade}}{[\IRvar{xs}, \IRvar{l_2}]}$ (tail to be evaluated next)
			      \end{itemize}
			      \begin{prooftree}
				      \AxiomC{$1 \in \mathbb{R}$}
				      \RightLabel{\textsc{Tick}}
				      \UnaryInfC{
					      $\exframe{\IRtick{1}{\tau_{\text{cons}}}}{\exstate{\rho_2}{\emptyset}{\emptyCache{}}{0}}{\idcontinuation{}} \leadsto$
					      $\exframe{\tau_{\text{cons}}}{\exstate{\rho_2}{\emptyset}{\emptyCache{}}{1}}{\idcontinuation{}}$
				      }
			      \end{prooftree}
			      \begin{itemize}[label=]
				      \item Rule \textsc{Tick} increments cost: $\mathcal{C} = 0 \leadsto \mathcal{C}' = 1$
				      \item $\tau_{\text{cons}} = \IRcons{\IRcall{mult}{[\IRvar{x}, \IRvar{l_2}]}}{\IRcall{dyade}{[\IRvar{xs}, \IRvar{l_2}]}}$
			      \end{itemize}
			      \begin{prooftree}
				      \AxiomC{$\rho_1 = \rho_{\text{init}}[\IRvar{l_1} \mapsto \descriptor{d_{l1}}, \IRvar{l_2} \mapsto \descriptor{d_{l2}}]$}
				      \AxiomC{$\tau_{\text{dyade}} = \IRsplit{\IRvar{l_1}}{\IRlist{}}{\IRvar{x}}{\IRvar{xs}}{\tau_{\text{body}}}$}
				      \RightLabel{\textsc{Split-Var-Nonempty}}
				      \BinaryInfC{
					      $\exframe{\tau_{\text{dyade}}}{\exstate{\rho_1}{\emptyset}{\emptyCache{}}{0}}{\idcontinuation{}} \leadsto$
					      $\exframe{\tau_{\text{body}}}{\exstate{\rho_2}{\emptyset}{\emptyCache{}}{0}}{\idcontinuation{}}$
				      }
			      \end{prooftree}
			      \begin{itemize}[label=]
				      \item $\descriptor{d'_{l1}} = \descriptiveList{1}{\text{"l1"}}{0}{2}{T_{\text{int}}}$ (updated descriptor at offset 1)
				      \item $\symbolic{v_1} = \texttt{Z3.Int.const("l1.1")}$ (fresh symbolic value for position 0)
				      \item $\rho_2 = \rho_1[\IRvar{x} \mapsto \symbolic{v_1}, \IRvar{xs} \mapsto \descriptor{d'_{l1}}]$
				      \item $\tau_{\text{body}} = \IRtick{1}{\IRcons{\IRcall{mult}{[\IRvar{x}, \IRvar{l_2}]}}{\IRcall{dyade}{[\IRvar{xs}, \IRvar{l_2}]}}}$
			      \end{itemize}
			      \begin{prooftree}
				      \AxiomC{$\rho_0(\IRfunName{dyade}) = \closure{[\IRvar{l_1}, \IRvar{l_2}]}{\tau_{\text{dyade}}}{\rho_{\text{init}}}$}
				      \RightLabel{\textsc{Call-Eval}}
				      \UnaryInfC{
					      $\exframe{\IRcall{dyade}{[\descriptor{d_{l1}}, \descriptor{d_{l2}}]}}{\exstate{\rho_0}{\emptyset}{\emptyCache{}}{0}}{\idcontinuation{}} \leadsto$
					      $\exframe{\tau_{\text{dyade}}}{\exstate{\rho_{\text{init}}[\IRvar{l_1} \mapsto \descriptor{d_{l1}}, \IRvar{l_2} \mapsto \descriptor{d_{l2}}]}{\emptyset}{\emptyCache{}}{0}}{\idcontinuation{}}$
				      }
			      \end{prooftree}
			      \begin{itemize}[label=]
				      \item $\tau_{\text{dyade}} = \IRsplit{\IRvar{l_1}}{\IRlist{}}{\IRvar{x}}{\IRvar{xs}}{\tau_{\text{body}}}$
			      \end{itemize}
			      \begin{prooftree}
				      \AxiomC{\IRvar{l_2} $\in dom(\rho_0)$}
				      \RightLabel{\textsc{Var}}
				      \UnaryInfC{$\exframe{\IRvar{l_2}}{\exstate{\rho_0}{\emptyset}{\emptyCache{}}{0}}{\idcontinuation{}} \leadsto (\descriptor{d_{l2}}, \emptyset, \emptyCache{}, 0)$}
				      \RightLabel{\textsc{Call-Step}}
				      \UnaryInfC{$\exframe{\IRcall{dyade}{[\descriptor{d_{l1}}, \IRvar{l_2}]}}{\exstate{\rho_0}{\emptyset}{\emptyCache{}}{0}}{\idcontinuation{}} \leadsto \exframe{\IRcall{dyade}{[\descriptor{d_{l1}}, \descriptor{d_{l2}}]}}{\exstate{\rho_0}{\emptyset}{\emptyCache{}}{0}}{\emptyCache{}}{\idcontinuation{}}$}
			      \end{prooftree}
			      \begin{itemize}[label=]
				      \item $\descriptor{d_{l2}} = \descriptiveList{0}{\text{"l2"}}{0}{1}{T_{\text{int}}}$ (list descriptor at offset 0)
			      \end{itemize}
			      \begin{prooftree}
				      \AxiomC{\IRvar{l_1} $\in dom(\rho_0)$}
				      \RightLabel{\textsc{Var}}
				      \UnaryInfC{$\exframe{\IRvar{l_1}}{\exstate{\rho_0}{\emptyset}{\emptyCache{}}{0}}{\idcontinuation{}} \leadsto (\descriptor{d_{l1}}, \emptyset, 0)$}
				      \RightLabel{\textsc{Call-Step}}
				      \UnaryInfC{$\exframe{\IRcall{dyade}{[\IRvar{l_1}, \IRvar{l_2}]}}{\exstate{\rho_0}{\emptyset}{\emptyCache{}}{0}}{\idcontinuation{}}
						      \leadsto
						      \exframe{\IRcall{dyade}{[\descriptor{d_{l1}}, \IRvar{l_2}]}}{\exstate{\rho_0}{\emptyset}{\emptyCache{}}{0}}{\idcontinuation{}}$}
			      \end{prooftree}
			      \begin{itemize}[label=]
				      \item $\descriptor{d_{l1}} = \descriptiveList{0}{\text{"l1"}}{0}{2}{T_{\text{int}}}$ (list descriptor at offset 0)
				      \item $S_{\emptyset} = ([],[],[])$ (empty cache)
			      \end{itemize}
		      \end{inferencetree}
		      \caption{An exemplary symbolic execution of \program{Dyade}.}
		      \label{fig:dyade-execution-call}
	      \end{figure}
	      Once all arguments are values (in this case, list descriptors), rule \textsc{Call-Eval} applies:
	      \[
		      \exframe{\IRcall{\IRfunName{dyade}}{[\descriptor{d_{l1}}, \descriptor{d_{l2}}]}}{\exstate{\rho_0}{\emptyset}{([],[],[])}{0}}{\idcontinuation{}} \leadsto
		      \exframe{\tau_{\text{dyade}}}{\exstate{\rho_1}{\emptyset}{([],[],[])}{0}}{\idcontinuation{}}
	      \]
	      where \descriptor{d_{l1}} and \descriptor{d_{l2}} are the list descriptors, and $\rho_1 = \rho_{\text{init}}[\IRvar{l_1} \mapsto \descriptor{d_{l1}}, \IRvar{l_2} \mapsto \descriptor{d_{l2}}]$.
	\item \textbf{Split operation (Rule \textsc{Split-Var-Nonempty}).}
	      \[
		      \tau_{\text{dyade}} = \exframe{\IRsplit{\IRvar{l_1}}{[]}{x}{xs}{\tau_{\text{body}}}}{\exstate{\rho_1}{\emptyset}{([],[],[])}{0}}{\idcontinuation{}}
	      \]
	      Evaluating $\tau_v = \IRvar{l_1}$ yields the list descriptor $\descriptor{d_{l1}} = \descriptiveList{0}{\text{"l1"}}{0}{2}{T_{\text{int}}}$ at offset 0. Since $0 < M = 2$, rule \textsc{Split-Var-Nonempty} is applicable. It creates a fresh symbolic Z3 value $\symbolic{v_1} = \texttt{Z3.Int.const("l1.1")}$ for position 0 and updates the descriptor to $\descriptor{d'_{l1}} = \descriptiveList{1}{\text{"l1"}}{0}{2}{T_{\text{int}}}$ with offset incremented to 1.
	      Applying rule \textsc{Split-Var-Nonempty}:
	      \[
		      \leadsto \exframe{\tau_{\text{body}}}{\exstate{\rho_2}{\emptyset}{([],[],[])}{0}}{\idcontinuation{}}
	      \]
	      where $\rho_2 = \rho_1[\IRvar{x} \mapsto \symbolic{v_1}, \IRvar{xs} \mapsto \descriptor{d'_{l1}}]$ and $\tau_{\text{body}} = \IRtick{1}{\tau_{\text{cons}}}$. The pattern variable $x$ is bound to the symbolic value $\symbolic{v_1}$, and $xs$ is bound to the tail descriptor at offset 1.
	\item \textbf{Tick (Rule \textsc{Tick}).}
	      \[
		      \exframe{\IRtick{1}{\tau_{\text{cons}}}}{\exstate{\rho_2}{\emptyset}{([],[],[])}{0}}{\idcontinuation{}}
	      \]
	      Rule \textsc{Tick} increments the cost counter from $\mathcal{C} = 0$ to $\mathcal{C}' = 1$ and returns $\tau_{\text{cons}}$. The path constraints remain unchanged: $\assertionSet{}' = \emptyset$.
	\item \textbf{Cons operation (Rule \textsc{Cons}).}
	      \[
		      \exframe{\IRcons{\IRcall{\IRfunName{mult}}{[\IRvar{x}, \IRvar{l_2}]}}{\IRcall{\IRfunName{dyade}}{[\IRvar{xs}, \IRvar{l_2}]}}}{\exstate{\rho_2}{\emptyset}{([],[],[])}{1}}{\idcontinuation{}}
	      \]
	      Rule \textsc{Cons} evaluates the head and tail expressions sequentially:
	      \begin{itemize}\small
		      \item First, evaluate $\tau_h = \IRcall{\IRfunName{mult}}{[\IRvar{x}, \IRvar{l_2}]}$ with state $\exstate{\rho_2}{\emptyset}{([],[],[])}{1}$.
		      \item The arguments evaluate to $\rho_2(\IRvar{x}) = \symbolic{v_1}$ and $\rho_2(\IRvar{l_2}) = \descriptor{d_{l2}} = \descriptiveList{0}{\text{"l2"}}{0}{1}{T_{\text{int}}}$.
		      \item Calling \IRfunName{mult} with these arguments splits $\descriptor{d_{l2}}$, creating $\symbolic{w_1} = \texttt{Z3.Int.const("l2.1")}$, and computes $[\symbolic{v_1} \cdot \symbolic{w_1}]$ with accumulated cost $\mathcal{C}_h = 2$.
	      \end{itemize}
	\item \textbf{Tail evaluation (recursive call).}
	      Continuing the \textsc{Cons} rule, evaluate the tail term:
	      \[
		      \langle \IRcall{\IRfunName{dyade}}{[\IRvar{xs}, \IRvar{l_2}]}, \exstate{\rho_2}{\emptyset}{([],[],[])}{2}, \kappa_t \rangle
	      \]
	      The arguments evaluate to:
	      \begin{itemize}
		      \item $\rho_2(\IRvar{xs}) = \descriptor{d'_{l1}} = \descriptiveList{1}{\text{"l1"}}{0}{2}{T_{\text{int}}}$ (tail descriptor at offset 1)
		      \item $\rho_2(\IRvar{l_2}) = \descriptor{d_{l2}} = \descriptiveList{0}{\text{"l2"}}{0}{1}{T_{\text{int}}}$
	      \end{itemize}
	      This recursively applies steps 2--4: split $\descriptor{d'_{l1}}$ creates $\symbolic{v_2} = \texttt{Z3.Int.const("l1.2")}$, tick increments cost to 3, call \texttt{mult} produces $[\symbolic{v_2} \cdot \symbolic{w_1}]$, then call \texttt{dyade} on the empty list (offset 2 $\geq$ maxLength 2) returns $[]$. The final tail value is $v_t = [[\symbolic{v_2} \cdot \symbolic{w_1}]]$ with cost $\mathcal{C}_t = 3$.
	\item \textbf{Final result.}
	      Completing the \textsc{Cons} rule application:
	      \begin{itemize}
		      \item Head value: $v_h = [\symbolic{v_1} \cdot \symbolic{w_1}]$ with cost $\mathcal{C}_h = 2$
		      \item Tail value: $v_t = [[\symbolic{v_2} \cdot \symbolic{w_1}]]$ with cost $\mathcal{C}_t = 3$
		      \item Path constraints: $\assertionSet{} = \emptyset$ (no constraints added)
	      \end{itemize}
	      The \textsc{Cons} rule constructs the final list $v_h :: v_t = [[\symbolic{v_1} \cdot \symbolic{w_1}], [\symbolic{v_2} \cdot \symbolic{w_1}]]$ and invokes the continuation
	      $\idcontinuation{}([[\symbolic{v_1} \cdot \symbolic{w_1}], [\symbolic{v_2} \cdot \symbolic{w_1}]], \emptyset, 3)$.
\end{enumerate}
\section{Experiments}
\label{appendix:experiments}
This section provides detailed analyses of selected benchmark programs from our experimental evaluation (Section~\ref{section:experimental-evaluation}).
Recall that we evaluated our approach on $46$ benchmarks from the \texttt{RaML} suite and Pham et al.~\cite{Pham:etal:2025},
comparing our inferred bounds against \texttt{RaML}'s static bounds and the known asymptotic complexities.
Here we present visualizations and detailed case studies for representative examples,
particularly focusing on programs where \texttt{RaML} reports infeasibility or where our method achieves notable results.
We have included plots for selected programs from Table~\ref{tab:experiments2} in Figures~\ref{fig:plots-examples-1}, Figure~\ref{fig:plots-examples-1} and~\ref{fig:plots-examples-3}. Each plot shows the costs (black dots), the static analysis prediction (blue surface), and the dynamic analysis prediction (red surface).
\paragraph{fib\_memo.}
The \program{FibMemo} program
implements memoised Fibonacci computation using a state monad to maintain a cache of previously computed values. The naive recursive Fibonacci has exponential complexity $O(2^n)$ due to redundant recomputation. Memoization eliminates this redundancy: when computing $\raml{fib n}$, the function \texttt{memoM} first checks whether the result is already in the cache (via \texttt{find}, which performs a linear search through the cache list), and if not, recursively computes it and inserts it into the cache. For input $n$, at most $n$ distinct Fibonacci values need to be computed (from 0 to $n$). Each computation requires checking the cache, which grows linearly with the number of cached values, resulting in $\sum_{i=0}^{n-1} i = \frac{n(n-1)}{2}$ cache lookups, yielding $\Theta(n^2)$ complexity. Our approach successfully infers the bound $2 + 11.95n + 1.35n^2$, correctly capturing the quadratic behaviour. The linear coefficient (11.95) accounts for the base cost of monadic operations (\raml{bind}, \raml{evalState}, \raml{fibM} itself), while the quadratic coefficient (1.35) captures the cache lookup overhead. \texttt{RaML} fails on this example because the state monad pattern, involving higher-order functions (\raml{bind}, \raml{liftM2}), closures capturing mutable state (functionally via state passing), and the complex interplay between memoization logic and recursive calls, exceeds the expressive power of its type-based resource analysis.
\begin{figure}[b]
	\centering
	\begin{subfigure}[b]{0.3\textwidth}
		\centering
		\includegraphics[width=\textwidth]{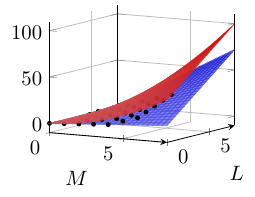}
		\subcaption{\program{Concat}.}\label{fig:concat-example}
	\end{subfigure}
	\hfill
	\begin{subfigure}[b]{0.3\textwidth}
		\centering
		\includegraphics[width=\textwidth]{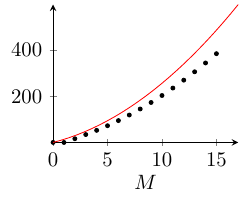}
		\subcaption{\program{FibMemo}.}\label{fig:fib_memo-example}
	\end{subfigure}
	\hfill
	\begin{subfigure}[b]{0.3\textwidth}
		\centering
		\includegraphics[width=\textwidth]{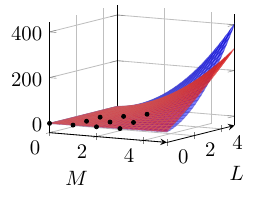}
		\subcaption{\program{Flatten}.}\label{fig:flatten-example}
	\end{subfigure}
	\vspace{1em}
	\begin{subfigure}[b]{0.3\textwidth}
		\centering
		\includegraphics[width=\textwidth]{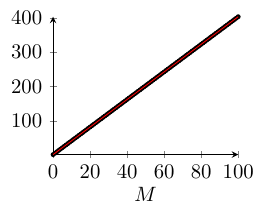}
		\subcaption{\program{FoldSum}.}\label{fig:foldsum-example}
	\end{subfigure}
	\hfill
	\begin{subfigure}[b]{0.3\textwidth}
		\centering
		\includegraphics[width=\textwidth]{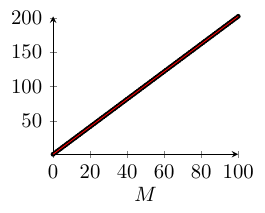}
		\subcaption{\program{Id}.}\label{fig:id-example}
	\end{subfigure}
	\hfill
	\begin{subfigure}[b]{0.3\textwidth}
		\centering
		\includegraphics[width=\textwidth]{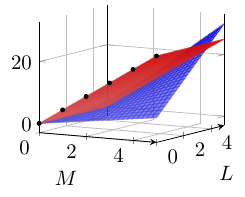}
		\subcaption{\program{MapAppend}.}\label{fig:map_append-example}
	\end{subfigure}
	\vspace{1em}
	\begin{subfigure}[b]{0.3\textwidth}
		\centering
		\includegraphics[width=\textwidth]{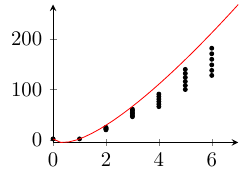}
		\subcaption{\program{MergesortDC}.}\label{fig:mergesort_dc-example}
	\end{subfigure}
	\hfill
	\begin{subfigure}[b]{0.3\textwidth}
		\centering
		\includegraphics[width=\textwidth]{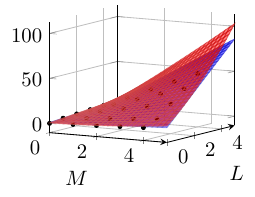}
		\subcaption{\program{LCS}.}\label{fig:lcs-example}
	\end{subfigure}
	\hfill
	\begin{subfigure}[b]{0.3\textwidth}
		\centering
		\includegraphics[width=\textwidth]{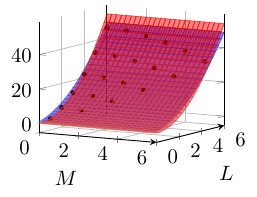}
		\subcaption{\program{QuickSelect}.}\label{fig:quickselect-example}
	\end{subfigure}
	\hfill
	\begin{subfigure}[b]{0.3\textwidth}
		\centering
		\includegraphics[width=\textwidth]{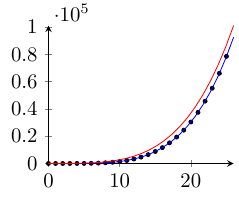}
		\subcaption{\program{Tuples}.}\label{fig:tuples-example}
	\end{subfigure}
	\vspace{1em}
	\begin{subfigure}[b]{0.3\textwidth}
		\centering
		\includegraphics[width=\textwidth]{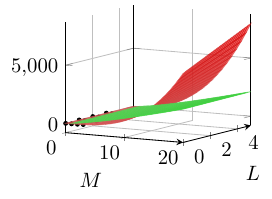}
		\subcaption{\program{AVLTree}.}\label{fig:avl-example}
	\end{subfigure}
	\hfill
	\begin{subfigure}[b]{0.3\textwidth}
		\centering
		\includegraphics[width=\textwidth]{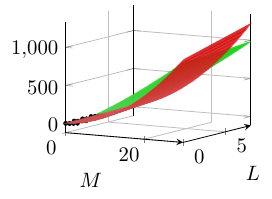}
		\subcaption{\program{RedBlackTree}.}\label{fig:redblacktree-example}
	\end{subfigure}
	\caption{Plots for selected programs from Table~\ref{tab:experiments2}, for which soundness can be validated through visual inspection. The black dots represent costs, the blue surface the static analysis prediction (where available), the red surface our dynamic analysis prediction, which upper-bounds all observations, and the green surface the ground truth bound taken from the artifact of \cite{Pham:etal:2025}.}
	\label{fig:plots-examples-1}
\end{figure}
\paragraph{flatten.}
The \program{Flatten} program flattens a binary tree (where each node contains a list of integers) into a single list, then sorts it using insertion sort. The tree is recursively traversed: for each \raml{Node(l, t1, t2)}, the node's list \raml{l} is appended to the concatenation of the flattened subtrees. The \raml{append} function executes 1 tick per element in its first argument. For a tree with $M$ nodes, each containing a list of maximum length $M_1$, the flattening produces a list of total length at most $M M_1$. The insertion sort then processes this list, which in the worst case requires $\sum_{i=1}^{M M_1-1} i = \frac{(M M_1)(M M_1-1)}{2}$ comparisons, yielding $\Theta(M^2 M_1^2)$ complexity. \texttt{RaML} correctly infers the bound $M_1 M + 0.5M_1 M^2 + 0.5M_1^2 M^2 + M$, where the dominant $0.5M_1^2 M^2$ term captures the insertion sort cost and the $M_1 M$ term represents the flattening overhead. Our approach infers $0.94M + 0.06M^2 + 0.33M M_1^2 + 2M^2 M_1$, which exhibits a different term ordering. The dominant term $2M^2 M_1$ suggests that within the explored input bounds, the worst-case quadratic sorting behaviour in the list length dimension ($M_1^2$) was not as pronounced as the quadratic behaviour in the tree size dimension ($M^2$). The $0.06M^2$ and $0.33M M_1^2$ terms still capture higher-order interactions. This structural difference arises from the particular distribution of tree shapes and list lengths in the symbolic execution, where the relationship between $M$ and the total flattened list length may have varied in ways that emphasised different complexity factors compared to the theoretical worst case.
\paragraph{foldsum.}
The \program{FoldSum} program with the main function \raml{foldsum l}, computes the sum of a list of integers using higher-order function composition. The function first maps \raml{plus} over the input list \raml{l}, creating a list of partially applied addition functions. It then uses \raml{foldr} with the composition combinator \raml{comp} to compose these functions into a single function, which is finally applied to 0. The execution proceeds in three phases: (i) \raml{map plus l} traverses the list, executing $M+1$ ticks for a list of length $M$ (including the base case), (ii) \raml{foldr comp id} traverses the resulting list of functions, executing $M+1$ ticks from the recursion plus $M$ ticks from the \raml{comp} calls, and (iii) applying the composed function to 0 triggers \raml{id} (1 tick) plus $M$ calls to the partially applied \raml{plus} functions, contributing $M$ additional ticks. Summing these contributions yields $(M+1) + (2M+1) + (M+1) = 4M + 3$ ticks. Our approach correctly infers the bound $3 + 4M$, matching the observed linear complexity. \texttt{RaML} fails on this example because the pattern of mapping a binary function to create partially applied closures, composing them via higher-order fold, and then triggering the accumulated computation by applying the result to an argument involves multiple levels of higher-order abstraction that exceed the expressive power of its type-based resource analysis. The analysis would need to track potential costs through unapplied function values and reason about the eventual composition structure, which is beyond the capabilities of the amortised analysis framework.
\paragraph{id.} The \program{Id} with the main function \raml{iterid n} takes a natural number $n$ (encoded as \raml{Zero} or \raml{S(n)}) and iteratively applies the successor function $n$ times to zero, computing the identity function. The auxiliary function \raml{iter f g x} executes 1 tick per recursive call, invoking \raml{f} on the result of \raml{iter f g x'} until reaching the base case \raml{Zero}, where it returns \raml{g}.
For \raml{iterid n}, the execution proceeds as follows: \raml{iter} is called $n+1$ times (including the base case), contributing $n+1$ ticks. Each of the $n$ recursive calls triggers \raml{compS}, which executes 1 tick per invocation, contributing $n$ ticks. Finally, the identity function \raml{id} is called once, contributing 1 tick. The exact cost is thus:
$
	\text{cost}(\texttt{iterid}, n) = (n+1) + n + 1 = 2n + 2.
$
Our approach infers the bound $2 + 2n$, matching the exact cost.
\texttt{RaML} cannot analyse this example, as the complex interplay of higher-order functions 
and function composition exceeds the expressive power of its type system.
\paragraph{map\_append.}
The \program{MapAppend} program takes a list of pairs (tuples) and applies \raml{append} to each pair, concatenating the two lists in each tuple. The \raml{append} function executes 1 tick for each element in the first list plus 1 tick for the initial pattern match. For a dataset containing $M$ pairs, where each first list has maximum length $M_1$, the total cost is $M(M_1+1) = M + MM_1$ ticks: each of the $M$ pairs incurs $(M_1+1)$ ticks during the append operation. \texttt{RaML} correctly infers this general bound as $M_1 M + M$. Our approach, however, infers the bound $5M$, which corresponds to the worst-case scenario within the bounded symbolic execution where $M_1$ was constrained to a maximum of 4. Substituting $M_1 = 4$ into the general formula yields $M(4+1) = 5M$, which precisely matches our prediction. This example illustrates a key difference between static analysis and our empirical approach: while \texttt{RaML} produces a parametric bound valid for any $M_1$, our method produces a concrete bound for the specific input size bounds explored during symbolic execution. For the analysed input range ($M_1 \leq 4$), our bound is exact; for larger inputs, the bound would need to be re-inferred with appropriately increased size bounds.
\paragraph{mergesort\_dc.}
The \program{MergesortDC} program implements with \raml{mergesort zs} merge sort using a generic higher-order \raml{divideAndConquer} combinator. Unlike the direct recursive mergesort implementation discussed earlier, this version abstracts the divide-and-conquer pattern: it takes as parameters (i) a predicate \raml{isDivisible} to determine when to subdivide a problem, (ii) a \raml{solve} function for base cases, (iii) a \raml{divide} function to split the problem, and (iv) a \raml{combine} function to merge sub-solutions. For sorting, \raml{divide} splits the input list into two halves by computing \raml{halve (length ys)}, then using \raml{take} and \raml{drop} to partition the list. At each recursion level, the cost includes: (i) computing the list length ($M$ ticks), (ii) halving the length ($\approx M/2$ ticks), (iii) taking the first half ($\approx 3M/2$ ticks, as each \raml{take} call invokes both \raml{head} and \raml{tail}), (iv) dropping the first half ($\approx M$ ticks), and (v) merging the sorted halves ($\approx M$ ticks). Summing these contributions yields approximately $6M$ ticks per level. Since the recursion has $\log_2(M)$ levels, the total cost is $\Theta(M \log M)$. Our approach infers the bound $2 + 13.57M \log(M)$, where the constant 13.57 captures the combined overhead of all operations across recursion levels, and the additive constant 2 accounts for the initial setup ticks in \raml{mergesort} and \raml{divideAndConquer}. \texttt{RaML} fails on this example because the generic higher-order divide-and-conquer pattern, parameterised by four function arguments and involving complex control flow through \raml{map} over the recursively generated sub-problems, exceeds the capabilities of its type-based amortised analysis framework.
\begin{figure}[t]
	\centering
	\begin{subfigure}[b]{0.3\textwidth}
		\centering
		\includegraphics[width=\textwidth]{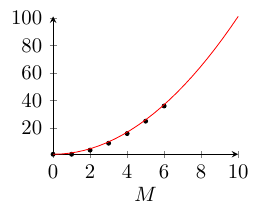}
		\subcaption{\program{BubbleSort}.}\label{fig:bubblesort-example}
	\end{subfigure}
	\hfill
	\begin{subfigure}[b]{0.3\textwidth}
		\centering
		\includegraphics[width=\textwidth]{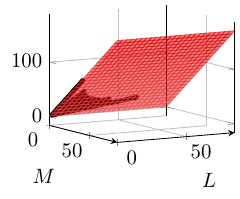}
		\subcaption{\program{DivBySub}.}\label{fig:div_by_sub-example}
	\end{subfigure}
	\hfill
	\begin{subfigure}[b]{0.3\textwidth}
		\centering
		\includegraphics[width=\textwidth]{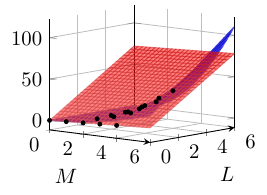}
		\subcaption{\program{Duplicates}.}\label{fig:duplicates-example}
	\end{subfigure}
	\vspace{1em}
	\begin{subfigure}[b]{0.3\textwidth}
		\centering
		\includegraphics[width=\textwidth]{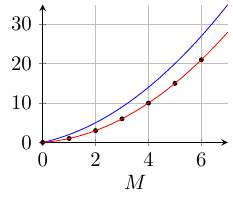}
		\subcaption{\program{MinSort}.}\label{fig:minsort-example}
	\end{subfigure}
	\hfill
	\begin{subfigure}[b]{0.3\textwidth}
		\centering
		\includegraphics[width=\textwidth]{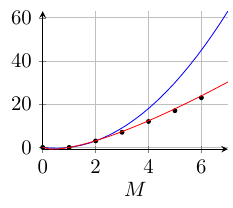}
		\subcaption{\program{MergeSort}.}\label{fig:mergesort-example}
	\end{subfigure}
	\hfill
	\begin{subfigure}[b]{0.3\textwidth}
		\centering
		\includegraphics[width=\textwidth]{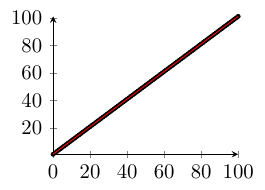}
		\subcaption{\program{RevDL}.}\label{fig:rev-dl-example}
	\end{subfigure}
	\vspace{1em}
	\begin{subfigure}[b]{0.3\textwidth}
		\centering
		\includegraphics[width=\textwidth]{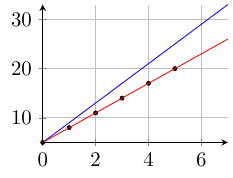}
		\subcaption{\program{MSS}.}\label{fig:mss-example}
	\end{subfigure}
	\hfill
	\begin{subfigure}[b]{0.3\textwidth}
		\centering
		\includegraphics[width=\textwidth]{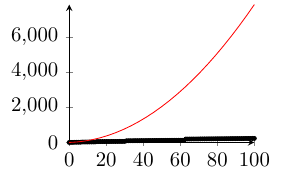}
		\subcaption{\program{Round}.}\label{fig:round-example}
	\end{subfigure}
	\vspace{1em}
	\begin{subfigure}[b]{0.3\textwidth}
		\centering
		\includegraphics[width=\textwidth]{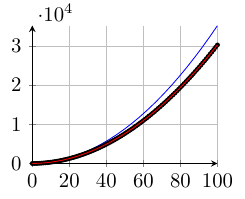}
		\subcaption{\program{Queue}.}\label{fig:queue-example}
	\end{subfigure}
	\hfill
	\begin{subfigure}[b]{0.3\textwidth}
		\centering
		\includegraphics[width=\textwidth]{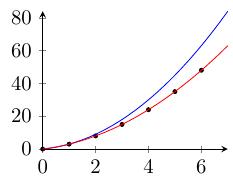}
		\subcaption{\program{SplitAndSort}.}\label{fig:splitandsort-example}
	\end{subfigure}
	\caption{Measured vs.\ predicted costs for selected programs from Table~\ref{tab:experiments2}. The black dots represent costs, the blue surface the static analysis prediction, and the red surface the dynamic analysis prediction. For these programs, we either have tighter ground truth bounds or \texttt{RaML} reported infeasible.}
	\label{fig:plots-examples-3}
\end{figure}
\paragraph{div\_by\_sub.}
The main function \raml{div x y} computes the integer division \(x \mathbin{//} y\) by repeated subtraction of \(y\) from \(x\). The function only terminates for inputs where $x \bmod y = 0$, otherwise raising an exception.
To determine the exact cost, observe that \raml{sub x y} executes 1 tick initially, then recursively decrements both arguments $y$ times, resulting in a cost of $y + 1$ ticks. The main function \raml{div x y} executes 1 tick initially, then calls \raml{sub (x-1) (y-1)} (costing $y$ ticks), and recursively calls \raml{div} on the result. For $x = k \cdot y$, this process repeats $k = x \mathbin{//} y$ times, yielding:
$
	\text{cost}(\raml{div}, x, y) = 1 + \sum_{i=0}^{k-1}(1 + y) = 1 + k(1 + y) = 1 + \frac{x}{y} + x
$
Since $y \geq 1$ implies $x \mathbin{//} y \leq x$, we have $1 + x + (x \mathbin{//} y) \leq 1 + 2x$. Our tool correctly infers the upper bound $1 + 2x$, which is tight for $y = 1$.
\texttt{RaML} fails on this example, likely due to the combination of nested recursion and exception handling, which challenges type-based resource analysis.
\paragraph{rev-dl.} The \program{RevDL} function reverses a list using continuation-passing style to achieve linear time complexity without explicit accumulator. The helper function \raml{walk} recursively traverses the input list: for each \raml{Cons(x, xs')}, it composes the result of \texttt{walk(xs')} with a closure that prepends \raml{x}. The \raml{comp} function, which performs function composition, executes 1 tick per call.
For a list of length $n$, \raml{walk} is invoked $n+1$ times (once per element plus the base case \raml{Nil}), but only the \raml{comp} calls contribute to the cost, occurring $n$ times (once per \raml{Cons} node). The main function \raml{rev} executes 1 tick initially. The exact cost $\text{cost}(\raml{rev-dl}, n)$ is $1 + n$.
Our approach correctly infers the bound $1 + n$.
\texttt{RaML} fails on this example because the continuation-passing style with nested closures and higher-order composition (\raml{walk} returns a function, not a list) makes it difficult for the type-based analysis to track resource usage through the composed continuations.
\paragraph{round.}
The \program{Round} function exhibits a recursive structure reminiscent of divide-and-conquer algorithms. It alternately halves the list tail (via \raml{half}, which extracts every other element) and doubles the result (via \raml{double}, which duplicates each element), creating a tree-like recursion pattern. For an input list of length $n$, \raml{round} recursively processes a list of approximately $\lfloor(n-1)/2\rfloor$ elements, leading to $\Theta(\log n)$ recursion depth. At each recursive level, \raml{half} and \raml{double} contribute linear work proportional to the list length at that level.
The exact cost depends on the binary representation of $n$: each recursion step halves the remaining length, and the costs accumulate based on the path through the recursion tree. The total cost can be expressed as $4n - 3 \cdot \text{popcount}(n+1) + 4 \cdot 2^{\lfloor \log_2(n+1) \rfloor}$, where $\text{popcount}$ counts the number of 1-bits in the binary representation of $n+1$.
Our tool predicts $3 + 3.25n + 0.75n^2$ (with an increased penalty constant to handle the irregular growth pattern), capturing the dominant linear and quadratic components. The quadratic term conservatively upper-bounds the logarithmic and bit-pattern-dependent fluctuations.
\texttt{RaML} cannot analyse this example, as the irregular recursion pattern and the interplay between halving and doubling operations challenge its type-based resource analysis framework.
\paragraph{splitandsort.}
The \program{SplitAndSort} program partitions a list of key-value pairs into sublists by key, then sorts each sublist using quicksort. The function \raml{split} recursively inserts each element into the appropriate sublist (calling \raml{insert}, which may traverse the entire accumulated result), while \raml{sortAll} applies quicksort to each sublist. In the worst case where all elements share the same key, \raml{split} creates a single sublist, and the cost combines: (i) the insertion cost of $\sum_{i=0}^{n-1} i = \frac{n(n-1)}{2}$ ticks from building the list, and (ii) quicksort's worst-case cost of $\Theta(n^2)$ on the resulting single sublist of $n$ elements.
Our tool infers the bound $4 + 25.78n + 5n^2 + 0.22n^3$ which is sound, but not tight.

\bibliographystyle{splncs04}
\bibliography{references}
\end{document}

%% file: include/dyade/data.tex
\addplot3[
    only marks,
    mark=*,
    mark size=1.5pt,
    color=red,
] coordinates {
(0.0,0.0,0.0)
(1.0,0.0,1.0)
(1.0,1.0,2.0)
(1.0,2.0,3.0)
(1.0,3.0,4.0)
(1.0,4.0,5.0)
(1.0,5.0,6.0)
(1.0,6.0,7.0)
(1.0,7.0,8.0)
(1.0,8.0,9.0)
(1.0,9.0,10.0)
(1.0,10.0,11.0)
(1.0,11.0,12.0)
(1.0,12.0,13.0)
(1.0,13.0,14.0)
(1.0,14.0,15.0)
(1.0,15.0,16.0)
(2.0,0.0,2.0)
(2.0,1.0,4.0)
(2.0,2.0,6.0)
(2.0,3.0,8.0)
(2.0,4.0,10.0)
(2.0,5.0,12.0)
(2.0,6.0,14.0)
(2.0,7.0,16.0)
(2.0,8.0,18.0)
(2.0,9.0,20.0)
(2.0,10.0,22.0)
(2.0,11.0,24.0)
(2.0,12.0,26.0)
(2.0,13.0,28.0)
(2.0,14.0,30.0)
(2.0,15.0,32.0)
(3.0,0.0,3.0)
(3.0,1.0,6.0)
(3.0,2.0,9.0)
(3.0,3.0,12.0)
(3.0,4.0,15.0)
(3.0,5.0,18.0)
(3.0,6.0,21.0)
(3.0,7.0,24.0)
(3.0,8.0,27.0)
(3.0,9.0,30.0)
(3.0,10.0,33.0)
(3.0,11.0,36.0)
(3.0,12.0,39.0)
(3.0,13.0,42.0)
(3.0,14.0,45.0)
(3.0,15.0,48.0)
(4.0,0.0,4.0)
(4.0,1.0,8.0)
(4.0,2.0,12.0)
(4.0,3.0,16.0)
(4.0,4.0,20.0)
(4.0,5.0,24.0)
(4.0,6.0,28.0)
(4.0,7.0,32.0)
(4.0,8.0,36.0)
(4.0,9.0,40.0)
(4.0,10.0,44.0)
(4.0,11.0,48.0)
(4.0,12.0,52.0)
(4.0,13.0,56.0)
(4.0,14.0,60.0)
(4.0,15.0,64.0)
(5.0,0.0,5.0)
(5.0,1.0,10.0)
(5.0,2.0,15.0)
(5.0,3.0,20.0)
(5.0,4.0,25.0)
(5.0,5.0,30.0)
(5.0,6.0,35.0)
(5.0,7.0,40.0)
(5.0,8.0,45.0)
(5.0,9.0,50.0)
(5.0,10.0,55.0)
(5.0,11.0,60.0)
(5.0,12.0,65.0)
(5.0,13.0,70.0)
(5.0,14.0,75.0)
(5.0,15.0,80.0)
(6.0,0.0,6.0)
(6.0,1.0,12.0)
(6.0,2.0,18.0)
(6.0,3.0,24.0)
(6.0,4.0,30.0)
(6.0,5.0,36.0)
(6.0,6.0,42.0)
(6.0,7.0,48.0)
(6.0,8.0,54.0)
(6.0,9.0,60.0)
(6.0,10.0,66.0)
(6.0,11.0,72.0)
(6.0,12.0,78.0)
(6.0,13.0,84.0)
(6.0,14.0,90.0)
(6.0,15.0,96.0)
(7.0,0.0,7.0)
(7.0,1.0,14.0)
(7.0,2.0,21.0)
(7.0,3.0,28.0)
(7.0,4.0,35.0)
(7.0,5.0,42.0)
(7.0,6.0,49.0)
(7.0,7.0,56.0)
(7.0,8.0,63.0)
(7.0,9.0,70.0)
(7.0,10.0,77.0)
(7.0,11.0,84.0)
(7.0,12.0,91.0)
(7.0,13.0,98.0)
(7.0,14.0,105.0)
(7.0,15.0,112.0)
(8.0,0.0,8.0)
(8.0,1.0,16.0)
(8.0,2.0,24.0)
(8.0,3.0,32.0)
(8.0,4.0,40.0)
(8.0,5.0,48.0)
(8.0,6.0,56.0)
(8.0,7.0,64.0)
(8.0,8.0,72.0)
(8.0,9.0,80.0)
(8.0,10.0,88.0)
(8.0,11.0,96.0)
(8.0,12.0,104.0)
(8.0,13.0,112.0)
(8.0,14.0,120.0)
(8.0,15.0,128.0)
(9.0,0.0,9.0)
(9.0,1.0,18.0)
(9.0,2.0,27.0)
(9.0,3.0,36.0)
(9.0,4.0,45.0)
(9.0,5.0,54.0)
(9.0,6.0,63.0)
(9.0,7.0,72.0)
(9.0,8.0,81.0)
(9.0,9.0,90.0)
(9.0,10.0,99.0)
(9.0,11.0,108.0)
(9.0,12.0,117.0)
(9.0,13.0,126.0)
(9.0,14.0,135.0)
(9.0,15.0,144.0)
(10.0,0.0,10.0)
(10.0,1.0,20.0)
(10.0,2.0,30.0)
(10.0,3.0,40.0)
(10.0,4.0,50.0)
(10.0,5.0,60.0)
(10.0,6.0,70.0)
(10.0,7.0,80.0)
(10.0,8.0,90.0)
(10.0,9.0,100.0)
(10.0,10.0,110.0)
(10.0,11.0,120.0)
(10.0,12.0,130.0)
(10.0,13.0,140.0)
(10.0,14.0,150.0)
(10.0,15.0,160.0)
(11.0,0.0,11.0)
(11.0,1.0,22.0)
(11.0,2.0,33.0)
(11.0,3.0,44.0)
(11.0,4.0,55.0)
(11.0,5.0,66.0)
(11.0,6.0,77.0)
(11.0,7.0,88.0)
(11.0,8.0,99.0)
(11.0,9.0,110.0)
(11.0,10.0,121.0)
(11.0,11.0,132.0)
(11.0,12.0,143.0)
(11.0,13.0,154.0)
(11.0,14.0,165.0)
(11.0,15.0,176.0)
(12.0,0.0,12.0)
(12.0,1.0,24.0)
(12.0,2.0,36.0)
(12.0,3.0,48.0)
(12.0,4.0,60.0)
(12.0,5.0,72.0)
(12.0,6.0,84.0)
(12.0,7.0,96.0)
(12.0,8.0,108.0)
(12.0,9.0,120.0)
(12.0,10.0,132.0)
(12.0,11.0,144.0)
(12.0,12.0,156.0)
(12.0,13.0,168.0)
(12.0,14.0,180.0)
(12.0,15.0,192.0)
(13.0,0.0,13.0)
(13.0,1.0,26.0)
(13.0,2.0,39.0)
(13.0,3.0,52.0)
(13.0,4.0,65.0)
(13.0,5.0,78.0)
(13.0,6.0,91.0)
(13.0,7.0,104.0)
(13.0,8.0,117.0)
(13.0,9.0,130.0)
(13.0,10.0,143.0)
(13.0,11.0,156.0)
(13.0,12.0,169.0)
(13.0,13.0,182.0)
(13.0,14.0,195.0)
(13.0,15.0,208.0)
(14.0,0.0,14.0)
(14.0,1.0,28.0)
(14.0,2.0,42.0)
(14.0,3.0,56.0)
(14.0,4.0,70.0)
(14.0,5.0,84.0)
(14.0,6.0,98.0)
(14.0,7.0,112.0)
(14.0,8.0,126.0)
(14.0,9.0,140.0)
(14.0,10.0,154.0)
(14.0,11.0,168.0)
(14.0,12.0,182.0)
(14.0,13.0,196.0)
(14.0,14.0,210.0)
(14.0,15.0,224.0)
(15.0,0.0,15.0)
(15.0,1.0,30.0)
(15.0,2.0,45.0)
(15.0,3.0,60.0)
(15.0,4.0,75.0)
(15.0,5.0,90.0)
(15.0,6.0,105.0)
(15.0,7.0,120.0)
(15.0,8.0,135.0)
(15.0,9.0,150.0)
(15.0,10.0,165.0)
(15.0,11.0,180.0)
(15.0,12.0,195.0)
(15.0,13.0,210.0)
(15.0,14.0,225.0)
(15.0,15.0,240.0)
};

%% file: include/residuals/1.tex
\begin{tikzpicture}[scale=0.8]
    \begin{axis}[
        view={20}{22},          
        xlabel={input size},
        ylabel={cost},
        width=6cm,
        height=4.5cm,
        axis x line=bottom,
        axis y line=left,
    ]

    \addplot[
        only marks,
        mark=*,
        mark size=1.2pt,
        color=black
    ] coordinates {
    (0.0,1.0)
    (1.0,1.0)
    (2.0,4.0)
    (3.0,9.0)
    (4.0,16.0)
    (5.0,25.0)
    (6.0,36.0)
    };

    \addplot[
        red,
        domain=0:8,
        domain y=0:8
    ]
    {1 + 6*x};
    \end{axis}
\end{tikzpicture}

%% file: include/residuals/2.tex
\begin{tikzpicture}[scale=0.8]
    \begin{axis}[
        view={20}{22},          
        xlabel={input size},
        ylabel={cost},
        width=6cm,
        height=4.5cm,
        axis x line=bottom,
        axis y line=left,
    ]

    \addplot[
        only marks,
        mark=*,
        mark size=1.2pt,
        color=black
    ] coordinates {
    (0.0,1.0)
    (1.0,1.0)
    (2.0,4.0)
    (3.0,9.0)
    (4.0,16.0)
    (5.0,25.0)
    (6.0,36.0)
    };

    \addplot[
        red,
        domain=0:7,
        domain y=0:7
    ]
    {1 + 1.2*x^2};
    \end{axis}
\end{tikzpicture}